\documentclass[a4paper,reqno,12pt]{article}

\usepackage[titletoc,title]{appendix}
\usepackage{authblk}
\usepackage{hyperref}
\usepackage{mathtools}
\usepackage{dsfont}
\usepackage{bbm}
\usepackage{amsmath}
\usepackage{mathrsfs}
\usepackage{amsthm}
\usepackage{xcolor}
\usepackage{amssymb}
\usepackage{color}
\usepackage{amsfonts}
\usepackage{amsthm}
\usepackage{amscd}
\usepackage[latin2]{inputenc}
\usepackage{t1enc}
\usepackage[mathscr]{eucal}
\usepackage{indentfirst}
\usepackage{graphicx}
\usepackage{graphics}
\usepackage{pict2e}
\usepackage{epic}

\numberwithin{equation}{section}

\newtheorem{theorem}{Theorem}[section]

\newtheorem{lemma}[theorem]{Lemma}

\newcommand{\diag}{\mathrm{diag}}
\newtheorem{thm}{Theorem}[section]

\newtheorem{rem}[thm]{Remark}

\newcommand\cB{{\mathcal B}}

\newcommand\cD{{\mathcal D}}
\newcommand\cE{{\mathcal E}}

\newcommand\cH{{\mathcal H}}

\newcommand\cL{{\mathcal L}}

\newcommand\cS{{\mathcal S}}

\newcommand\cP{{\mathcal P}}

\newcommand\cU{{\mathcal U}}
\newcommand\cR{{\mathcal R}}

\newcommand\bC{{\mathbb C}}
\newcommand\bE{{\mathbb E}}
\newcommand\bN{{\mathbb N}}
\newcommand\bP{{\mathbb P}}
\newcommand\bR{{\mathbb R}}
\newcommand\bI{{\mathbb I}}

\newcommand\bZ{{\mathbb Z}}

\newcommand\fr{{\mathfrak{r}}}
\newcommand\fp{{\mathfrak{p}}}
\newcommand\fe{{\mathfrak{e}}}
\newcommand\fz{{\mathfrak{z}}}
\newcommand\frH{{\mathfrak{H}}}
\newcommand\frf{{\mathfrak{f}}}
\newcommand\frP{{\mathfrak{P}}}
\newcommand\frR{{\mathfrak{R}}}

\newcommand\ssr{{\mathsf{r}}}
\newcommand\ssp{{\mathsf{p}}}
\newcommand\sse{{\mathsf{e}}}

\newcommand\ssb{{\mathsf{b}}}
\newcommand\ssa{{\mathsf{a}}}

\newcommand\rmr{{\mathrm{r}}}
\newcommand\rmp{{\mathrm{p}}}

\newcommand\rmb{{\mathrm{b}}}

\newcommand\hmr{{\tilde{r}^{\bullet}}}
\newcommand\hmp{{\tilde{p}^{\bullet}}}

\newcommand\roo{{\varrho_n^{\bar{r},\bar{p},\bar{b}}}}
\newcommand\Roo{{\varrho^n_{\bar{r},\bar{p},\beta}}}

\newcommand{\expval}[1]{\langle#1\rangle}
\newcommand{\ket}[1]{|#1\rangle}
\newcommand{\braket}[2]{\langle#1|#2\rangle}
\newcommand{\Tr}[1]{\mathrm{Tr}(#1)}
\newcommand{\ketbra}[1]{|#1\rangle\langle#1|}

\begin{document}

\title{Derivation of Euler equations from quantum and classical microscopic dynamics}

\author[1]{Amirali Hannani}
\author[2]{Fran\c cois Huveneers}
\affil[1]{
ICERM, Brown University, 121 S Main St, Providence, RI 02903, United States }
\affil[2]{
Ceremade,
UMR-CNRS 7534, 
Universit\'e Paris Dauphine, 
PSL Research University, 
%Place du Mar\'echal de Lattre de Tassigny, 
75775 Paris cedex 16, 
France
}

\date{\today}
    \maketitle    
    
\begin{abstract}
We derive Euler equations from a Hamiltonian microscopic dynamics. 
The microscopic system is a one-dimensional disordered harmonic chain, and the dynamics is either quantum or classical. 
This chain is an Anderson insulator with a symmetry protected mode: 
Thermal fluctuations are frozen while the low modes ensure the transport of elongation, momentum and mechanical energy, that evolve according to Euler equations in an hyperbolic scaling limit. 
In this paper, we strengthen considerably the results in \cite{BHO} and \cite{A20}, where we established a limit in mean starting from a local Gibbs state: 
We now control the second moment of the fluctuations around the average, yielding a limit in probability, and we enlarge the class of admissible initial states.
% in the quantum set-up. Our control yields a limit in probability. 
\end{abstract}

\section{Introduction and results}
The deep understanding of transport properties of materials is of fundamental physical interest. 
It is also an utmost challenging task from a mathematical perspective, and only some specific cases have been successfully investigated so far. 
Here we focus on the derivation of \emph{Euler} system of conservation laws in an appropriate scaling limit, 
the so-called hyperbolic space-time rescaling, where both space and time are rescaled by a common factor $n\to\infty$.  

The usual route to Euler equations rests on some form of \emph{ergodicity}: 
The replacement of time averages by a spatial averages allows to close the equations of motion for the conserved fields. 
These fields usually involve at least the number of particles, conveniently replaced by elongation in 1-d systems, momentum and energy. 
However, while the common belief is that ``typical'' anharmonic dynamics are ergodic in the large volume limit, 
it has so far not been possible to establish it at a mathematical level of rigor. 
Instead, most rigorous results rely on the introduction of some noise that mimics the effect of anharmonic interactions and ensure ergodicity
\cite{stefano93,stefano14,stefano16}.
In addition, recently, 
the macroscopic evolution of conserved charges was derived rigorously for a class of 1-d quantum systems in the hyperbolic scaling \cite{doyon}, 
but the number of charges remains unknown as a consequence of the lack of knowledge on the ergodicity. 

In \cite{BHO} and \cite{A20}, we have introduced a completely new approach to this question, that we further investigate in this paper. 
{Indeed, we consider an \emph{integrable}, hence strongly non-ergodic, dynamics and show the validity of Euler equations in an hyperbolic scaling limit.
While the integrability and some particular features of our model render possible a full mathematical treatment, 
our primary objective is not to exploit these specific aspects.
Instead, we want to convey the message that ergodicity needs not to be the crucial ingredient for the emergence of the hydrodynamic behavior described by Euler equations.
In our work, the macroscopic description is derived from the fact that mechanical and thermal modes, corresponding respectively to low and high modes, evolve on strikingly different time scales.
We live to further investigations the possibility of implementing our strategy in the presence of anharmonic interactions.}

\paragraph{Model.}
We study the disordered harmonic chain described by the Hamiltonian 
\begin{equation*}%\label{hamiltonian}
	H \; = \; \sum_{x=1}^n \frac{p_x^2}{2m_x} + g \frac{(q_{x+1} - q_x)^2}{2} 
\end{equation*}
where $n$ is the number of particles, with canonical coordinates $(p_x,q_x)$ 
(boundary conditions will be specified later on and play no crucial role). 
We will further fix energy units so that $g=1$. 
The equations of motion are given by 
\begin{equation*}%\label{hamilton}
	m_x \frac{d q_x}{dt} \; = \; p_x, \qquad
	\frac{d p_x}{dt} \; = \; (\Delta q)_x
\end{equation*}
where $\Delta$ denotes the lattice Laplacian. 
We consider both classical and quantum dynamics, 
and these equations of motion should thus be interpreted either as Hamilton equations, or as Heisenberg equations. 

Crucially, the masses $m_x$ are assumed to be random and independent. 
This guarantees that the spectrum of $M^{-1}\Delta$ is \emph{Anderson localized}, with a localization length diverging as one approaches the ground state, 
where $M$ is the diagonal matrix with $M_{x,x}=m_x$, cf.\@ \cite{Joel1,Theo,Dhar,FA}. 
The ground state, explicitly given by $\psi(x)=1/\sqrt{N}$ for $x=1,\dots,n$, is an example of a symmetry protected mode, here issuing from the conservations of momentum. 
Anderson localization guarantees that thermal fluctuations are frozen on all time scales, 
so that the evolution of the mechanical energy is entirely slaved to the evolution of momentum and elongation (cf.~below).

The system has three obvious locally conserved quantities (that may be destroyed at the boundary depending on the choice of boundary conditions): 
energy, momentum and elongation: 
\begin{equation} \label{conserved fields}
	H_n \; = \; \sum_{x=1}^n \frac12 \left(\frac{p_x^2}{m_x}+r_x^2 \right) 
	\; =: \; \sum_{x=1}^n e_x,\qquad
	P_n\; = \; \sum_{x=1}^n p_x,  \qquad
	R_n\; = \; \sum_{x=1}^n r_x,
\end{equation}
with 
$$
	r_x = q_{x+1} - q_x .
$$
Since the dynamics is integrable, there exists actually a full set of conserved quantities.
{As said though, we are not primarily interested in studying peculiarities of integrable systems;
these other quantities will actually remain frozen in the macroscopic limit for the type of initial states that we consider,}
cf. \cite{BHO,A20}. 

Our choice of working with a one-dimensional lattice is mainly motivated by technical considerations: 
No detailed description of the spectrum of $M^{-1}\Delta$ is available in higher dimensions, and it is no longer possible to use the convenient elongation variables $r_x$ instead of $q_x$ if $d>1$. 
We expect however that our approach extends to higher dimensions, though one may not be able to carry over a full mathematical proof. 
In particular, in dimension 3 at moderate values of the disorder strength, high modes presumably become diffusive instead of staying localized. 
Since diffusive motion is too slow to be observed in the hyperbolic scaling limit, 
thermal fluctuations would appear frozen as well and we reckon that the hydrodynamic behavior will still be determined by the evolution of mechanical modes. 
The same phenomenology is at play for clean harmonic chains perturbed by a noise~\cite{stefano14,stefano16,stefano18}. 

\paragraph{Initial states.}
We assume that the system is initially in a local equilibrium state,
{that is presumably the most common state of matter in usual circumstances.} 
In our case, it is characterized by three macroscopic (smooth) profiles: the inverse temperature profile $\beta(y)$, 
the momentum profile $\bar p(y)$ and the elongation profile $\bar r(y)$, with $y\in [0,1]$. 
This means that, at the microscopic site $x$, the local inverse temperature is $\beta_x := \beta(x/n)$, 
the average momentum is $\bar p_x := \bar p(x/n)$ and the average elongation is $\bar r_x := \bar r(x/n)$.

For a classical dynamics, this state is given by a local Gibbs state, see \eqref{localgibbsstate} in Section~\ref{sec:modelclassical}.
If the system is quantum, the precise definition is provided by Assumptions \ref{AA1}-\ref{AA3} in Section~\ref{subsec: quantum model and results}, 
that involve some clustering conditions. % (see also remark \ref{rempair}).
 We notice that a local Gibbs state analogous to the classical local Gibbs state satisfies these conditions, 
see the definition \eqref{locallyGibbs} and Theorem~\ref{thmGibbs} in Section~\ref{subsec: quantum model and results}. 
However, our definition is more general and may in particular be satisfied by pure states. 
Actually, the eigenstate thermalization hypothesis and the notion of quantum typicality \cite{deutsch}\cite{srednicki}\cite{rigol} suggest that, 
if the system is prepared through some thermal process, the state of the system will share many similarities with a thermal state, 
and may satisfy the required conditions.      

\paragraph{Outline of the results.}
On the macroscopic level, the evolution of the three conserved fields $H_n$, $P_n$ and $R_n$ 
is described by the following system of conservation laws, known as Euler equations:
\begin{equation}\label{eq: Euler equations}
	\partial_t \fr(y,t)= \frac{1}{\bar{m}} \partial_y \fp(y,t), \qquad 
	\partial_t \fp(y,t)=\partial_y \fr(y,t),  \qquad
	\partial_t \fe(y,t)=\frac{1}{\bar{m}}\partial_y\left(\fr(y,t)\fp(y,t) \right)
\end{equation}
where $\overline{m}$ denotes the average mass
(one may not immediately recognize Euler equations, because they are here written in Lagrangian coordinates, through the use of the elongation variable $r_x$).
More precisely, we have established in \cite{BHO,A20} that, given a smooth test function $f$ on $[0,1]$, 
the following limit on average holds:  
\begin{equation}\label{eq: limit on average}
	\frac{1}{n} \sum_{x=1}^n f\left(\frac{x}{n}\right) 
	\langle z_x(nt) \rangle_{\rho} \to \int_0^1\ f(y)\fz(y,t)dy
\end{equation}
as $n\to\infty$ almost surely w.r.t.\@ the distribution of the masses, 
uniformly for any time $t$ smaller than some arbitrary time horizon $T$, 
where $z$ stands for $r$, $p$ or $e$, 
where $\fz$ stands for $\fr$, $\fp$ or $\fe$ respectively, 
and where $\langle \cdot \rangle_\rho$ denotes the averaging w.r.t.\@ a local Gibbs state.

{It is interesting to observe that the solution $\fe(y,t)$ to the above Euler equations reads 
$$
	\fe(y,t) \; = \; \frac{\fp^2(y,t)}{2 \bar m} + \frac{\fr^2 (y,t)}{2} + C(y)
$$
where $C(y)$ is a constant that depends on the initial conditions and is simply equal to the local temperature in the classical set-up.
This means that the macroscopic evolution of the energy is \emph{slaved} to the macroscopic evolution of the elongation and momentum, 
with a functional dependence that is the same as on the microscopic level. 
We expect this to be more general: If $O_x = f(r_x,p_x)$ is a microscopic local observable, its macroscopic description would be simply given by $f(\fr(y,t),\fp(y,t))$
(we have here taken a local observable depending on a single site but this restriction can be dropped).
We leave this to further investigations.}
%\begin{align*}
%	&\frac{1}{n} \sum_{x=1}^n f\left(\frac{x}{n}\right) 
%	\langle r_x(nt) \rangle_{\rho} \to \int_0^1\ f(y)\fr(y,t)dy, \\%\label{hydrolimitr1}\\
%	&\frac{1}{n} \sum_{x=1}^n f\left(\frac{x}{n}\right) 
%	\langle p_x(nt) \rangle_{\rho} \to \int_0^1\ f(y)\fp(y,t)dy,\\%\label{hydrolimitp1}\\	
%	&\frac{1}{n} \sum_{x=1}^n f\left(\frac{x}{n}\right) 
%	\langle e_x(nt)\rangle_{\rho} \to \int_0^1\ f(y)\fe(y,t)dy, %\label{hydrolimite1}	
%\end{align*}
%as $n\to \infty$,  almost surely w.r.t.\@ the distribution of the masses, 
%and uniformly for any time $t$ smaller than some arbitrary time horizon $T$. 
%Here $\langle \cdot \rangle_\rho$ denotes the averaging w.r.t.\@ a local Gibbs state.  

In this paper, we extend our previous results by proving a convergence in quadratic variation and allowing for a broader class of initial states in the quantum set-up. 
In Theorem \ref{thmhighermomentclassical} and Theorem \ref{mainthmQ} below, 
devoted respectively to the classical and quantum dynamics, we establish that 
\begin{align*}
	&\left\langle\left( \frac{1}{n} \sum_{x=1}^n f\left(\frac{x}{n}\right) r_x(nt) -  \int_0^1\ f(y)\fr(y,t)dy\right)^2 \right\rangle_\rho \; \to \; 0, \\
	&\left\langle\left( \frac{1}{n} \sum_{x=1}^n f\left(\frac{x}{n}\right) p_x(nt) -  \int_0^1\ f(y)\fp(y,t)dy\right)^2 \right\rangle_\rho \; \to \; 0, \\
	&\left\langle\left( \frac{1}{n} \sum_{x=1}^n f\left(\frac{x}{n}\right) e_x(nt) -  \int_0^1\ f(y)\fe(y,t)dy\right)^2 \right\rangle_\rho \; \to \; 0.
\end{align*}
as $n\to \infty$, almost surely w.r.t.\@ the distribution of the masses,
and uniformly for any time $t$ smaller than some arbitrary time horizon $T$. 
In addition, $\langle \cdot \rangle_\rho$ refers now to any local equilibrium state as described above. 

The convergence in quadratic variation implies a strong concentration around the average behavior 
and provides concrete support to the claim that ergodic averages suppress fluctuations at the macroscopic scale. 
It is also a first step towards a more global control over out-of-equilibrium fluctuations. 
Indeed, in macroscopic fluctuation theories~\cite{RevModPhys.87.593}, 
the Euler equation would appear as the minimizer of an action functional and macroscopic fluctuations are expected to be exponentially suppressed as a function of the system size. 
We hope that the techniques developed in this paper will prove useful in establishing such a global picture.

As they stand, our bounds already yield a convergence in probability of the empirical distributions to their classical limit: 
For any $\delta > 0$, 
$$
	\mathsf P_\rho \left(\left|  \frac{1}{n} \sum_{x=1}^n f\left(\frac{x}{n}\right) z_x(nt) -  \int_0^1\ f(y)\fz(y,t)dy \right| > \delta \right) \quad \to \quad 0
$$ 
as $n\to \infty$, with $z=r$, $z=p$ or $z=e$.
In the classical set-up,  $\mathsf P_\rho$ denotes the probability induced by the average $\langle \cdot \rangle_\rho$, and the claim follows from Markov inequality. 
Indeed, writing 
$$
	O_n \; = \; \frac{1}{n} \sum_{x=1}^n f\left(\frac{x}{n}\right) z_x(nt) -  \left(\int_0^1\ f(y)\fz(y,t)dy\right), 
$$
we derive
$$
	\mathsf P_\rho(|O_n|> \delta) \; = \;  \mathsf P_\rho(O_n^2> \delta^2) \;\le\; \frac{\langle O_n^2 \rangle_\rho}{\delta^2}
$$
that converges to 0 as $n\to\infty$.

In the quantum set-up, one first needs to properly define the probability $\mathsf P_\rho$ on the left-hand side. 
Let us define the operators
$$
	O_n \; = \; \frac{1}{n} \sum_{x=1}^n f\left(\frac{x}{n}\right) z_x(nt) -  \left(\int_0^1\ f(y)\fz(y,t)dy\right)\mathrm{Id}
$$ 
for any $n\ge 1$. 
Let also $P_\delta^n$ be the projector on the spectral subspace spanned by all normalized states $f$ such that $\| O_n^2f\|_2 \le \delta^2$.
% of $O_n$ with absolute spectral value smaller than $\delta$. 
%Let us first assume that $\rho$ is a pure state, i.e.\@ $\rho =  {|\pmb\psi\rangle\langle \pmb\psi|}$ for some state ${|\pmb\psi\rangle}$. 
Following the basic postulates of quantum mechanics, we define 
$$ 
	\mathsf P_\rho (|O_n|> \delta) \; = \; \langle\mathrm{Id} - P_\delta^n \rangle_\rho
$$
and we obtain
$$
	 \langle \mathrm{Id} - P_\delta^n \rangle_\rho
	 \; \le \;
	 \frac{\langle (\mathrm{Id} - P_\delta^n) O_n^2 (\mathrm{Id} - P_\delta^n) \rangle_\rho}{\delta^2} \; \le \; \frac{\langle O_n^2 \rangle_\rho}{\delta^2}
$$
that converges here as well to 0 as $n\to\infty$. 
To our knowledge, it is the first time that such a limit in probability is shown starting from a fully quantum system. 
Reminding that $\rho$ may be a pure state, the result shows how the probabilistic quantum description may lead to a purely classical and deterministic behavior in the macroscopic limit.

\paragraph{Remarks on the proofs.}
The proofs of Theorem \ref{thmhighermomentclassical} and Theorem \ref{mainthmQ} rely first on our previous result, the convergence on average \eqref{eq: limit on average},
and on the methods developed in \cite{BHO,A20}. 
We refer to these papers for a heuristic description of these methods.  
Let us here comment on the principal difficulty that needed to be overcome in order to establish the Theorems \ref{thmhighermomentclassical} and \ref{mainthmQ}. 
The main technical step was to derive a sufficient decay of correlations for local observables at the hyperbolic time scale, cf.\@ \eqref{eq:4} for the classical case, and \eqref{Qeq:4} for the quantum case.
Indeed, let us take the observable $p_x$ as an example. 
For a finite time $t$, $\expval{p_x(t)p_y(t)}_{\rho}$ decays sufficiently fast as a function of $|x-y|$, assuming that the initial state has suitable decaying properties, as we do. 
This decay follows from analyticity and holds independently of localization.
However, this argument breaks as soon as we rescale the time $t$ with the length of the system $n$, since faraway observables may in principle become entangled and correlated. 
But correlation is suppressed by localization in our case. 
More specifically, when $|x-y| \geq n^{\theta}$ for some suitable $0<\theta<1$, 
localization estimates \eqref{localizationestimate} in Lemma \ref{localizationlemma} provide the desired decay for the high modes of $p_x(nt)$ at any time-scale,
 cf.\@ \eqref{highmodedef} and \eqref{Qhighmodedef}. 
We prove that this decay holds and that the contributions from the low modes and from the sites with $|x-y|<n^{\theta}$ is vanishing.

\paragraph{Plan of the paper.}
In the rest of this paper, the above claims will be made mathematically rigorous. 
First, we deal with the classical dynamics in Section \ref{sec:classicalcase}, where we state and prove Theorem~\ref{thmhighermomentclassical}. 
Second, Section \ref{section: quantum} is devoted to the study of the quantum dynamics; Theorems \ref{mainthmQ} and \ref{thmGibbs} are shown there.

\section{Classical case} \label{sec:classicalcase}

\subsection{Model and Result} \label{sec:modelclassical}
	We consider a classical disordered chain of $n$ harmonic oscillators in one dimension. 
	%We define the model with $n$ particles. 
	The phase space is given by $\{(p,q) \in \bR^{2n}\}$, where $p=(p_1,\dots,p_n)$, $q=(q_1,\dots,q_n) \in \bR^n$ denote the momentum and position vectors, respectively.
	We label particles by $x \in \bI_n=\{1,\dots,n\}$. %{\color{purple}I have changed $\bT$ to $\bI$ since this seems to be the notation used in the quantum case.}
	Let us define the discrete gradients $\nabla_- :\bR^{n-1}\to\bR^n$ and $\nabla_+:\bR^n \to \bR^{n-1}$ such that, 
	for any $v \in \bR^{n-1}$, $w \in \bR^n$,  $x \in \bI_n$, and $y \in \bI_{n-1}$, we have
	\begin{equation} \label{eq:discretegrad}
	(\nabla_- v)_x= v_{x}-v_{x-1}, \quad \quad (\nabla_+ w)_y= w_{y+1}-w_{y},
	\end{equation}
  	with the convention $v_n=0$. 
	Let $\Delta$ denote the discrete Laplacian with free boundary conditions: For any $v \in \bR^n$ and $x \in \bI_n$, 
  	\begin{equation} \label{eq:discreteLap}
  	\Delta = \nabla_-\nabla_+, \quad \quad (\Delta v)_x=v_{x+1}-2v_x+v_{x-1}, 
	\end{equation}  	 
	with the conventions $v_{n+1}=v_n$ and $v_0=v_1$.
	
	The Hamiltonian $H_n$ is defined on $\bR^{2n}$ as:
	\begin{equation} \label{Hamclassical}
	H_n(p,q)= \frac12 \sum_{x=1}^n \left( \frac{p_x^2}{m_x}+ (q_{x+1}-q_x)^2 \right),
	\end{equation}	 
	where $\{m_x\}_{x=1}^{\infty}$ are i.i.d random variables defined on a 
	probability space $(\Omega, \mathcal{F},\bP)$. We denote the expectation w.r.t $\bP$ with $\bE$. We let also
	$$
		\bar m  \; = \; \bE(m_x).
	$$
	Furthermore, we assume the law of the random variables to be smooth and compactly supported in $[m_-,m_+]$ with $0<m_-<m_+<\infty$. 
	These assumptions are the same as in \cite{FA}\cite{BHO} and allow to exploit results on Anderson localization. 
	Notice that we assume free boundary conditions $q_0=q_1$ and $q_n=q_{n+1}$ in 
	\eqref{Hamclassical}, for the sake of concreteness. 
	
	Let us define the elongation variable $r$ as follows: for any $x \in \bI_{n-1}$,
	\begin{equation} \label{eq:elogation}
	r_x = (\nabla_+ q)_x.
	\end{equation}	  
	Free boundary conditions in terms of elongation variable read $r_0=r_n=0$. From 
	now on, we describe our model in terms of these variables.
	%\footnote{This means that we describe this system from a center of mass observer, and state space is $\bR^{2n-1}$ in this coordinates.}. 
	The equations of motion (Hamiltonian 
	dynamics) in this coordinate is as follows:
	\begin{equation} \label{eqofmotion}
	\dot{r}_x= (\nabla_+ M^{-1}p)_x, \quad x \in \bI_{n-1},
	\qquad
	\quad \dot{p}_x= (\nabla_-r)_x, \quad x \in \bI_n, 
	\end{equation}
	where $M= \diag(m_1,\dots,m_n) $ is the diagonal matrix of masses, in the 
	first equation $r_0=r_n=0$ is considered, and $\dot{a}$ denotes the time derivative, i.e.\@ $\dot{a}={da}/{dt}$.
	
	Despite the existence of a full set of conserved quantities, we focus on the evolution of the locally conserved quantities $H_n$, $P_n$ and $R_n$ defined in \eqref{conserved fields} 
	(let us notice that while $H_n$ and $P_n$ are also globally conserved, the conservation of $R_n$ is broken at the boundaries).
%	As we mentioned, this system is completely 
%	integrable and it has $n$ conserved quantities. But here we 
%	are interested in the following three locally conserved quantities:
%	energy $H_n$, momentum $P_n$, and elongation $R_n$ \eqref{conserved fields}.
%	%\begin{equation} \label{ConservedQuantity}
%	%	H_n=\sum_{x=1}^n \frac12 \left(\frac{p_x^2}{m_x}+r_x^2 \right)=:
%	%	\sum_{x=1}^n e_x
%	%	, \quad
%	%	P_n=\sum_{x=1}^n p_x,  \quad
%	%	R_n=\sum_{x=1}^nr_x.
%	%\end{equation}
%	Notice that in the first expression of \eqref{conserved fields}, 
%	we implicitly defined the energy of the 
%	$x$-th particle as 
%	\begin{equation}\label{energy}
%	e_x=\frac{1}2 \left(\frac{p_x^2}{m_x}+r_x^2 \right).
%	\end{equation}	 
%	It is worth mentioning that $H_n$ and $P_n$ are globally conserved, whereas 
%	conservation of $R_n$ is broken at the boundaries due to the above-mentioned 
%	boundary conditions ($\dot{R}_n=p_n/m_n-p_1/m_1$).  
	As initial state of the chain, we take a locally Gibbs state parametrized by the local values of these three quantities. 
%	Corresponding to the aforementioned conserved quantities, we define a locally 
%	Gibbs state (local equilibrium state) demonstrating the chain's initial state.
	This state is thus parameterized by a temperature profile $\beta \in C^0([0,1])$ with $0< \beta_- \leq \beta(y) \leq \beta_+ < \infty$, 
	a momentum profile $\bar{p} \in C^1([0,1])$ and an elongation profile $\bar{r} \in C^1([0,1])$ with $\bar{r}(0)=\bar{r}(1)=0$.
	Given $\beta, \bar{p},\bar{r}$, as above, the locally Gibbs state is a probability distribution 
	defined on $\bR^{2n-1}$, with the following probability density 
	$\rho_{\beta,\bar{p},\bar{r}}^n$:
	\begin{equation} \label{localgibbsstate}
	 \rho_{\beta,\bar{p}, \bar{r}}^n(r,p) = \frac{1}{Z_n} \exp\left(-\frac12 
	 \sum_{x=1}^n \left[\frac{\beta_x}{m_x} \left(p_x- \bar{p}_x
	 \frac{m_x}{\bar{m}}\right)^2 + \beta_x \left(r_x -
	 \bar{r}_x\right)^2 \right] \right),
	\end{equation}	   
	where we used the shorthand notation $\beta_x:= \beta(\frac{x}{n})$, 
	$\bar{p}_x:=\bar{p}(\frac{x}{n})$, $\bar{r}_x:=\bar{r}(\frac{x}{n})$ and where $Z_n$ is a normalizing constant. 
	We will drop the subscripts and superscripts from $\rho$ whenever there is no risk of confusion. 
	We denote the expectation w.r.t $\rho$ by $\expval{.}_{\rho}$:
	For any suitable observable $O(r,p): \bR^{2n-1} \to \bR$, we define 
	\begin{equation} \label{eq:expectation}
	 \langle O\rangle_{\rho} := \int_{\bR^{2n-1}} O(r,p) 
	 \rho_{\beta,\bar{p},\bar{r}}^n(r,p)
	 dr dp.
	\end{equation}
	
	At the macroscopic level, fix the macroscopic time $T>0$ and let the initial profile of elongation and momentum evolve according to the following system of conservation laws: 
	\begin{equation} \label{eq:macro}	
	\partial_t \fr(y,t)= \frac{1}{\bar{m}} \partial_y \fp(y,t), \quad 
	\partial_t \fp(y,t)=\partial_y \fr(y,t),   \quad 
%	\end{equation}
%	\begin{equation} \label{eq:macro2}
	\partial_t \fe(y,t)=\frac{1}{\bar{m}}\partial_y\left(\fr(y,t)\fp(y,t) \right),
	\end{equation}
	where $\fr,\fp \in C^1([0,1]\times [0,T])$, with the following initial and boundary conditions: 
	\begin{equation} \label{eq:macro3}
	\begin{split}
		&\fr(y,0)=\bar{r}(y), \quad \fp(y,0)=\bar{p}(y), \quad \fe(y,0)=
		\frac{\bar{p}^2(y)}{2\bar{m}}+\frac{\bar{r}^2(y)}{2}
		+ \frac{1}{\beta(y)}, \\
		&\fr(0,t)=\fr(1,t)=0 \quad \forall t \in [0,T]. 	
	\end{split}
	\end{equation}
	The hydrodynamic limit on average \eqref{eq: limit on average} was proven in \cite{BHO} and we now improve this result into 
	\begin{theorem} \label{thmhighermomentclassical}
		Let $f \in C^0([0,1])$ be an arbitrary test function. Fix $T>0$. Initially 
		let the chain to be in the locally Gibbs state $\rho^n_{\beta,\bar{p},\bar{r}}$
		  \eqref{localgibbsstate}, corresponding to the profiles $\beta$, $\bar{p}$, 
		$\bar{r}$ satisfying the assumptions stated in the definition of 
		\eqref{localgibbsstate}, and recall $\expval{\cdot}_{\rho^n}$ 
		to be the average w.r.t
		this state \eqref{eq:expectation}. Denote the solution to the microscopic evolution 
		equation \eqref{eqofmotion} by $(r(t),p(t),e(t))$. Moreover, let 
		$\fr(y,t), \fp(y,t), \fe(y,t)$ denote the solution to the macroscopic 
		evolution equation \eqref{eq:macro} with
		initial and boundary conditions \eqref{eq:macro3}. 
		Then for any $t \in [0,T] $ we have: 
	\begin{equation} \label{highermomentclassicalr}
	\left\langle \left(\frac1n \sum_{x=1}^n f(\frac{x}{n})r_x(nt) -\int_0^1 
	f(y) \fr(y,t) dy \right)^2 \right\rangle_{\rho^n} \to 0,
	\end{equation}
	\begin{equation}\label{highermomentclassicalp}
	\left\langle \left(\frac1n \sum_{x=1}^n f(\frac{x}{n})p_x(nt) -\int_0^1 
	f(y) \fp(y,t) dy \right)^2 \right\rangle_{\rho^n} \to 0,
	\end{equation}
	\begin{equation}\label{highermomentclassicale}
	\left\langle \left(\frac1n \sum_{x=1}^n f(\frac{x}{n})e_x(nt) -\int_0^1 
	f(y) \fe(y,t) dy \right)^2 \right\rangle_{\rho^n} \to 0,
	\end{equation}
	as $n \to \infty$ almost surely w.r.t the distribution of the masses.
	\end{theorem}		
	Before proceeding, let us briefly mention that the 
	usual mathematical tool to control moments similar to the above ones is the relative entropy method (cf. \cite{Yau91}\cite{stefano93}). 
	However, it is not clear that this method could be applied here, and we are able to handle them with a novel approach, i.e.\@ by using localization estimates.
We devote the rest of this section to the proof of Theorem \ref{thmhighermomentclassical}.

\subsection{Solution to the equation of motion and localization}
		The proof of Theorem \ref{thmhighermomentclassical} rests on localization of the high eigenmodes of the chain. 
		In this section, we solve the equation of motion \eqref{eqofmotion} explicitly and state the localization estimates. 
		Consequently, we can obtain proper decay estimates. 
	
\subsubsection{Solution to the equation of motion}	 
	We denote the inner product in $\bR^n$ by $\langle,\rangle_n$, and we drop the subscript 
	whenever it is convenient.	
	Recall the definition of matrices $\nabla_-$, $\nabla_+$, $\Delta$, and $M$ 
	\eqref{eq:discretegrad}, \eqref{eq:discreteLap}, \eqref{eqofmotion}. 
	Taking the time derivative of \eqref{eqofmotion} we get: 
	\begin{equation} \label{eqofmotion2}
	\ddot{r}_x=(\nabla_+M^{-1}\nabla_-r)_x, \quad 1 \leq x \leq n-1; 
	\quad \ddot{p}_x=(\Delta M^{-1}p)_x, \quad 1 \leq x \leq n,
	\end{equation}
	where one should recall $r_0=r_n$ in the first equation, and free boundary conditions 
	($p_0/m_0=p_1/m_1$, $p_n/m_n=p_{n+1}/m_{n+1}$) in the second equation.
	
	Notice that 
	$\Delta^{\dagger}=\Delta$ and $\nabla^{\dagger}_+=-\nabla_-$. Define 
	\begin{equation} \label{mainmatrices}
		A_p:=M^{-\frac12}(-\Delta)M^{-\frac12} \in \bR^{n\times n}, \quad
		A_r:=- \nabla_+ M^{-1}\nabla_- \in \bR^{(n-1)\times(n-1)}. 
	\end{equation}	   
	$A_p$ is symmetric, positive semidefinite, 
	and it has a non-degenerate spectrum almost 
	surely (cf. proposition II.1 of \cite{Kunz}). We denote this spectrum by 
	$0=\omega_0^2< \omega_1^2<\dots<\omega_{n-1}^2$, and designate their corresponding 
	eigenvectors by $\{ \varphi^k \}_{k=0}^{n-1}$, such that they form an orthonormal basis 
	for $\bR^n$. In particular, $\expval{\varphi^k,\varphi^{j}}=\delta_{k,j}$.
	
	On the other hand, one can check that $A_r$ is a positive symmetric matrix. 
	Furthermore, if for 
	$k \in \bI_{n-1}$, we denote 
	$\phi^k:= \frac{1}{\omega_k} \nabla_+M^{-\frac12} \varphi^k$, then by using the fact 
	that $\Delta=\nabla_-\nabla_+$,  it is straightforward to 
	observe that $\{ \phi^k \}_{k=1}^{n-1}$ is a complete set of eigenvectors for $A_r$
	 with eigenvalues $\omega_1^2<\omega_2^2<\dots<\omega_{n-1}^2$ similar to the 
	 eigenvalues of $A_p$, and they form 
	 an orthonormal basis for $\bR^{n-1}$ i.e., 
	 $$A_r \phi^k=\omega_k^2 \phi^k, \quad 1 \leq k \leq n-1; \quad 
	 \langle \phi^k, \phi^j \rangle_{n-1} =\delta_{k,j}, \quad 1 \leq k,j \leq n. $$
	  We obtain
	the normal modes of the chain as follows: define 
	 $\hat{p}_k$ for $k \in \bI_{n-1}^o=\{0,\dots,n-1\}$, and $\hat{r}_k$ for $k \in \bI_{n-1}$   as 
	 \begin{equation} \label{rphat}
	 \hat{p}_k:= \langle \varphi^k, M^{-\frac12} p \rangle_n= 
	 \langle \tilde{\varphi}^k,p \rangle_{n}, \quad  
	 \hat{r}_k:= \langle \phi^k,r \rangle_{n-1}. 
	\end{equation}	    
	Notice that we implicitly defined $\tilde{\varphi}^k := M^{-\frac12}\varphi^k$. We take 
	$\hat{r}_0=0$ by convention. Let us mention that from the definition we have 
	$\varphi^0_x=m_x^{1/2} (\sum_{y=1}^n m_y)^{-1/2}$ for all $x\in\mathbb I_n$.
	%$\varphi^0=(\sum_{x=1}^n)^{-\frac12} M^{\frac12} \ket{\mathbbm{1}}$, where $\mathbbm{1}$ denotes the vector $(\underbrace{1,1,\dots,1}_{n \text{ times}})^{\dagger}$.
	
	From equation of motion \eqref{eqofmotion2}, and  definition of $\hat{p}_k$,
	 $\hat{r}_k$ \eqref{rphat}, thanks to the fact that  $\phi^k$, $\varphi^k$,
	  are eigenvetors of $A_r$, $A_p$ with eigenvalue $\omega_k^2$ respectively,  
	 we obtain: 
	 \begin{equation} \label{hatevolve}
	 \ddot{\hat{p}}_k=-\omega_k^2 \hat{p}_k, \qquad \ddot{\hat{r}}_k= -\omega_k^2 \hat{r}_k.
	\end{equation}	 	
	Solving the latter yields:  for $k \in \bI_{n-1}^o$ (recall the convention $\hat{r}_0=0$),
	\begin{equation*}% \label{hatsol}
	\begin{split}	
	\hat{p}_k(t)= \hat{p}_k(0) \cos(\omega_kt)-\hat{r}_k(0)\sin(\omega_kt)
	=\langle M^{-\frac12} \varphi^k, p(0)\rangle_n \cos(\omega_kt)-
	\langle \phi^k,r(0) \rangle_{n-1} \sin(\omega_kt), \\
	\hat{r}_k(t)= \hat{r}_k(0) \cos(\omega_kt)+\hat{p}_k(0)\sin(\omega_kt)
	=\langle \phi^k,r(0) \rangle_{n-1} \cos(\omega_kt) + \langle M^{-\frac12} \varphi^k,
	p(0) \rangle_n \sin(\omega_kt).  \\
	\end{split}	
	\end{equation*}
	Therefore, thanks to the above expressions, by using the inverse of \eqref{rphat},
	time evolution of $r,p$ has the following explicit expression: 
	\begin{equation} \label{soleqmotion}
	\begin{split}		
	&p(t)= \sum_{k=0}^{n-1} M^{\frac12} \varphi^k\hat{p}_k(t) = \sum_{k=0}^{n-1}
	\left(\cos(\omega_kt) \hat{p}_k(0) - \sin(\omega_kt) \hat{r}_k(0) \right) M^{\frac12}
	\varphi^k, \\
	& r(t)= \sum_{k=1}^{n-1}  \phi^k\hat{r}_k(t) = \sum_{k=1}^{n-1}
	\left(\cos(\omega_kt) \hat{r}_k(0) + \sin(\omega_kt) \hat{p}_k(0) \right) 
	\phi^k.
%	&p_x(t)=\sum_{k=0}^{n-1} \sqrt{m_x} \varphi^k_x \hat{p}_k(t)=
%	\sum_{k=0}^{n-1} \left(\cos(\omega_kt) \hat{p}_k(0) - \sin(\omega_kt) \hat{r}
%	_k(0) \right)\sqrt{m_x} \varphi^k_x, \\
%	 &r_x(t)= \sum_{k=1}^{n-1}  \phi^k_x \hat{r}_k(t) = \sum_{k=1}^{n-1}
%	\left(\cos(\omega_kt) \hat{r}_k(0) + \sin(\omega_kt) \hat{p}_k(0) \right) 
%	\phi^k_x.
	\end{split}
	\end{equation}

\subsubsection{Localization}	
	The proof of the limit on average \eqref{eq: limit on average} relies on localization properties of 
	$\tilde{\varphi}^k=M^{-\frac12}\varphi^k$ for $k \gg \sqrt{n}$. 
	We use similar properties in order to obtain proper decay estimates and deduce \eqref{highermomentclassicalr}-
	\eqref{highermomentclassicale}. Therefore, 
	we express the desired localization estimates directly from \cite{BHO} and \cite{A20} (cf.\@ \cite{A98}\cite{Theo}  for general theory, 
	and \cite{FA}\cite{Theo}  for more precise 
	estimates). 
	\begin{lemma} \label{localizationlemma}
		Recall the definition of the vectors $\varphi^k$,  as 
		ordered eigenvetors of $A_p$ \eqref{mainmatrices} and 
		$\tilde{\varphi}^k= M^{-\frac12}\varphi^k$. 
		Let $0<\alpha<\frac12$, and define $I(\alpha):=]n^{(1-\alpha)},n] \cap \bZ$. There exists $C,c>0$ 
		independent of $n$ such that 
		\begin{equation} \label{localizationestimate}
		\bE \left(\sum_{k \in I(\alpha)} |\tilde{\varphi}^k_x \tilde{\varphi}^k_y| \right)
		\leq C \exp(-\frac{c|x-y|}{\xi(\alpha)}),  	\quad \text{with} \quad 
		\xi(\alpha)=n^{2 \alpha},	
		\end{equation}
	for $x,y \in \bI_n$, 
	where we remind that $\bE$ is the expectation with respect to the distribution of the masses. 
	Notice that thanks to the fact that masses are compactly supported, the same estimate holds for $\varphi^k$ with possibly different constants $c,C>0$ independent of $n$. 	
	\end{lemma}
	The above lemma is enough to deal with the terms involving $p$ variables. 
	The terms with $r$ variables involve the vectors $\phi^k$ but, instead of directly establishing a localization estimate for these, we will use Lemma 6.3 in \cite{A20}:  
		\begin{lemma} \label{boundomegalemma}
		Recall the definition of $A_p=M^{-\frac12}(-\Delta)M^{\frac12}$ 
		\eqref{mainmatrices}, and its ordered 
		eigenvalues $0=\omega_0<\omega_1<\dots<\omega_{n-1}$. Fix $\alpha,\gamma>0$, 
		such that $0< 2 \alpha<\gamma<1$. There exists almost surely $n_0 \in \bN$ such
		 that $\forall n> n_0$, and for $k \in I(\alpha)=]n^{(1-\alpha)},n-1]\cap \bZ$ 
		 we have:
		 \begin{equation} \label{boundomega}
		 \omega_k^{-1} \leq cn^{\frac{3\gamma}{2}},
		\end{equation}		 	 
		where $c$ is a constant independent of $n$. 
	 \end{lemma}
	 \noindent
	Let us remark that the estimate \eqref{boundomega} is not optimal; we exepect indeed that $1/\omega_k \leq cn$ for all $k \in \bI_{n-1}$. 
	However, \eqref{boundomega} is sufficient for our purposes.

\subsection{Evolution of $r$ and $p$} \label{sec:rp}
	In this section, we prove the two limits \eqref{highermomentclassicalr} and \eqref{highermomentclassicalp} featuring in Theorem~\ref{thmhighermomentclassical}. 
	In the proof, we take advantage of the limit on average \eqref{eq: limit on average} which have been proved in \cite{BHO}. 
	We also use a couple of bounds introduced in \cite{BHO}.%, we only address them, and refer to \cite{BHO} for their proofs.
%	Let us emphasize that the proof of \eqref{hydrolimitr1}, and \eqref{hydrolimitp1} have been done \textit{without using} the localization estimates. 
%	On the other hand, we use the localization estimates to prove \eqref{highermomentclassicalr}, and \eqref{highermomentclassicalp}. 
%	We should mention that one can prove the same result without using localization, at least by assuming extra regularity assumptions on $\beta$, and $f$. 
%	However, we believe localization estimate \eqref{localizationestimate} provides a "more natural" proof, which can be generalized to other cases, such as the case where $\beta$ is less regular.  

	\begin{proof}[Proof of Theorem \ref{thmhighermomentclassical}: 
	\eqref{highermomentclassicalr}, 
	\eqref{highermomentclassicalp}] 
	We define the \emph{fluctuation} variables $\tilde{r}_x(nt)$ and $\tilde{p}_x(nt)$
	 for any $x \in \bI_n$ as 
	follows:
	\begin{equation} \label{tildevariable}
	\tilde{r}_x(nt):= r_x(nt) -\expval{r_x(nt)}_{\rho^n}, \quad 
	\tilde{p}_x(nt):=p_x(nt)- \expval{p_x(nt)}_{\rho^n}.
	\end{equation}
	Adding and subtracting 
	$\frac1n\sum_{x=1}^n f(\frac{x}{n})\expval{p_x(nt)}_{\rho^n}$ inside the 
	square in \eqref{highermomentclassicalp}, and using the fact that $\expval{\tilde{p}_x(nt)}_{\rho}=0$
	 we get:
	 \begin{equation}
	\begin{split}	 
	 \left\langle \left(\frac1n \sum_{x=1}^n f(\frac{x}{n})p_x(nt) -\int_0^1 
	f(y) \fp(y,t) dy \right)^2 \right\rangle_{\rho^n}= 
	\left\langle \left(\frac1n \sum_{x=1}^n f(\frac{x}{n})\tilde{p}_x(nt) \right)^2
	\right\rangle_{\rho^n}\\
	 + \left\langle \left(\frac1n \sum_{x=1}^n f(\frac{x}{n})\expval{p_x(nt)}_{\rho^n} 
	-\int_0^1 
	f(y) \fp(y,t) dy \right)^2 \right\rangle_{\rho^n}. 
		\end{split}
		\end{equation}	
		Notice that the last term converges to zero almost surely w.r.t.\@ the distribution of the masses thanks to the limit on average \eqref{eq: limit on average}. 
		One can do the same for $r$ as well. 
		Therefore, in order to prove \eqref{highermomentclassicalr}, and 
		\eqref{highermomentclassicalp}
		 it is sufficient to prove: 
		\begin{equation} \label{eq:proof3}
		\cR_n:=\left\langle \left(\frac1n \sum_{x=1}^n f(\frac{x}{n})\tilde{r}_x(nt) \right)^2
	\right\rangle_{\rho^n} \to 0, \quad 
	\cP_n:=\left\langle \left(\frac1n \sum_{x=1}^n f(\frac{x}{n})\tilde{p}_x(nt) \right)^2
	\right\rangle_{\rho^n} \to 0, 
		\end{equation}
	as $n \to \infty$, almost surely w.r.t the distribution of the masses.
	
	 Before proceeding, notice that the evolution equation \eqref{eqofmotion} is linear.
	 Therefore, $\expval{r_x}_{\rho^n}$, $\expval{p_x}_{\rho^n}$ satisfy the same set of 
	 equations with modified initial data. This is still true for $\tilde{r}_x$ and 
	 $\tilde{p}_x$. Consequently, one can solve the equation of motion featuring 
	 $\tilde{r}_x$ and $\tilde{p}_x$ as in \eqref{soleqmotion}.  More precisely,
	 denote the vectors $\tilde{p}=p-\expval{p}_{\rho^n}$, 
	 $\tilde{r}=r-\expval{r}_{\rho^n}$ and define 
	 \begin{equation} \label{tildenitial}
	 \hat{\tilde{p}}_k= \hat{\tilde{p}}_k(0):=\langle \varphi^k, M^{-\frac12} \tilde{p} 
	 \rangle_n, \quad  
	 \hat{\tilde{r}}_k:= \langle \phi^k, \tilde{r} \rangle_{n-1}, 
	\end{equation}	 
	then $\tilde{r}_x(t)$, $\tilde{p}_x(t)$ is given by:
	 \begin{equation}\label{tildesol}
	\begin{split}	 
	& \tilde{r}_x(t)= \sum_{k=1}^{n-1}  \phi^k_x \hat{\tilde{r}}_k(t) = \sum_{k=1}^{n-1}
	\left(\cos(\omega_kt) \hat{\tilde{r}}_k(0) + \sin(\omega_kt) \hat{\tilde{p}}_k(0) \right) 
	\phi^k_x, \\
	& \tilde{p}_x(t)=\sum_{k=0}^{n-1} \sqrt{m_x} \varphi^k_x \hat{\tilde{p}}_k(t)=
	\sum_{k=0}^{n-1} \left(\cos(\omega_kt) \hat{\tilde{p}}_k(0) - \sin(\omega_kt) 
	\hat{\tilde{r}}
	_k(0) \right)\sqrt{m_x} \varphi^k_x.
		\end{split}	 
	 \end{equation}   
	 In order to prove \eqref{eq:proof3}, we divide each sum into three parts, and prove that
	each part converges to zero almost surely w.r.t the distribution of the masses. Based
	 on this, we divide the rest of the proof into three steps.
	
	 \paragraph{Step 1: Low modes.}
	  Lemma \ref{localizationlemma} only provides  
	 localization of the high modes. Hence we treat the contribution of the low modes separately.
	  For proper $\gamma$ such that $0<\gamma<\frac12$, define $\tilde{p}^o(nt)$, and 
	  $\tilde{r}^o(nt)$ as the ``low mode portion'' of $p$, and $r$: for $x \in \bI_n$ let 
	  \begin{equation} \label{lowmodedef}
			\tilde{p}^o_x(nt):= \sum_{k \in \bZ \cap [0,n^{(1-\gamma)}]} \hat{\tilde{p}}_k 
			(nt) \sqrt{m_x} \varphi_x^k, \quad 	  
			\tilde{r}^o_x(nt):= \sum_{k \in \bZ \cap [1,n^{(1-\gamma)}]} \hat{\tilde{r}}_k 
			(nt)  \phi_x^k. 
		\end{equation}	    
		In this step, we prove for any $t \in [0,T]$: 
		\begin{equation} \label{lowmodelem}
		\mathcal{L}_n := \left\langle \left( \frac1n \sum_{x=1}^n f(\frac{x}{n})
		\tilde{r}^o_x(nt) \right)^2 \right\rangle_{\rho^n} \to 0,  \quad
		\mathcal{L}_n' := \left\langle \left( \frac1n \sum_{x=1}^n f(\frac{x}{n})
		\tilde{p}^o_x(nt) \right)^2 \right\rangle_{\rho^n} \to 0.
 		\end{equation}
		for any realization of the masses.
		
		First, notice that $\rho^n$ \eqref{localgibbsstate} is a Gaussian probability 
		density function, and $\{r_x\}_{x=1}^{n-1}$ and $\{p_x\}_{x=1}^n$ 
		are independently distributed. Therefore, 
		we have for proper $x,y \in \bI_n$:  
		\begin{equation} \label{cov1}
		\expval{\tilde{r}_x\tilde{r}_y}_{\rho^n}=\frac{\delta_{x,y}}{\beta(\frac{x}{n})},
			\quad \expval{\tilde{p}_x\tilde{p}_y}_{\rho^n}=\delta_{x,y}
			 \frac{m_x}{\beta(\frac{x}{n})}, \quad \expval{\tilde{r}_x\tilde{p}_y}=0. 
		\end{equation}		
		Recall the definition of $\hat{\tilde{r}}_k$, $\hat{\tilde{p}}_k$ \eqref{tildenitial}. 
		For any $k$, there exists $c>0$ independent of $n$, and 
		independent of the realization of the masses such that
		\begin{equation} \label{initialbound}
		\begin{split}		
		\expval{\hat{\tilde{p}}_k^2}_{\rho^n} = \sum_{x,y=1}^n \varphi^k_x \varphi^k_y
		\frac{\expval{\tilde{p}_x\tilde{p}_y}_{\rho^n}}{\sqrt{m_x m_y}}=
		\sum_{x=1}^n \beta^{-1}(\frac{x}{n}) (\varphi^k_x)^2 \leq c, \\
		\expval{\hat{\tilde{r}}_k^2}_{\rho^n} = \sum_{x=1}^{n-1} \beta^{-1}(\frac{x}{n})
		(\phi^k_x)^2 \leq c, \quad
		\expval{\hat{\tilde{p}}_k\hat{\tilde{r}}_{k'}}_{\rho^n}=0, 
		\end{split}		
		\end{equation}		 
	where first we exploit \eqref{cov1}, then we take advantage of the fact that 
	$0<\beta_-<\beta(y)$ (in fact, we can take $c=\beta_-^{-1}$), and finally we use the fact 
	that $\phi^k$, and $\varphi^k$ are members of two orthonormal basis. Recall 
	$\hat{\tilde{p}}_k(nt)$, $\hat{\tilde{r}}_k(nt)$ from \eqref{tildesol}, then from
	 \eqref{initialbound} it is evident that for any $t \in [0,T]$, and any $k \in 
	 \bI_{n-1}^o$ 
	 \begin{equation} \label{initialboundt}
	 \expval{\hat{\tilde{p}}_k^2(nt)}_{\rho^n} \leq c, \quad 
	 \expval{\hat{\tilde{r}}_k^2(nt)}_{\rho^n}\leq c,
	 \end{equation}
		where, as before $c$ is independent of $n$ and the realization of the masses.
			
		Let $f_x:= f(\frac{x}{n})$ and $\tilde{I}(\gamma):= [0,n^{(1-\gamma)}] \cap \bZ$.
		We bound $\mathcal{L'}_n$ defined in \eqref{lowmodelem} as follows: 
	\begin{align} \label{lowmodecomp}
	%\begin{split}	
	&0 \leq \mathcal{L'}_n = \frac{1}{n^2} \sum_{x,y=1}^{n-1} f_xf_y 
	\expval{\tilde{p}_x^o(nt) \tilde{p}_y^o(nt)}_{\rho^n} \leq 
	\frac{C}{n^2} \sum_{x,y=1}^{n}  
	\frac{\Big|\expval{\tilde{p}_x^o(nt) \tilde{p}_y^o(nt)}_{\rho^n}\Big|}{\sqrt{m_x}
	\sqrt{m_y}} \nonumber\\
	&\le\frac{C}{n^2}\sum_{x,y=1}^{n} \frac{1}{\sqrt{m_x m_y}}\expval{(\tilde{p}_x^o(nt))^2}_{\rho^n}^{\frac12}
	\expval{(\tilde{p}_y^o(nt))^2}_{\rho^n}^{\frac12} =C
	\left(\frac{1}{n} \sum_{x=1}^n \frac{1}{\sqrt{m_x}}\expval{(\tilde{p}_x^o(nt))^2}_{\rho^n}^{\frac12}  \right)^2 \nonumber\\  
	&\le\frac{C}{n} \sum_{x=1}^n \frac{\expval{(\tilde{p}_x^o(nt))^2}_{\rho^n}}{m_x}
	=\frac{C}{n} \sum_{x=1}^n \left\langle
	 \sum_{k,k' \in \tilde{I}(\gamma)} \hat{\tilde{p}}_k(nt) \hat{\tilde{p}}_{k'}(nt)
	 \varphi^k_x \varphi^{k'}_x \right\rangle_{\rho^n} \\
	 &\le\frac{C}{n}\sum_{k,k'\in \tilde{I}(\gamma)} 
	 \expval{\hat{\tilde{p}}_k(nt) \hat{\tilde{p}}_{k'}(nt)}_{\rho^n} \sum_{x=1}^n 
	\varphi^k_x \varphi^{k'}_x= \frac{C}{n}\sum_{k \in \tilde{I}(\gamma)} 
	\expval{(\hat{\tilde{p}}_k(nt))^2}_{\rho^n} \nonumber\\ 
	&\le cC\frac{|\tilde{I}(\gamma)|}{n}=cC \frac{n^{(1-\gamma)}}{n} \to 0\nonumber
	%\end{split}		
	\end{align}
	%\begin{equation} \label{lowmodecomp}
	%\begin{split}	
	%0 \leq \mathcal{L}_n = \frac{1}{n^2} \sum_{x,y=1}^{n-1} f_xf_y 
	%\expval{\tilde{r}_x^o(nt) \tilde{r}_y^o(nt)}_{\rho^n} \leq 
	%\frac{C}{n^2} \sum_{x,y=1}^{n}  
	%\Big|\expval{\tilde{r}_x^o(nt) \tilde{r}_y^o(nt)}_{\rho^n}\Big| \leq \\
	%\frac{C}{n^2}\sum_{x,y=1}^{n} \expval{(\tilde{r}_x^o(nt))^2}_{\rho^n}^{\frac12}
	%\expval{(\tilde{r}_y^o(nt))^2}_{\rho^n}^{\frac12} =C
	%\left(\frac{1}{n} \sum_{x=1}^n \expval{(\tilde{r}_x^o(nt))^2}_{\rho^n}^{\frac12}  \right)^2 \leq 
%	\\  \frac{C}{n} \sum_{x=1}^n \expval{(\tilde{r}_x^o(nt))^2}_{\rho^n}
%	=\frac{C}{n} \sum_{x=1}^n \left\langle
%	 \sum_{k,k' \in \tilde{I}(\gamma)} \hat{\tilde{r}}_k(nt) \hat{\tilde{r}}_{k'}(nt)
%	 \phi^k_x \phi^{k'}_x \right\rangle_{\rho^n} \leq \\
%	 \frac{C}{n}\sum_{k,k'\in \tilde{I}(\gamma)} 
%	 \expval{\hat{\tilde{r}}_k(nt) \hat{\tilde{r}}_{k'}(nt)}_{\rho^n} \sum_{x=1}^n 
%	\phi^k_x \phi^{k'}_x= \frac{C}{n}\sum_{k \in \tilde{I}(\gamma)} 
%	\expval{(\hat{\tilde{r}}_k(nt))^2}_{\rho^n} \leq \\ cC\frac{|\tilde{I}(\gamma)|}{n}
%	=cC \frac{n^{(1-\gamma)}}{n} \to 0,  
%	\end{split}		
%		\end{equation}
	as $n \to \infty$, for any realization of the masses; where, in the first line we bounded 
	$|f_xf_y|\sqrt{m_x m_y}$ by 
	$C=||f||_{\infty}^2m_+>0$ which  is independent of $n$ (recall $f$ is 
	continuous on $[0,1]$, and the distribution of the masses is supported in $[m_-,m_+]$),
	 in the second line we take advantage of the Cauchy-Schwartz 
	inequality i.e.\@  $|\expval{AB}_{\rho^n}|^2 \leq \expval{A^2}_{\rho^n}
	\expval{B^2}_{\rho^n}$, in the third line we used another Cauchy-Schwartz inequality 
	($|\langle a,b \rangle_n|^2 \leq |a|^2|b|^2$ for $a,b \in \bR^n$), 
	in the third line, we also
	 used the definition 
	of $\tilde{p}^o_x(nt)$ from \eqref{lowmodedef} and obtained a double sum by squaring 
	this definition, in the fourth line we used the linearity of $\expval{\cdot}_{\rho^n}$
	 as well as the fact that $\varphi^k$ is an orthonormal basis 
	 ($\expval{\varphi^k,\varphi^{k'}}_{n}=\delta_{k,k'}$), finally in the last inequality 
	 we used the second bound in \eqref{initialboundt} ($\expval{\hat{\tilde{p}}_k^2(nt)}_{\rho^n} \leq c$).
	 
	 We can prove $\cL_n \to 0$ deterministically in a similar way with some small modifications:
	  we should use the definition of $\tilde{p}^o$ instead of $\tilde{r}^o$, then use the 
	  fact that $\expval{\phi^k,\phi^{k'}}_{n-1}=\delta_{k,k'}$, and finally take advantage 
	  of the second bound in \eqref{initialbound} ( 
	  $\expval{\hat{\tilde{r}}_k^2(nt)}_{\rho^n} \leq c$).\\
	  
	  \paragraph{Step2: High modes.}
	  Recall $0<\gamma<\frac12$, $I(\gamma)= ]n^{(1-\gamma)},n]\cap \bZ$, and \eqref{lowmodedef}, then define the ``high mode portion''
	  of $p$, and $r$ as follows: 
	  \begin{equation} \label{highmodedef}
			\begin{split}			
			&\tilde{p}_x^{\bullet}(nt):= \tilde{p}_x(nt)-\tilde{p}_x^{o}(nt)=\sum_{k \in I(\gamma)} \hat{\tilde{p}}_k(nt) \sqrt{m_x}\varphi_x^k,\\
			&\tilde{r}_x^{\bullet}(nt):= \tilde{r}_x(nt)-\tilde{r}_x^{o}(nt)=\sum_{k \in I(\gamma)} \hat{\tilde{r}}_k(nt) \phi_x^k.
		\end{split}		
		\end{equation}
 In this step we prove that %for a proper choice of $\gamma$: 
	 \begin{equation} \label{highmodelemrp}
		\mathcal{U}_n' := \left\langle \left( \frac1n \sum_{x=1}^n f(\frac{x}{n})
		\tilde{r}^{\bullet}_x(nt) \right)^2 \right\rangle_{\rho^n} \to 0,  \quad
		\mathcal{U}_n := \left\langle \left( \frac1n \sum_{x=1}^n f(\frac{x}{n})
		\tilde{p}^{\bullet}_x(nt) \right)^2 \right\rangle_{\rho^n} \to 0,
 		\end{equation}
	 	almost surely w.r.t the distribution of the masses.
	 	Take $0<2 \gamma <\theta <1$.
	 	 %(we choose $\theta$, and $\gamma$ opportunely later).
	 	 Expanding the sum and bounding $|f(y)f(y') m_xm_y| \leq 
	 	C=||f||_{\infty}^2m_+$, we have for $n$ sufficiently large:
	 	\begin{equation}\label{highmodedecouplerp}
	 	\begin{split}
	 	&0 \leq  \mathcal{U}_n \leq \frac{C}{n^2} \sum_{x,y=1}^n \left|
	 	\frac{\expval{\hmp_x(nt) \hmp_y(nt)}_{\rho^n}}{\sqrt{m_x m_y}}
	 	 \right|= \\ &C \left( \underbrace{\frac1{n^2} \sum_{|x-y|\leq 2n^{\theta}}  \left|
	 	\frac{\expval{\hmp_x(nt) \hmp_y(nt)}_{\rho^n}}{\sqrt{m_x m_y}} \right|}_{:=\mathcal{U}_n^<} +
	 	\underbrace{\frac1{n^2}\sum_{|x-y|> 2n^{\theta}}  \left|
	 	\frac{\expval{\hmp_x(nt) \hmp_y(nt)}_{\rho^n}}{\sqrt{m_x m_y}}\right|}_{:=\mathcal{U}_n^>}
		\right)	 	
	 	\end{split}
	 	\end{equation}
		
	\paragraph{Step 2.1: High modes I ($|x-y| >2 n^{\theta}$).} %\label{step11}
	We defined $\mathcal{U}_n^{<}$, $\mathcal{U}_n^{>}$ in 
	\eqref{highmodedecouplerp} implicitly. First, we prove 
	$\mathcal{U}_n^{>} \to 0$, almost surely w.r.t the distribution of the masses.
	Before proceeding, recall the definition of 
	$\hat{\tilde{p}}_k(0)\equiv \hat{\tilde{p}}_k$, and $\hat{\tilde{r}}_k(0) \equiv
	 \hat{\tilde{r}}_k$ from \eqref{tildenitial}. 
	 By using \eqref{cov1}, observe that for proper $k,k'$ (we use the shorthand notation $\beta^{-1}\big(\frac{x}{n} \big)=: b_x$):
	\begin{equation} \label{covkk'}
	\begin{split}
		&\expval{\hat{\tilde{p}}_k \hat{\tilde{p}}_{k'}}_{\rho^n}=\sum_{z,z'} 
		\frac{\varphi^k_z \varphi^{k'}_{z'}}{\sqrt{m_z m_{z'}}}	\expval{\tilde{p}_z 
		\tilde{p}_{z'}}_{\rho^n} = \sum_{z=1}^n b_z \varphi^k_z \varphi^{k'}_{z}, \quad
		\expval{\hat{\tilde{r}}_k \hat{\tilde{r}}_{k'}}_{\rho^n}= \sum_{z=1}^{n-1}
		b_z \phi^k_z \phi^{k'}_{z} .
	\end{split}
	\end{equation}	 
	Thanks to \eqref{highmodedef} we have: 
	\begin{equation} \label{MAINdecomp}
	\begin{split}
		&\mathcal{U}_n^>= \frac1{n^2} \sum_{|x-y|>2n^{\theta}}\left|\sum_{k,k' \in 
		I(\gamma)} \expval{\hat{\tilde{p}}_k(nt) \hat{\tilde{p}}_{k'}(nt) }_{\rho^n}
		\varphi^k_x \varphi^{k'}_y  \right|	=	\\
		&\frac1{n^2} \sum_{|x-y|>2n^{\theta}}\left|\sum_{k,k' \in 
		I(\gamma)} \varphi^k_x \varphi^{k'}_y
		\left( \expval{\hat{\tilde{p}}_k\hat{\tilde{p}}_{k'}}_{\rho^n} \cos(\omega_knt)
		\cos(\omega_{k'}nt)+  \expval{\hat{\tilde{r}}_k\hat{\tilde{r}}_{k'}}_{\rho^n} 
		\sin(\omega_knt)
		\sin(\omega_{k'}nt)  \right)\right| \leq \\
		& \frac{1}{n^2} \sum_{|x-y|>2n^{\theta}} 
		\left|\sum_{z=1}^n b_z 
		\left(\sum_{k\in I(\gamma)} \varphi^k_x\varphi^k_z \cos(\omega_k nt) \right)
		\left(\sum_{k \in I(\gamma)} \varphi^{k'}_y \varphi^{k'}_z \cos(\omega_k nt) \right)
		 \right| + \\
		 &	\frac{1}{n^2} \sum_{|x-y|>2n^{\theta}} 
		\left|\sum_{z=1}^{n-1} b_z 
		\left(\sum_{k\in I(\gamma)} \varphi^k_x\phi^k_z \sin(\omega_k nt) \right)
		\left(\sum_{k \in I(\gamma)} \varphi^{k'}_y \phi^{k'}_z \sin(\omega_k nt) \right)
		 \right| \leq  \\
		 &\underbrace{\frac{C'}{n^2} \sum_{|x-y|>2n^{\theta}}\sum_{z=1}^n 
		 \left(\sum_{k\in I(\gamma)} |\varphi^k_x||\varphi^{k}_z|\right)\left(
		 \sum_{k'\in I(\gamma)} |\varphi^{k'}_y||\varphi^{k'}_z|\right)}_
		 {:= \mathcal{U}_n^{>,p}}+ \\ 	
		 &\underbrace{\frac{C'}{n^2} \sum_{|x-y|>2n^{\theta}}\sum_{z=1}^n 
		 \left(\sum_{k\in I(\gamma)} |\varphi^k_x||\phi^{k}_z|\right)\left(
		 \sum_{k'\in I(\gamma)} |\varphi^{k'}_y||\phi^{k'}_z|\right)}_
		 {:= \mathcal{U}_n^{>,r}},
	\end{split}
	\end{equation}
	where in the second line we used the definition of $\hat{\tilde{p}}_k(nt)$ 
	\eqref{tildesol},
	and the 
	fact that $\expval{\hat{\tilde{r}}_k\hat{\tilde{p}}_{k'}}_{\rho^n}=0$
	 \eqref{initialbound}, in the third line we take advantage of \eqref{covkk'}, and in the 
	 fourth line we bounded $|\sin(\cdot)|$, $|\cos(\cdot)| \leq 1$, and $|b_z|=\beta^{-1}
	 (\frac{z}{n}) \leq C'=\beta_-^{-1}$ uniform in $n$. 
	 
	 For any $x \in \bI_n$, denote	$J(x):= \{z \in \bI_n \big| \: |x-z| \leq n^{\theta} \}
	 $. Notice the definition of $\mathcal{U}_n^{<,p}$, $\mathcal{U}_n^{>,p}$ from 
	 the last lines of \eqref{MAINdecomp}. Take a single term from the main sum 
	 featuring $\mathcal{U}_n^{<,p}$ and observe for 
	 $n$ sufficiently large (recall $\bE$ is the expectation w.r.t 
	 the distribution of the masses and recall that here $|x-y|>n^{2\theta}$):  
	 \begin{equation} \label{eq:4}
	\begin{split}	 
	& \bE\left(\sum_{z=1}^n 
		 \left(\sum_{k\in I(\gamma)} |\varphi^k_x||\varphi^{k}_z|\right)\left(
		 \sum_{k'\in I(\gamma)} |\varphi^{k'}_y||\varphi^{k'}_z|\right) \right) =
	  \sum_{z \in J(x)} \bE\left( 
		\left(\sum_{k\in I(\gamma)} |\varphi^k_x||\varphi^{k}_z|\right)\left(
		 \sum_{k'\in I(\gamma)} |\varphi^{k'}_y||\varphi^{k'}_z|\right)\right)+ \\
		 &\sum_{z \notin J(x) } \bE\left(
		 \left(\sum_{k\in I(\gamma)} |\varphi^k_x||\varphi^{k}_z|\right)\left(
		 \sum_{k'\in I(\gamma)} |\varphi^{k'}_y||\varphi^{k'}_z| \right)\right)
		\leq \sum_{z\in J(x)} \bE \left(\sum_{k' \in I(\gamma)} |\varphi^{k'}_y 
		\varphi^{k'}_z| \right) + \\ & \sum_{z \notin J(x)} \bE\left(
		\sum_{k\in I(\gamma)} |\varphi^k_x\varphi^{k}_z|\right) \leq
		C \sum_{z \in J(x)} \exp(-c|y-z|/n^{2 \gamma}) + 
		C\sum_{z \notin J(x)} \exp(-c|x-z|/n^{2 \gamma}) \leq \\
		& 4Cn^{\theta} \exp(-c(n^\theta-n^{2\gamma})) + Cn \exp(-c(n^\theta-n^{2\gamma}))
		\leq 4C \exp(-\frac{c}{4}(n^{\theta})),
	\end{split} 		
	 \end{equation}
	where in the first inequality we bounded: 
	\begin{equation} \label{eq:5}
	\sum_{k \in I(\gamma)} |\varphi^k_x||\varphi^{k}_z| \leq
	\left(\sum_{k \in I(\gamma)}|\varphi^k_x|^2 
	\right)^{\frac12}  \left(\sum_{k \in I(\gamma)} 
	|\varphi^k_z|^2 \right)^{\frac12} 
	 \leq \left(\sum_{k=0}^{n-1}|\varphi^k_x|^2 
	\right)^{\frac12}  \left(\sum_{k=0}^{n-1} 
	|\varphi^k_z|^2 \right)^{\frac12} = 1,
	 \end{equation}
	thanks to a Cauchy-Schwartz inequality and the fact that $\varphi^k$ is  an orthonormal
	basis. We bounded $\sum_{k' \in I(\gamma)}|\varphi^{k'}_y||\varphi^{k'}_z| \leq 1$ 
	similarly. The second inequality in \eqref{eq:4} is deduced from Lemma 
	\ref{localizationlemma} (modification of \eqref{localizationestimate} by a constant). 
	Notice that this is the crucial place where we use localization. The fourth inequality 
	in \eqref{eq:4} is deduced from the definition of $J(x)$ and the fact that 
	$|x-y|>2n^{\theta}$: more precisely, for $z \notin J(x)$, $|x-z|>n^{\theta}$, and for 
	$z \in J(x)$ (i.e. $|z-x|\leq n^{\theta}|$), 
	we have $|y-z|>n^{\theta}$ thanks to the choice of $y$: $|x-y|> 
	2n^{\theta}$. In the last inequality we use the choice of $\theta$: 
	$0<2 \gamma < \theta<1$, and that $n$ is sufficiently large. 
	
	Recalling the definition of $\cU_n^{<,p}$ \eqref{MAINdecomp}, thanks to the latter 
	estimate \eqref{eq:4} there exists $c_1,C>0$ independent of $n$ such that for 
	$n$ sufficiently large we have:
	\begin{equation} \label{eq:6}
	\bE\Big(\big|\cU_n^{>,p}\big|)\Big) \leq C \exp(-c_1 n^{\theta}) \implies
	\sum_{n} \bE(\big|\cU_n^{>,p}\big|) < \infty \implies \cU_n^{>,p} \to 0,
	\end{equation}
	almost surely w.r.t the distribution of the masses, 
	where we used the fact that $\theta>0$, which means that the 
	above series is summable, and 
	finally we conclude thanks to Borel-Cantelli's lemma. 
	
	Recall $\cU_n^{>,r}$  from \eqref{MAINdecomp}. This expression can be treated similar to 
	$\cU_n^{>,p}$ thanks to Lemma \ref{boundomegalemma}: 
	First, recall for $z \in \bI_{n-1}$,
	 $\phi^k_z= \frac{1}{\omega_k}(\frac{\varphi^k_{z+1}}{m_{z+1}}-\frac{\varphi^k_z}
	 {\sqrt{m_z}})$. Let us denote $\tilde{\omega}:=\omega_{k_*}$ with $k_*:= \min\{ 
	 I(\gamma) \}$. For any $n$ define the sequence $a_n:= \tilde{\omega} n^{3\frac{\gamma}
	 2}$. Since $\tilde{\omega} \leq \omega_k$ for any $k \in I(\omega)$ we have
	 $C,c>0$ independent of $n$ such that:
	 \begin{equation} \label{eq:7}
	 	\bE\left(a_n \sum_{k \in I(\gamma)} |\varphi^k_x| |\phi^k_z|\right)
	 		 	\leq n^{\frac{3\gamma}{2}} 2C \exp(-\frac{c}{2}|z-x|/n^{2\gamma}),
	 \end{equation}
	where we used the definition of $\phi^k$, and estimate \eqref{localizationestimate}.
		
	Having above estimate, by doing a computation similar to \eqref{eq:4}, and \eqref{eq:5}
	 we can deduce for $n$ sufficiently large there exist $c',C>0$ independent of $n$:
	 \begin{equation} \label{eq:71}
	\bE(|a_n \cU_n^{>,r}|) \leq Cn^{\frac{3\gamma}{2}} \exp(-c' n^{\theta}) \leq 
	C \exp(-\frac{c'}{2}n^{\theta}).   
	\end{equation}	  
	  	Therefore, by Borel-Cantelli's lemma we have $a_n \cU_n^{>,r} \to 0 $, 
	  	almost surely
	 w.r.t the distribution of the masses. However, thanks to Lemma \ref{boundomegalemma}
	  (i.e. \eqref{boundomega}), 
	  we know there exists almost surely $n_0$, such that for $n> n_0$, we have 
	 $a_n \geq \frac{1}{c}$ with $c>0$. Using these two facts we deuce that 
	 $\cU_n^{>,r} \to 0 $, almost surely w.r.t distribution of the masses.
	   Combining $\cU_n^{>,r} \to 0 $, with 
	 \eqref{eq:6} ($\cU_n^{>,p}\to 0$) we get $\cU_n^> \to 0$ almost surely w.r.t the  
	 distribution of the masses thanks to \eqref{MAINdecomp}.
	 
	 Recall the definition of $\cU_n'$ \eqref{highmodelemrp}. We can define $\cU_n^{'<}$
	 , $\cU_n^{'>}$ corresponding to $r$ variable similar to $\cU_n^>$, $\cU_n^<$ as in 
	 \eqref{highmodedecouplerp}: 
	 \begin{equation} \label{highdecouple2}
	 	\cU_n^{'<} := \frac{1}{n^2} \sum_{|x-y|\leq2n^\theta} \left|\expval{\hmr_x(nt) 
	 	\hmr_y(nt)}_{\rho^n} \right|, \quad 
	 	\cU_n^{'>} := \frac{1}{n^2} \sum_{|x-y|>2n^\theta} \left|\expval{\hmr_x(nt) 
	 	\hmr_y(nt)}_{\rho^n} \right|.
	\end{equation}	  
	 Then we can prove $\cU_n^{'>} \to 0$ almost surely, with respect to the 
	 distribution of the masses as we did for $\cU_n^{>}$. We only sketch the former, since 
	 these proofs are almost identical. 
	 First, one can compute $\cU_n^{'>}$ similar to 
	 \eqref{MAINdecomp} and decompose the $r$ and $p$ contribution and define 
	 $\cU_n^{'>,r}, \cU_n^{'>,p}$ similar to \eqref{MAINdecomp}. Then one can 
	 observe that $\cU_n^{'>,p}=\cU_n^{>,r} \to 0$ almost surely. Proving  
	 $\cU_n^{'>,r} \to 0$, almost surely, is similar to the previous terms: First 
	 we decompose the sum in $\cU_n^{'>,r}$ into two parts 
	 similar to \eqref{eq:4}, then we use proper estimates. The only 
	 difference lies in the following estimate:  
	 Define $b_n =\tilde{\omega}^2 n^{3 \gamma}$ (recall $\tilde{\omega}= \omega_{k_*} $ 
	 with $k_*= \min \{ I(\gamma) \}$). Then similar to \eqref{eq:7}, one can observe 
	 that there exist $c,C$ such that 
	 \begin{equation}\label{eq:8}
	 \bE\left(b_n \sum_{k \in I(\gamma)} |\phi^k_x \phi^k_z| \right) \leq n^{3 \gamma} C \exp(-c|x-z|/
	 n^{2\gamma}). 
	\end{equation}	
	Consequently, one can deduce that $b_n \cU_n^{'>,r} \to 0$ almost surely, as we did 
	for $a_n\cU_n^{>,r}$ in \eqref{eq:71}. Therefore, thanks to Lemma \ref{boundomega} we 
	can deduce there exists almost surely $n_0$ such that for $n>n_0$, we have  
	$b_n\geq \frac1{c^2}$. Hence,  we have  $\cU_n^{'>,r} \to 0$ almost surely. This 
	finishes the proof of the fact that $\cU_n^{'>} \to 0$ almost surely w.r.t the 
	distribution of the masses.
	
	\paragraph{Step 2.2: High modes II ($|x-y| \leq 2n^{\theta}$).}
	In this step, we prove $\cU_n^<$ and $\cU_n^{'<}$ defined in \eqref{highmodedecouplerp} and \eqref{highdecouple2} respectively, converge to zero almost surely. 
	Similarly to the computation in \eqref{eq:4}, thanks to the definition of $\hmp_x(nt)$ \eqref{highmodedef}, 
	 and identities \eqref{initialbound}, \eqref{covkk'}, we have for any 
	 $x \in \bI_n$ (recall
	 $I(\gamma)=]n^{1-\gamma},n] \cap \bZ$, $b_z=\beta^{-1}(\frac{z}{n})$): 
	 \begin{equation} \label{eq:10}
		\begin{split}		
		&\expval{\frac{(\hmp_x(nt))^2}{m_x}}_{\rho^n} =\sum_{k,k' \in I(\gamma)}	
			\varphi^k_x\varphi^{k'}_x \expval{\hat{\tilde{p}}_k(nt) 
			\hat{\tilde{p}}_{k'}(nt)}_{\rho^n} = \\
			&\sum_{z=1}^n b_z\left(
			\sum_{k \in I(\gamma)} \cos(\omega_k nt)\varphi^k_x \varphi^k_z  \right)^2+
			\sum_{z=1}^{n-1} b_z\left(
			\sum_{k \in I(\gamma)} \sin(\omega_k nt)\varphi^k_x \phi^k_z  \right)^2 \leq \\
		& C \sum_{z=1}^n \left(
			\sum_{k \in I(\gamma)} \cos(\omega_k nt)\varphi^k_x \varphi^k_z  \right)^2+
		C \sum_{z=1}^{n-1} \left(
			\sum_{k \in I(\gamma)} \sin(\omega_k nt)\varphi^k_x \phi^k_z  \right)^2 =\\ 
			&C \sum_{k,k' \in I(\gamma)}  \cos(\omega_k nt) \cos(\omega_{k'} nt) 
			\varphi^k_x \varphi^{k'}_x \sum_{z=1}^n \varphi^k_z \varphi^{k'}_z + 
			C \sum_{k,k' \in I(\gamma)}  \sin(\omega_k nt) \sin(\omega_{k'} nt) 
			\varphi^k_x \varphi^{k'}_x \sum_{z=1}^{n-1} \phi^k_z \phi^{k'}_z = \\
			&C\sum_{k \in I(\gamma)} (\cos^2(\omega_k nt)+\sin^2(\omega_k nt)) 
			(\varphi^k_x)^2 \leq C\sum_{k=0}^{n-1} (\varphi^k_x)^2 =C,
		\end{split}	
	 \end{equation}
	 where in the first inequality we bounded $b_z$ by $C=\beta_-^{-1}$ (which is 
	 independent of $n$ and the realization of the masses). Notice that in 
	 the first inequality, we used the fact that all the terms in the sum are positive 
	 (i.e. $b_z>0$). In the last two lines we used the fact that $\varphi^k$, $\phi^k$ 
	 are orthonormal basis ($\expval{\varphi^k,\varphi^{k'}}_n =\delta_{k,k'}$, 
	 $\expval{\phi^k,\phi^{k'}}_{n-1} =\delta_{k,k'}$, $\sum_z \varphi^k_z \varphi^{k'}_z=
	 \sum_{z}\phi^k_z\phi^{k'}_z =\delta_{k,k'} $).  
	 
	 Recall the definition of $\cU_n^<$ \eqref{highmodedecouplerp} with $\theta<1$,
	 then thanks to 
	 \eqref{eq:10} we have: 
	 \begin{equation} \label{eq:11}
	 \begin{split}
	 &\cU_n^< = \frac{1}{n^2} \sum_{|x-y| \leq 2 n^{\theta} } 
	 \left| \expval{\frac{\hmp_x(nt) \hmp_y(nt)}{\sqrt{m_xm_y}}}_{\rho^n} \right|
	 \leq \frac{1}{n^2} \sum_{|x-y|\leq 
	 2n^{\theta}} \expval{\frac{(\hmp_x(nt))^2}{m_x}}_{\rho^n}^{\frac12}
	 \expval{\frac{(\hmp_y(nt))^2}{m_y}}_{\rho^n}^{\frac12} \leq \\
	 & \frac{1}{n^2} \sum_{|x-y<2n^{\theta}|} C \leq 8C \frac{n^{1+ \theta}}{n^2} \to 0
	 ,\end{split}
	 \end{equation}
	 where in the first line we bounded $\cU_n^{<}$ by Cauchy-Schwartz inequality, 
	 and in the second line we used \eqref{eq:10}. \
	 Recall the decomposition $\cU_n \leq C(\cU_n^<+ \cU_n^>)$ \eqref{highmodedecouplerp}. 
	 In Step 2.1, we observed that $\cU_n^> \to 0$, almost surely. 
	 Combining this fact with \eqref{eq:11}, we deduce that $\cU_n \to 0$, almost surely 
	 w.r.t the distribution of the masses. 
	 
	 Recall the definition of $\cU_n^{'<}$ \eqref{highdecouple2}. Exactly similar to 
	 \eqref{eq:10},
	 one can bound $\expval{(\hmr_x(nt))^2}_{\rho^n}$ thanks to its definition 
	 \eqref{highmodedef}, and averages \eqref{initialbound}, \eqref{covkk'}: for any 
	 $x \in \bI_n$  there exists $\beta^{-1}_- =C>0$ 
	 independent of $n$ and realization of the masses such that 
	 \begin{equation}\label{eq:12}
	 \expval{(\hmr_x(nt))^2}_{\rho^n} \leq C.
	 \end{equation}	   
		Similar to \eqref{eq:11}, by using a Cauchy-Schwartz inequality and the bound 
		\eqref{eq:12} we deduce, that $\cU_n^{'<} \to 0$. Having the latter, we also 
	proved $\cU_n^{'>} \to 0$, almost surely in the previous step. Therefore, thanks to the
	bound $\cU_n' \leq C (\cU_n^{'>} +\cU_n^{'<}  )$ we conclude $\cU_n' \to 0$ almost 
	surely w.r.t distribution of the masses. This finishes the proof of 
	\eqref{highmodelemrp}.  
	
	\paragraph{Step 3: Summing up.}
	Recall the definition of $\cP_n$, $\cL'_n$, and $\cU_n$ \eqref{eq:proof3}, 
	\eqref{lowmodelem}, and \eqref{highmodelemrp}. Since $\tilde{p}_x(nt)= 
	\tilde{p}_x^o(nt)+\hmp_x(nt)$ (cf. \eqref{highmodedef}, \eqref{lowmodedef}), one can 
	observe that $\cP_n \leq \cL'_n +\cU_n + (\cL_n^{'})^{\frac12}(\cU_n)^{\frac12}$ thanks
	 to a Cauchy-Schwartz inequality  ($|\expval{AB}_{\rho^n}| \leq \expval{A^2}_{\rho^n}
	 ^{\frac12}\expval{B^2}_{\rho^n}^{\frac12}$). Hence we have $\cP_n \to 0$, almost surely 
	 w.r.t the distribution of the masses, since $\cL'_n \to 0$, and $\cU_n \to 0$ almost 
	 surely
	 thanks to \eqref{lowmodelem}, and \eqref{highmodelemrp}. We can deduce $\cR_n \to 0$
	 \eqref{eq:proof3} almost surely similarly, by using  the fact that 
	 $\tilde{r}_x(nt)= \tilde{r}_x^o(nt)+\hmr_x(nt)$,  and $\cL_n \to 0$, $\cU'_n \to 0$,
	 almost surely. This finishes the proof of \eqref{eq:proof3}, and we conclude first two 
	 limit of Theorem \ref{thmhighermomentclassical}: \eqref{highermomentclassicalr}, and  
	 \eqref{highermomentclassicalp}.  
	  \end{proof}
		
		\begin{rem}
			 The bound \eqref{eq:10} is true for $\tilde{p}_x$ as well. 
			Moreover, similar to \eqref{eq:10}, we can bound $\expval{\tilde{p}_x^2(nt)/m_x}_{\rho^n}$ by $\beta_+^{-1}$ from below. This means the temperature at each point is bounded 
			by the maximum and minimum of the temperature profile at each time. 
		\end{rem}
	\begin{rem} \label{fbdd}
	Notice that in the above proof, for proving $\cP_n \to 0$, $\cR_n \to 0$
	 almost surely w.r.t distribution of the masses, \eqref{eq:proof3}, we only used the fact 
	that $f$ is bounded. Therefore,it is straightforward to observe that
	 \eqref{eq:proof3} holds for bounded $f$ as well.
	\end{rem}

	\subsection{Evolution of $e$}
	In this section, we prove \eqref{highermomentclassicale}.  Again, our proof rests on the limit on average \eqref{eq: limit on average}.
	\begin{proof}[Proof of Theorem \ref{thmhighermomentclassical}: \eqref{highermomentclassicale}] 
	Similarly to \eqref{tildevariable}, let us define: 
	\begin{equation} \label{tildee}
	\tilde{e}_x(nt):=e_x(nt)-\expval{e_x(nt)}_{\rho^n} =\frac12 
	\left(\frac{p_x^2(nt)}{m_x}+r_x^2(nt) -\frac{\expval{p_x^2(nt)}_{\rho^n}}{m_x}-
	\expval{r_x^2(nt)}_{\rho^n} \right).
	\end{equation}
	Similarly to \eqref{eq:proof3}, one can add and subtract $\frac1n\sum_{x=1}^n 
	f(\frac{x}{n})\expval{e_x}_{\rho^n}$ inside the square in 
	\eqref{highermomentclassicale}.
	Then using the limit on average \eqref{eq: limit on average}, it is clear that in order to obtain 
	\eqref{highermomentclassicale} it is sufficient to prove
	\begin{equation} \label{eq:proof4}
	\cE_n :=\left\langle \left(\frac1n \sum_{x=1}^n f(\frac{x}{n})\tilde{e}_x(nt) \right)^2
	\right\rangle_{\rho^n} \to 0,
	\end{equation}
	almost surely w.r.t the distribution of the masses. Plugging the definition of 
	\eqref{tildee} into  \eqref{eq:proof4}, and taking advantage of Cauchy-Schwartz 
	inequality ($|\expval{AB}_{\rho^n}| \leq  
	\expval{A^2}_{\rho^n}^{\frac12} \expval{B^2}_{\rho^n}^{\frac12}$), we observe 
	that for deducing \eqref{eq:proof4} it is enough to show: 
	\begin{equation} \label{eq:proof5}
	\begin{split}	
	\frP_n:= \left\langle \left(\frac1n \sum_{x=1}^n f(\frac{x}{n})
	\left(\frac{p_x^2(nt)}{m_x} - \expval{\frac{p_x^2(nt)}{m_x}}_{\rho^n}\right) \right)^2
	\right\rangle_{\rho^n} \to 0, \\
	\frR_n:= \left\langle \left(\frac1n \sum_{x=1}^n f(\frac{x}{n})
	\left(r_x^2(nt) - \expval{r_x^2(nt)}_{\rho^n}\right) \right)^2
	\right\rangle_{\rho^n} \to 0,
	\end{split}
	\end{equation}	  
	almost surely w.r.t the distribution of the masses. In the following, we prove  
	$\frP_n \to 0$ almost surely. The proof of $\frR_n \to 0$ goes in parallel, and we omit
	it for the sake of brevity. \\
	Let us denote $$\bar{p}_x(nt):=\expval{p_x(nt)}_{\rho^n},$$ and
	recall the definition of $\tilde{p}_x(nt)= p_x(nt)-\bar{p}_x(nt)$ \eqref{tildevariable}.
	Since $\expval{\tilde{p}_x(nt)}_{\rho^n}=0$, 
	after a straightforward computation, we have
	for $x,y \in \bI_n$:
	\begin{equation} \label{eq:13}
	\begin{split}	
	&\expval{\left(p_x^2(nt)- \expval{p_x^2(nt)}_{\rho^n}\right)\left(
	p_y^2(nt)-\expval{p_y^2(nt)}_{\rho^n}\right)}_{\rho^n}=
	\\ &\expval{\tilde{p}^2_x(nt)
	\tilde{p}^2_y(nt)}_{\rho^n}- \expval{\tilde{p}^2_x(nt)}_{\rho^n}
	\expval{\tilde{p}^2_y(nt)}_{\rho^n} +4\bar{p}_x(nt)\bar{p}_y(nt) 
	\expval{\tilde{p}_x(nt)\tilde{p}_y(nt)}_{\rho^n}+ \\
	&2 \bar{p}_x(nt) \expval{\tilde{p}_x(nt) \tilde{p}_y^2(nt)}_{\rho^n} + 
	2\bar{p}_y(nt) \expval{\tilde{p}^2_x(nt)
	\tilde{p}_y(nt)}_{\rho^n}.
	\end{split}
	\end{equation}
	Notice that owing to the definition of $\rho^n$ \eqref{localgibbsstate}, for any 
	realization of the masses $(p(0),r(0))$ is a Gaussian random vector, where the randomness refers to the probability distribution $\rho^n$.
	% (normal random vector).\footnote{Here ``random'' refers to the probability distribution $\rho^n$ and not randomness of the masses.} 	
	In addition, the evolution equation \eqref{eqofmotion} is linear,  and the vector $(p(nt),r(nt))$
	can be written as $(p(nt),r(nt))^{\dagger}=O(t) (p(0),r(0))^{\dagger}$ for a linear 
	transformation $O(t)$, as we observed in \eqref{soleqmotion}. Therefore, the vector 
	$(p(nt),r(nt)$ is another Gaussian vector. Consequently, 
	$(\tilde{p}(nt),\tilde{r}(nt))$ is a zero mean Gaussian vector (equivalently 
	if $\rho^n(t)$ denotes the solution to the backward equation,  $\rho^n(t)$ will be 
	Gaussian). Hence we can use Isserlis' theorem (Wick's theorem cf. \cite{Wick}) and observe that
	any odd moments of $\tilde{p}$, $\tilde{r}$ is zero. Moreover, 
	for any even moment we have pairing. More precisely, 
	for $k$ odd $\expval{\tilde{p}_{x_1}(nt) \dots \tilde{p}_{x_k}(nt)}_{\rho^n}=0$, 
	for any $x_1,\dots, x_k \in \bI_k$. For even $k$ we have   
	\begin{equation}\label{pairing}
	\begin{split}	
	&\expval{\tilde{p}_{x_1}(nt) \dots \tilde{p}_{x_k}(nt)}_{\rho^n}=\sum_{\pi \in \Pi}
	\prod_{\{i,j\} \in \pi} \expval{\tilde{p}_{x_i}(nt)\tilde{p}_{x_j}(nt)}_{\rho^n} 
	\implies\\
	& \expval{\tilde{p}^2_x(nt)
	\tilde{p}^2_y(nt)}_{\rho^n}- \expval{\tilde{p}^2_x(nt)}_{\rho^n}
	\expval{\tilde{p}^2_y(nt)}_{\rho^n} = 2 \expval{\tilde{p}_x(nt) \tilde{p}_y(nt)}_{\rho^n}^2; \quad \expval{\tilde{p}_x(nt) \tilde{p}_y^2(nt)}_{\rho^n}=0 
	\end{split}	
	\end{equation}  
  where $\pi$ denotes a  
	pairing of $\bI_k$ i.e. a partition of $\bI_n$ into pairs $\{i,j \}$, 
	and $\Pi$ denotes the set of all these pairings; the product here is over all the pairs
	inside a pairing.
	
	As a result of \eqref{pairing}, we rewrite \eqref{eq:13} as: 
	\begin{multline} \label{eq:14}
	\expval{(p_x^2(nt)- \expval{p_x^2(nt)}_{\rho^n})(
	p_y^2(nt)-\expval{p_y^2(nt)}_{\rho^n})}_{\rho} \\
	= 2 \expval{\tilde{p}_x(nt) \tilde{p}_y(nt)}_{\rho^n}^2+ 4\bar{p}_x(nt)\bar{p}_y(nt) 
	\expval{\tilde{p}_x(nt)\tilde{p}_y(nt)}_{\rho^n}.
	\end{multline}
	Before proceeding, let us recall the estimate (3.6) from \cite{BHO}:
	There exists a constant $C>0$ independent of $n$ and the realization of the masses such that, for any $n$ and $x \in \bI_n$ :
	\begin{equation}\label{BHObound}
		\left|\expval{p_x(nt)}_{\rho^n}\right| \leq C, \quad 
		\left|\expval{r_x(nt)}_{\rho^n}\right| \leq C.
	\end{equation}	  	
	Furthermore, following the exact same line of \eqref{eq:10} (with only replacing
	$\hmp_x(nt)$ with $\tilde{p}_x(nt)$ and summing over $\bI_n^o$ instead of $I(\gamma)$),
	 it is evident that there exist $C>0$ independent of $n$ and realization 
	 of the masses such that for any $n$ and $x,y \in \bI_n$, we have (
	 in fact one can take $C= \beta^{-1}_-$): 
	 \begin{equation} \label{thermalbound}
	 \expval{\frac{\tilde{p}_x^2(nt)}{m_x}}_{\rho^n} \leq C, \quad 
	 \expval{\tilde{r}_x^2(nt)}_{\rho^n} \leq C \implies 
	 \left|\frac{\expval{\tilde{p}_x(nt)\tilde{p}_y(nt)}_{\rho^n}}{\sqrt{m_xm_y}}\right|
	 \leq C,
	\end{equation}	    
	where we used a Cauchy-Schwartz inequality.
		
	Recall $\frP_n$ from \eqref{eq:proof5}. Expanding the square, and using the identity 
	\eqref{eq:14} we have, with $f_x:=f(\frac{x}{n})$,
	\begin{equation} \label{eq:15}
	\begin{split}	
	&0 \leq \frP_n =\frac{2}{n^2} \sum_{x,y=1}^n f_xf_y \left(\frac{
	\expval{\tilde{p}_x(nt)\tilde{p}_y(nt)}_{\rho^n}^2}{m_xm_y}+ 2 
	\frac{\bar{p}_x(nt)\bar{p}_y(nt)\expval{\tilde{p}_x(nt)\tilde{p}_y(nt)}_{\rho^n}}{m_xm_y} \right) \leq \\
	&\frac{2c_1}{n^2} \sum_{x,y=1}^n \left|\frac{\expval{\tilde{p}_x(nt)\tilde{p}_y(nt)}_{\rho^n}}{\sqrt{m_xm_y}} \right|\left(\left|\frac{\expval{\tilde{p}_x(nt)\tilde{p}_y(nt)}_{\rho^n}}{\sqrt{m_xm_y}}\right|+ 2\left|\frac{\bar{p}_x(nt) \bar{p}_y(nt)}{\sqrt{m_xm_y}} \right| \right)	\leq 
	\underbrace{\frac{c_2}{n^2} \sum_{x,y=1}^n \left|\frac{\expval{\tilde{p}_x(nt)\tilde{p}_y(nt)}_{\rho^n}}{\sqrt{m_xm_y}} \right|}_{=:\frP_n'}, 
	\end{split}
	\end{equation}	 
	where $c_1$, $c_2$ are independent of $n$ and realization of the masses. 
	In the first inequality we bounded $|f(y)f(y')| \leq c_1$, since $f$ is bounded,
	 and in the second 
	inequality we take advantage of \eqref{thermalbound} to bound the first term, 
	and \eqref{BHObound} as well as the fact that distribution of the masses is compactly 
	supported in $[m_-,m_+]$ to bound the second term. 
	
	Gathering a couple of results from Sec.~\ref{sec:rp}, it is straightforward to observe 
	that $\frP_n' \to 0$ almost surely, where $\frP'_n$ is defined in the last line of
	\eqref{eq:15}: Recall $\hmp_x(nt)$, $\tilde{p}_x^o(nt)$ \eqref{highmodedef}, \eqref{lowmodedef}, where $\tilde{p}_x(nt)= \hmp_x(nt)+
	\tilde{p}_x^o(nt)$ then we have: 
	\begin{equation} \label{eq:1511}
	\begin{split}
		0 \leq \frP_n' \leq  \underbrace{\frac{c_2}{n^2} \sum_{x,y=1}^n \left|\frac{\expval{\tilde{p}^o_x(nt)\tilde{p}^o_y(nt)}_{\rho^n}}{\sqrt{m_xm_y}} \right|}_{
		:=\frP_n^1} + 
		\underbrace{\frac{c_2}{n^2} \sum_{x,y=1}^n\left|\frac{\expval{\tilde{p}^{\bullet}_x(nt)\tilde{p}^{\bullet}_y(nt)}_{\rho^n}}{\sqrt{m_xm_y}} \right|  }_{:=\frP_n^2} + \\
	 \underbrace{\frac{c_2}{n^2} \sum_{x,y=1}^n\left(\left|\frac{\expval{\tilde{p}^{
	 \bullet}_x(nt)\tilde{p}^o_y(nt)}_{\rho^n}}{\sqrt{m_xm_y}} \right| + 
		\left|\frac{\expval{\tilde{p}^{o}_x(nt)\tilde{p}^{\bullet}_y(nt)}_{\rho^n}}{\sqrt{m_xm_y}} \right| \right) }_{:=\frP_n^3}	
	\end{split}
	\end{equation}
	Observe the definition of $\frP_n^1, \frP_n^2, \frP_n^3$ in the above expression.
	 $\frP_n^1 \to 0$ as we proved in \eqref{lowmodecomp}. $\frP_n^2 \leq c'(\cU_n^>+
	\cU_n^<)$ thanks to \eqref{highmodedecouplerp}. However, we observed that 
	$\cU_n^> \to 0 $ almost surely w.r.t the distribution of the masses after \eqref{eq:71}. 
	Moreover, we showed $\cU_n^< \to 0$ in \eqref{eq:11}. Hence, 
	 $\frP_n^2 \to 0$ almost surely w.r.t the distribution of the masses. Therefore, it only
	 remains to prove that $\frP_n^3 \to 0$. Thanks to Cauchy-Schwartz inequality 
	 ($|\expval{AB}|\leq \expval{A^2}^{\frac12} \expval{B^2}^{\frac12}$) we have:
	 \begin{equation}
	  0 \leq \frP_n^3 \leq 2c_2 \left(\frac1n\sum_{x=1}^n \frac{
	  \expval{(\tilde{p}^o_x(nt))^2}_{\rho^n}^{\frac12}}{\sqrt{m_x}}\right)
	  \left(\frac1n \sum_{y=1}^n \frac{\expval{(\tilde{p}^{\bullet}_y(nt))^2}_{\rho^n}^{\frac12}}{\sqrt{m_y}}\right) \leq C'\underbrace{\left(\frac1n\sum_{x=1}^n \frac{
	  \expval{(\tilde{p}^o_x(nt))^2}_{\rho^n}^{\frac12}}{\sqrt{m_x}}\right)}_{:=\cL_n'''}
	  \to 0,
	 \end{equation}
	 where we bounded the second sum thanks to the bound 
	 $\expval{\frac{(\hmp_x(nt))^2}{m_x}}_{\rho^n}\leq C $ \eqref{eq:10}; moreover, the 
	 last limit is deduced from \eqref{lowmodecomp}, where we proved $(\cL'''_n)^2 \to 0$ 
	 ($\cL'''_n$ is positive). Therefore, we conclude that $\frP_n' \to 0$ almost surely 
	thanks to \eqref{eq:1511}, and this finishes the proof of \eqref{eq:proof5} for
	$\frP_n$ thanks to \eqref{eq:15}. As we mentioned the proof of the fact that 
	$\frR_n \to 0$ almost surely w.r.t distribution of the masses is exactly similar, where,
	one should take into account the corresponding bound for $r$ in \eqref{BHObound}, 
	\eqref{thermalbound}. This finish the proof of \eqref{eq:proof5} which finishes the proof of \eqref{highermomentclassicale} and Theorem \ref{thmhighermomentclassical}.
	\end{proof}

\section{Quantum Case} \label{section: quantum}
%		 The goal of this section is to prove a theorem similar to  Theorem 
%		 \ref{thmhighermomentclassical} for the quantum counterpart of the system defined in
%		 \eqref{Hamclassical}, where initially we prepare our system in a rather "generic"
%		 initial state. As an application, we deduce a similar result for a quantum 
%		 system which is
%		 prepared in a proper locally Gibbs state (i.e. quantum counterpart of $\rho^n$ 
%		 \eqref{localgibbsstate}).     
	 
\subsection{Model and Results}\label{subsec: quantum model and results}
	   The quantum disordered harmonic chain of size $n$ is defined as follows: 
	   Denote the space variable by $\pmb{\xi} \in \bR^{n-1}$,  
	   let the state space be given by the Hilbert space $\cH_n:= L^2(\bR^{n-1},d\pmb{\xi})$, 
	   denote the elements of $\cH_n$ by bold ket notation $\pmb{\ket{\psi}}$,
	   and denote the inner product in $\cH_n$ by $\pmb{\braket{\phi}{\psi}}$. 
	   Denote also the Schwartz space by $\cS(\bR^{n-1};\bC)$, which is dense in $\cH_n$. The 
	   elongation operator $\ssr_x$ for any $x \in \bI_{n-1}$ is defined on 
	   $\cS(\bR^{n-1},\bC)$ as follows: for any $\pmb{\ket{\psi}} \in \cS(\bR^{n-1};\bC)$ 
	   $$\ssr_x \pmb{\ket{\psi(\xi_1,\dots,\xi_x,\dots,\xi_{n-1})}} = \pmb{\xi_x} 
	   \pmb{\ket{\psi(\xi_1,\dots,\xi_x,\dots,\xi_{n-1})}},$$  
	   Correspondingly, for 
	   $x \in \bI_n$, $\ssp_x$ denotes the momentum operator of particle $x$, 
	   and it is defined as follows: for any $\pmb{\ket{\psi}} \in \cS(\bR^{n-1};\bC)$ 
	   $$ \ssp_x \pmb{\ket{\psi(\xi_1,\dots,\xi_x,\dots,\xi_{n-1})}}=-i \Big(\frac{\partial}
	   {\partial \pmb{\xi}_{x-1}}-\frac{\partial}{\partial \pmb{\xi}_x}\Big) 
	   \pmb{\ket{\psi(\xi_1,\dots,\xi_x,\dots,\xi_{n-1})}}.$$
	   We assume free boundary conditions $r_0=r_n=0$, which means 
	   $p_1=i\partial/\partial_{\xi_1}$, and $p_n=-i \partial/\partial_{\xi_{n-1}}$
	   ($\partial/\partial_{\xi_0}=\partial/\partial_{\xi_n}=0$, by convention). 
	   This means that $\sum_{x=1}^n \ssp_x =0$, corresponding to the fact that we define our model from a center of mass observer's viewpoint, without loss of generality.
	   The canonical commutation relations (CCR) read ($[a,b]:=ab-ba$):
	   \begin{equation}\label{CCR}
	   	[\ssr_x,\ssr_y]=[\ssp_x,\ssp_y]=0, \quad 
	   	[\ssr_x,\ssp_y]=i\left(\delta_{x,(y-1)}-\delta_{x,y} \right),
	   	\quad \forall x \in \bI_{n-1}, y \in \bI_{n}.
	   \end{equation} 
	   
	   The Hamiltonian operator $\mathfrak{H}_n$ is defined on $\cS(\bR^{n-1};\bC)$ as 
	   \begin{equation} \label{hamquantum}
	   \mathfrak{H}_n \;=\; \frac12 \sum_{x=1}^n\left(\frac{\ssp_x^2}{m_x}+\ssr_x^2\right) \; =: \; \sum_{x=1}^n \sse_x,
	   \end{equation}
		where $\{m_x \}_{x=1}^{\infty}$ are the i.i.d random variables defined after \eqref{Hamclassical}.  
%		Although usually this system is defined in terms of $\ssq$ and $\ssp$ operators, for technical reasons our definition
%		is in terms of elongation operator $\ssr$ (cf. Appendix B of \cite{A20} for more 
%		explanations). One could think of $\ssr_x$ 
%		as the length of the spring between particle
%		$x$ and $x+1$.  
%	
%	We define the energy of particle $x$, as 
%	$\sse_x:= \frac12 \left(\frac{\ssp_x^2}{m_x} +\ssr_x^2 \right) $. 
	The operators $\ssp_x$, $\ssr_x$, and consequently $\sse_x$ and $\mathfrak{H}_n$ are essentially self-adjoint on $\cH_n$ (cf. \cite{simon}, \cite{NSS}). 
	Therefore, we consider their closure on $\cH_n$, that we denote with the same symbols. The dense domain of $\frH_n$ is denoted by $\cD(\frH_n)$.
	
	The chain evolves in time according to the Heisenberg dynamics 
	generated by $\frH_n$: For any $t \in \bR$, $e^{it \frH_n}$ is well define thanks to the 
	spectral theorem, since $\frH_n$ is self adjoint. Therefore, using Stone's theorem 
	we can define the one-parameter group of 
	authomorphism $\tau_t^n$, on $\mathcal{B}(\cH_n)$ (set of bounded operator on $\cH_n$) 
	as 
	$$a(t):= \tau_t^n(a)=e^{it \frH_n} a e^{-it\frH_n}, \quad a \in \mathcal{B}(\cH_n). $$
	Notice that $a(t)$ is the solution to the Heisenberg equation 
\begin{equation} \label{heisenberg}
	\dot{a}=i[\frH_n,a], \quad a(0)=a_0,
\end{equation}	   
	where $\dot{a}=\partial_t a(t)$. This equation holds in the strong sense on the proper domain. 
	Appealing to  Stone's theorem, we can extend the domain of $\tau_t^n$ to certain unbounded operators such as $\ssr_x$, $\ssp_x$, $\sse_x$ (cf. \cite{RB1}\cite{RB2}\cite{NSS}\cite{A20}\cite{simon}). 
	Denoting
	 $\tau_t^n(\ssr), \tau_t^n(\ssp), \tau_t^n(\sse) $ by $\ssr_x(t), \ssp_x(t), 
	 \sse_x(t)$ respectively, one can observe that, thanks to the CCR \eqref{CCR},
	 $\ssr_x(t)$ and $\ssp_x(t)$ are the strong solutions to the Heisenberg equations 
	  \begin{equation} \label{eqofmotionquantum}
	  	\begin{split}
	  	&\dot{\ssp}_x=i[\frH_n,\ssp_x]=(\nabla_-\ssr)_x=\ssr_x(t) -\ssr_{x-1}(t), 
	  	\quad x \in 
	  	\bI_n, \\
	  	&\dot{\ssr}_x=i[\frH_n,\ssr_x]=(\nabla_+ M^{-1} \ssp)_x=\frac{\ssp_{x+1}}{m_{x+1}}-
	  	\frac{\ssp_x}{m_x},
	  	 \quad x \in \bI_{n},	  	  	
	  	\end{split}
	  \end{equation}
	on a proper domain, where  $\nabla_-$, $\nabla_+$ and $M$ have been defined in \eqref{eq:discretegrad} and below \eqref{eqofmotion}. 
	Notice that $\ssr$, $\ssp$ in \eqref{eqofmotionquantum} should be understood as vectors of
	elongation and momentum operators (and by abusing the notation 
	matrix multiplication
	$\nabla_- \ssr$ should be understood as multiplication of a finite dimensional 
	matrix and a vector of operators). 
	Let us emphasize that since our system is 
	quadratic, the evolution equation for $\ssr$ and $\ssp$ 
	in the quantum case \eqref{eqofmotionquantum} is
	similar to the classical case \eqref{eqofmotion}.  
	
	Given initial smooth profiles $\bar{p},\bar{r}\in C^1([0,1])$, such that 
	$\bar{r}(0)=\bar{r}(1)=0$,  $\int_0^1 \bar{p}(y)=0$, and  
	$\bar{b} \in C^0([0,1])$ with $0<b_-<\bar{b}(y)<b_+$,
	we define the state (density matrix) $\varrho_n^{\bar{r},\bar{p},\bar{b}} \in \cB(\cH_n)$ such that for each $n$, 
	$\roo$ is trace-class and bounded with $\Tr(\roo)=1$ (as in the classical case, supersripts will be omitted whenever possible).
	In addition, for any operator $a$ such that $a \roo$ be trace-class we define: 
	 \begin{equation} \label{averagestate}
	 \expval{a}_{\roo} := \Tr(a\roo).
	\end{equation}	    
	Finally we impose that $\roo$ represents a local equilibrium state, by imposing the following set of conditions: 
	
	\paragraph{Assumptions on $\roo$:}
	\begin{enumerate}
	
	\item\label{AA1} 
	Local mechanical equilibrium: For any $y \in [0,1]$, 
	\begin{equation} \label{avgassumrp}
	\expval{\ssr_{[ny]}}_{\roo}= \bar{r}\Big(\frac{[ny]}{n}\Big), \quad
	\frac{1}{m_{[ny]}}\expval{\ssp_{[ny]}}_{\roo}=  \frac{1}{\bar{m}}\bar{p}\Big(\frac{[ny]}{n}\Big) +\epsilon_n^{[ny]},
	\end{equation}  
	where for all $y \in [0,1]$, $\epsilon_n^{[ny]} \to 0$ as $n \to \infty$ 
	almost surely w.r.t the distribution of the masses.
	
	\item\label{AA2} 
	Local thermal equilibrium: Let
	\begin{equation} \label{tildedef}
	\tilde{\ssp}_x:=\ssp_x-\expval{\ssp_x}_{\roo}, \quad 
	\tilde{\ssr}_x:= \ssr_x-\expval{\ssr_x}_{\roo}, \quad
	\breve{\sse}_x= \frac{1}{2} \left(\frac{\tilde{\ssp}_x^2}{m_x} + 
	\tilde{\ssr}_x^2 \right). 
	\end{equation}	
	For any test function $f \in C^0([0,1])$, 
	\begin{equation} \label{avgassume}
	\frac1n \sum_{x=1}^n f\big(\frac{x}{n}\big) \expval{\breve{\sse}_x}_{\roo} \to 
	\int_0^1 f(y)\bar{b}(y)dy,
	\end{equation}
	almost surely w.r.t to the distribution of the masses. 	
	Contrary to the classical case, $b(y)$ cannot longer be identified with a local temperature $\beta^{-1}(y)$, cf.~\eqref{fdef} below.
	
	\item\label{AA3} 
	Polynomial clustering: 
	For $x \in \bI_n$, let us use the notations $O_x^r:= \tilde{\ssr}_x$ and $O_x^p:=\tilde{\ssp}_x$, see \eqref{tildedef}. 
	There exists $C>0$ such that for any $n$, any
	$x_1 , x_2 , x_3 , x_4 \in \bI_n$, any $\sharp_1,\sharp_2,\sharp_3,\sharp_4 
	\in \{r,p\}$, and any  realization of the masses, we have: 
	\begin{equation} \label{decayassum}
	\begin{split}	
	&\left|\expval{O_{x_1}^{\sharp_1} O_{x_2}^{\sharp_2}}_{\roo} \right| \leq  \frac{C}{|x_1-x_2|^{a_1}}, \quad \text{with} \quad a_1\geq 2, \\
	&\left| \expval{O_{x_1}^{\sharp_1} O_{x_2}^{\sharp_2} O_{x_3}^{\sharp_3}}_{\roo} \right|
	\leq \frac{C}{|\max\{|x_2-x_1|,|x_3-x_2|\}|^{a_2}}, \quad \text{with} \quad a_2 \geq 3, \\
	 &\left|\expval{O_{x_1}^{\sharp_1} O_{x_2}^{\sharp_2} O_{x_3}^{\sharp_3}
	 O_{x_4}^{\sharp_4}}_{\roo}- 
	 \expval{O_{x_1}^{\sharp_1} O_{x_2}^{\sharp_2}}_{\roo} \expval{O_{x_3}^{\sharp_3}
	 O_{x_4}^{\sharp_4}}_{\roo}\right| \leq \\ & \frac{C}{|x_2-x_3|^{a_3} 
	 |x_1-x_4|^{a_3}}+
	 \frac{C}{|x_1-x_3|^{a_3}|x_2-x_4|^{a_3}},  \quad \text{with}\: a_3 \geq 2
	\end{split}	
	\end{equation}
	with the convention $\frac10=1$.
	%In the above expressions, by convention we take $\frac10=1$, i.e., each "pair" moment is 
	%bounded by a constant.   
	\end{enumerate}
	Initially, we prepare the chain in the state $\roo$ having the above properties. 
	Let us emphasize that this state could be either a mixed state or a pure state 
	i.e., $\roo=\pmb{\ketbra{\psi}}$ for a suitable wave function $\pmb{\psi}\in \cH_n$. 
	%In fact, as we explained before, the latter case could have interesting physical interpretation. 
	
	The evolution at the macroscopic level is still given by \eqref{eq:macro},
	 where one  should modify the initial 
	datum \eqref{eq:macro3} by replacing  $\beta^{-1}(y)$ with  $\bar{b}(y)$ featuring in \eqref{avgassume}: 
	\begin{equation} \label{eq:macro3Q}
	\begin{split}
		&\fr(y,0)=\bar{r}(y), \quad \fp(y,0)=\bar{p}(y), \quad \fe(y,0)=
		\frac{\bar{p}^2(y)}{2\bar{m}}+\frac{\bar{r}^2(y)}{2}
		+\bar{b}(y), \\
		&\fr(0,t)=\fr(1,t)=0 \quad \forall t \in [0,T]. 	
	\end{split}
	\end{equation}
%	 and
%	 \eqref{eq:macro3}. 
We prove the following result: 
	 \begin{theorem} \label{mainthmQ}
	 Let $f \in C^0([0,1])$, and fix $T>0$. 
	 We assume that the initial state $\roo$, parametrized by the smooth profiles $\bar{p},\bar{r},\bar{b}$ satisfies the conditions \eqref{avgassumrp} to \eqref{decayassum}. 
	 Let $(\ssr(nt), \ssp(nt), \sse(nt))$ be the vector of evolved operators according to the Heisenberg evolution \eqref{eqofmotionquantum}. 
	 In addition, let $(\fr(y,t),\fp(y,t),\fe(y,t))$ be solutions to the macroscopic evolution equations \eqref{eq:macro} with initial datum \eqref{eq:macro3Q}.
%	 \eqref{eq:macro1Q}, and \eqref{eq:macro2Q} with boundary and initial conditions \eqref{eq:macro3Q}. 
	 Then for any $t \in [0,T]$,
	 \begin{equation} \label{highermomentclassicalrQ}
	\left\langle \left(\frac1n \sum_{x=1}^n f(\frac{x}{n})\ssr_x(nt) -\int_0^1 
	f(y) \fr(y,t) dy \right)^2 \right\rangle_{\roo} \to 0,
	\end{equation}
	\begin{equation}\label{highermomentclassicalpQ}
	\left\langle \left(\frac1n \sum_{x=1}^n f(\frac{x}{n})\ssp_x(nt) -\int_0^1 
	f(y) \fp(y,t) dy \right)^2 \right\rangle_{\roo^n} \to 0,
	\end{equation}
	\begin{equation}\label{highermomentclassicaleQ}
	\left\langle \left(\frac1n \sum_{x=1}^n f(\frac{x}{n})\sse_x(nt) -\int_0^1 
	f(y) \fe(y,t) dy \right)^2 \right\rangle_{\roo^n} \to 0,
	\end{equation}
	as $n \to \infty$ almost surely w.r.t the distribution of the masses.
	 \end{theorem}
	 
	In order to make the above theorem more concrete, let us give an example of state that satisfies conditions \eqref{avgassumrp} to \eqref{decayassum} stated above.
	Given momentum and elongation profiles $\bar{p}$ and $\bar{r} \in  C^1([0,1])$ as before, and a profile $\beta \in C^0([0,1])$ such that $0<\beta_-\leq \beta(y) \leq \beta_+ $,
	we define the \textit{locally Gibbs state} as the density operator acting on $\cH_n$ given by 
	 \begin{equation} \label{locallyGibbs}
	 \Roo=\frac{1}{Z_n} \exp(-\frac{1}{2}\sum_{x=1}^n \Big[ \frac{\beta(\frac{x}{n})}{m_x} 
	 (\ssp_x-\bar{p}(\frac{x}{n})\frac{m_x}{\bar{m}})^2+\beta(\frac{x}{n})(\ssr_x-\bar{r}
	 (\frac{x}{n}))^2 \Big] ) =: \frac{1}{Z_n}\exp(-H_{\beta}^n),
	\end{equation}
	with $Z_n$ a normalizing constant such that $\Tr(\Roo)=1$.
	The operator $H_{\beta}^n$ implicitly defined by \eqref{locallyGibbs} is essentially self-adjoint on $\cS(\bR^{n-1},\bC)$.
	Its closure is self-adjoint and we denote it with the same symbol. 
	Since $H_{\beta}^n$ has a discrete spectrum with non-negative eigenvalues, one can observe that $\Roo$ is a well-defined density operator
	(bounded and trace-class, cf.~\cite{A20}). 
	Let us denote the average with respect to $\Roo$ by 
	\begin{equation} \label{avgthermal}
	  	\expval{a}_{\beta} := \Tr(a \Roo)
	\end{equation}	   
	for any suitable operator $a$.
	Thanks to the properties of $\Roo$, it is straightforward to observe that all the averages appearing subsequently are well defined and bounded (cf. \cite{A20}).
	 Given the distribution of the masses $\mu$, and temperature profile $\beta$, we 
	 define the macroscopic thermal energy profile as follows:
	  %of $\breve{\sse}_x$ \eqref{tildedef}
	 \begin{equation} \label{fdef}
	 	\frf_{\beta}^{\mu} (y) := \lim_{n \to \infty}
	 	\bE\left( \expval{\breve{\sse}_{[ny]}}_{\beta}\right).
	\end{equation}	
	with  $\breve{\sse}_x$ defined in \eqref{tildedef}.
	In \cite{A20}, it is proven that $\frf_{\beta}^{\mu}$ is well-defined and continuous, and $\frf_{\beta}^{\mu}(y)=\frf^{\mu}(\beta(y))$, where $\frf^{\mu}(\beta_{eq})$ is 
	the corresponding function for a constant profile of temperature $\beta(y)=\beta_{eq}
	 , \: \forall y$ (cf. \cite{A20} Appendix A). %Now we are prepared to express the second theorem of this section. 
	 
	 \begin{theorem} \label{thmGibbs}
	 The state $\Roo$ defined in \eqref{locallyGibbs} satisfies conditions \eqref{avgassumrp} to \eqref{decayassum} with $\bar{p},\bar{r}$, $\bar{b}(y)=\frf_{\beta}^{\mu}(y)$. 
%	 
%	 Recall the definition of the locally Gibbs state $\Roo$ \eqref{locallyGibbs}, and 
%	 smooth profiles $\bar{r}$, $\bar{p}$, $\beta$ in its definition. Correspondingly,
%	  recall the  
%	 function $\frf^{\mu}_{\beta}$. Then $\Roo$ satisfies the properties \ref{A1}, 
%	 \ref{A2}, and \ref{A3} i.e., \eqref{avgassumrp}, \eqref{avgassume}, and 
%	 \eqref{decayassum} with $\bar{p},\bar{r}$, $\bar{b}(y)=\frf_{\beta}^{\mu}(y)$. 
%	 Therefore, thanks to the Theorem \ref{mainthmQ}, the identities 
%	 \eqref{highermomentclassicalrQ}, \eqref{highermomentclassicalpQ}, and 
%	 \eqref{highermomentclassicaleQ} holds for $\Roo$, where in these identities
%	 $(\fr(y,t),\fp(y,t), \fe(y,t))$ are solution to the evolution equation \eqref{eq:macro1Q}, \eqref{eq:macro2Q}, with initial datum \eqref{eq:macro3Q}, with 
%	 $\bar{b}=\frf_{\beta}^{\mu}$. 
	  \end{theorem}
	 % As we'll see in the proof, the locally Gibbs state satisfies stronger conditions comparing to \eqref{decayassum}, i.e., we have exponential decay. \\
	
	Before proceeding, let us briefly mention that in the assumption \eqref{decayassum}, the estimates need to hold for all realizations of the masses. 
	However, it is possible to show that they could be replaced by estimates on the expectations (w.r.t $\mu$) of the same quantities, at the cost of stronger decay rates. 
	%However, this can be replaced by estimates on the expectations (w.r.t $\mu$) of the same quantities with stronger decays. 

	\subsection{Evolution of $r,p$}
	 \subsubsection{Time evolution}
	 Since our system is harmonic, equation of motion \eqref{eqofmotionquantum} can be 
	 solved similarly to the classical case \eqref{soleqmotion}. Let us briefly recall the 
	 desired transformations in this case from \cite{A20}, mainly to fix proper notations.
	 Recall the random matrices $A_p$, and $A_r$ \eqref{mainmatrices}, and their 
	 corresponding set of eigenvalues and eigenvectors ($0< \omega_1^2< \dots <\omega_{n-1}
	 ^2$ and $\{\varphi^k \}_{k=0}^n$ for the former and $ \omega_1^2< \dots 
	 \omega_{n-1}^2$ and $\{\phi^k \}_{k=1}^n$ for the latter). Define the following set of 
	 operators (recall that $\expval{,}_n$ is the inner product in $\bR^n$):
	 \begin{equation} \label{hatQ1}
	 \hat{\ssp}_k:= \langle \varphi^k,M^{-\frac12} \ssp \rangle_n= \sum_{x=1}^n
	 \frac{\varphi^k_x}{\sqrt{m_x}} \ssp_x, \quad \forall k \in \bI_{n-1}^o; \quad 
	 \hat{\ssr}_k:= \langle \phi^k, \ssr \rangle_{n-1}=\sum_{x=1}^{n-1}\phi^k_x 
	 \ssr_x, 
	 \quad \forall k \in \bI_{n-1}.
	 \end{equation}
	After a straightforward computation (cf. Sect.3 of \cite{A20}), we have the following
	commutation relations for $k,k' \in \bI_{n-1}^o$ (by convention $\hat{\ssr}_0=0$): 
	\begin{equation} \label{CCRhat}
	[\hat{\ssr}_k,\hat{\ssr}_{k'}]=[\hat{\ssp}_k,\hat{\ssp}_{k'}]=0, 
	\quad [\hat{\ssr}_k,\hat{\ssp}_{k'}]=i\omega_k \delta_{k,k'}.
	\end{equation}	 
	Consequently, we define the corresponding bosonic operators 
	$\hat{\ssb}_k$, $\hat{\ssb}_k^*$, 
	having suitable commutation relations (thanks to \eqref{CCRhat})
	as follows: for any  $k,k' \in \bI_{n-1}$:
	\begin{equation} \label{bosonic}
	\begin{split}
	&\hat{\ssb}_k := \frac{1}{\sqrt{2 \omega_k}}(\hat{\ssr}_k+ i\hat{\ssp}_k), \quad 
	\hat{\ssb}_k^* := \frac{1}{\sqrt{2 \omega_k}}(\hat{\ssr}_k- i\hat{\ssp}_k),\\
	&[\hat{\ssb}_k, \hat{\ssb}_{k'}^*]= \delta_{k,k'}, \quad 
	[\hat{\ssb}_k, \hat{\ssb}_{k'}]=[\hat{\ssb}_k^*, \hat{\ssb}_{k'}^*]=0.
	\end{split}	
	\end{equation}
		Thanks to \eqref{hatQ1} and \eqref{bosonic}, and by using the definition
		of $\varphi^k$, $\phi^k$, we can rewrite the Hamiltonian 
	$\mathfrak{H}_n$ in terms of bosonic operators:
	\begin{equation} \label{hamboson}
		\mathfrak{H}_n = \frac12 \sum_{x=1}^n(\frac{\ssp_x^2}{m_x}+\ssr_x^2)=
		\frac12 \hat{\ssp}_0 + \frac12 \sum_{k=1}^{n-1}(\hat{\ssp}_k^2 +\hat{\ssr}_k^2)
		=\frac12 \sum_{k=1}^{n-1}(\hat{\ssp}_k^2 +\hat{\ssr}_k^2)
		=\sum_{k=1}^{n-1} \omega_k(\hat{\ssb}_k^* \hat{\ssb}_k+\frac12),
	\end{equation}		
	where we used the fact that $\ssp_0=\sum_{x=1}^n \ssp_x=0$. Notice that thanks to the
	above expression, one could obtain the spectrum of $\mathfrak{H}_n$ explicitly,
	which is a discrete spectrum (cf. Sect.3 of \cite{A20}). In addition, 
	thanks to the above 
	representation of $\mathfrak{H}_n$, and commutation relations \eqref{bosonic}, 
	the time evolution of the bosonic operators (action of of $\tau^n_t$ i.e., solution to the 
	Heisenberg equation \eqref{heisenberg}) is given by: 
	 \begin{equation} \label{timebosonic}
	 \hat{\ssb}_k(t)= e^{- i \omega_k t} \hat{\ssb}_k(0), \quad 
	 \hat{\ssb}_k^*(t)= e^{i \omega_k t} \hat{\ssb}_k^*(0).
	 \end{equation}
	By using the inverse of \eqref{bosonic}, we get for $k\in \bI_{n-1}$: 
	\begin{equation} \label{kevolution}
	\hat{\ssp}_k(t) =\cos(\omega_kt) \hat{\ssp}_k(0)- \sin(\omega_kt) \hat{\ssr}_k(0),
	\quad
	\hat{\ssr}_k(t) =\cos(\omega_kt) \hat{\ssr}_k(0)+ \sin(\omega_kt) \hat{\ssp}_k(0)
	\end{equation}	 
	Plugging the latter inside the inverse of \eqref{hatQ1}, the time evolution of 
	$p_x$, $r_x$ is given by (recall the definition of $\hat{\ssp}$, 
	$\hat{\ssr}$ \eqref{hatQ1}): 
	\begin{equation} \label{Qevolution}
\begin{split}	
	&\ssp(t)=\sum_{k=0}^{n-1}  M^{\frac12} \varphi^k \hat{\ssp}_k (t) = \sum_{k=0}^{n-1} \big(\cos(\omega_kt)\hat{\ssp}_k(0)-\sin(\omega_kt) \hat{\ssr}_k(0)\big)M^{\frac12} \varphi^k, \\
	%&\ssp_x(t)=\sum_{k=0}^{n-1}  \sqrt{m_x} \varphi^k_x \hat{\ssp}_k (t) = \sum_{k=0}^{n-1} \big(\cos(\omega_kt)\hat{\ssp}_k(0)-\sin(\omega_kt) \hat{\ssr}_k(0)\big)\sqrt{m_x} \varphi^k_x,\\
	&\ssr(t)=\sum_{k=1}^{n-1} \phi^k \ssr_k(t)=\sum_{k=1}^{n-1} \big(\cos(\omega_kt)\hat{\ssr}_k(0)+\sin(\omega_kt) \hat{\ssp}_k(0)\big)\phi^k.
	%&\ssr_x(t)=\sum_{k=1}^{n-1} \phi^k_x \ssr_k(t)=\sum_{k=1}^{n-1} \big(\cos(\omega_kt)\hat{\ssr}_k(0)+\sin(\omega_kt) \hat{\ssp}_k(0)\big)\phi^k_x.
\end{split}
\end{equation}
	These equations are formally similar to their classical counterpart \eqref{soleqmotion}, but $\ssr$ and $\ssp$ are (non-commuting) operators rather than numbers.

\subsubsection{Proof for $r,p$}
In this section, we prove the first two limits in Theorem \ref{mainthmQ}, i.e.\@ \eqref{highermomentclassicalrQ} and \eqref{highermomentclassicalpQ}. 
Similarly to the classical case we use the explicit time evolution \eqref{Qevolution}, as well as the localization estimates \eqref{localizationestimate}. 
However, our state is no longer Gaussian and simple correlation identities \eqref{cov1}, \eqref{covkk'} need to be replaced by \eqref{decayassum}.

\begin{proof} [Proof of \eqref{highermomentclassicalrQ} and \eqref{highermomentclassicalpQ}.]  
Recall that the state $\roo$ satisfies assumptions \eqref{avgassumrp} to \eqref{decayassum}. 
Let us define the \emph{average} variables 
$\bar{\ssp}_x(nt)$, $\bar{\ssr}_x(nt)$, and \emph{fluctuation} operators $\tilde{\ssp}_x(nt)$, $\tilde{\ssr}_x(nt)$  as follows for any $x \in \bI_n$: 
\begin{equation}\label{avgfluc}
\begin{split}
\bar{\ssp}_x(nt):= \expval{\ssp_x(nt)}_{\roo}, \quad \bar{\ssr}_x(nt):= \expval{\ssr_x(nt)}_{\roo};  \\
\tilde{\ssp}_x(nt) := \ssp_x(nt)-\bar{\ssp}_x(nt), \quad  
 \tilde{\ssr}_x(nt) := \ssr_x(nt)-\bar{\ssr}_x(nt).
\end{split}
\end{equation} 
Thanks to the linearity of the trace, and the evolution equation \eqref{eqofmotionquantum},
$(\bar{\ssp}(t)$, $\bar{\ssr}(t)) \in \bR^{2n-1}$ 
satisfies the following system of ODEs: 
\begin{equation} \label{ODEQ}
\frac{d}{dt} \bar{\ssp}_x(t) = \bar{\ssr}_x(t)-\bar{\ssr}_{x-1}(t), \quad
\frac{d}{dt} \bar{\ssr}_x(t)=\frac{\bar{\ssp}_{x+1}(t)}{m_{x+1}} - \frac{\bar{\ssp}_x(t)}
{m_x},
\end{equation}
where the initial data are given by $\bar{\ssr}_x(0)=\bar{r}(\frac{x}{n})$, and $\bar{\ssp}_x(0)=\frac{m_x}{\bar{m}} \bar{p}(\frac{x}{n})+ \epsilon^x_n$ for all $x$, thanks to \eqref{avgassumrp}.
In \cite{A20}, it is shown that % has been proved that sequence of solutions to a similar system of ODEs satisfy the following limit (cf. \cite{A20}, in particular Sec.~4, (4.5), (2.16) and (2.17)): 
\begin{equation} \label{hydrolimitQrp}
 \lim_{n \to \infty} \frac1n \sum_{x=1}^n 
 \bar{\ssp}_x(nt) f(\frac{x}{n}) \to \int_{0}^1 \fp(y,t) f(y) dy, \quad
  \lim_{n \to \infty} \frac1n \sum_{x=1}^n 
 \bar{\ssr}_x(nt) f(\frac{x}{n}) \to \int_{0}^1 \fr(y,t) f(y) dy,
\end{equation} 
i.e.\@ the limit on average \eqref{eq: limit on average}.
The initial conditions in \cite{A20} are not the same as here, but the difference vanishes in the limit. 
%with initial conditions that differ from here with a difference that vanishes in the limit. 
%	almost surely w.r.t the distribution of the masses. Let us mention that ODEs of 
%	\cite{A20} differ 
%	with ours in the initial datum with a factor which vanishes in the limit. 
%	However, the above-mentioned result can be adapted to our setup rather easily. We omit 
%	the straightforward adaptation here (cf. \cite{A20} Remark 4.1. and proceeding proof).
	As in the classical case, we can add and subtract  
	$\frac1n \sum_{x=1}^n f(\frac{x}{n}) \expval{\ssp_x(nt)}_{\roo}$ inside the square in \eqref{highermomentclassicalpQ}, and similarly for $r$. 
	Then by using the basic properties of the trace (linearity and scalar multiplication), and limits \eqref{hydrolimitQrp}, 
	the proof of \eqref{highermomentclassicalrQ} and \eqref{highermomentclassicalpQ} boils down to
	proving that 
	\begin{equation} \label{eqQ:proof3}
		\mathsf{R}_n:=\left\langle \left(\frac1n \sum_{x=1}^n f(\frac{x}{n})\tilde{\ssr}_x(nt) \right)^2
	\right\rangle_{\roo} \to 0, \quad 
	\mathsf{P}_n:=\left\langle \left(\frac1n \sum_{x=1}^n f(\frac{x}{n})
	\tilde{\ssp}_x(nt) \right)^2
	\right\rangle_{\roo} \to 0, 
		\end{equation}
	as $n \to \infty$, almost surely w.r.t the distribution of the masses, where one should 
	recall the fluctuation operators $\tilde{\ssp}_x, \tilde{\ssr}_x$ \eqref{avgfluc}.
	
%	Since $(\ssp(nt),\ssr(nt))$ and $(\bar{\ssp}(nt),\bar{\ssr}(nt))$ are solution
%	to the "similar" linear equations (),  
	The evolution of $\tilde{\ssp}(nt)$ and $\tilde{\ssr}(nt)$ can be obtained explicitly from \eqref{eqofmotionquantum} and \eqref{ODEQ}.
	For $k \in \bI_n^o$, let 
	 \begin{equation} \label{Qtildenitial}
	 \hat{\tilde{\ssp}}_k= \hat{\tilde{\ssp}}_k(0):=\langle \varphi^k, M^{-\frac12} 
	 \tilde{\ssp} 
	 \rangle_n, \quad  
	 \hat{\tilde{\ssr}}_k:= \langle \phi^k, \tilde{\ssr} \rangle_{n-1}, 
	\end{equation}	 
	where the vectors $\tilde{\ssp}$ and $\tilde{\ssr}$ are defined in \eqref{tildedef}, 
	\eqref{avgfluc}. 
	Then $\tilde{\ssr}_x(t)$, $\tilde{\ssp}_x(t)$ is given by:
	 \begin{equation}\label{Qtildesol}
	\begin{split}	 
	& \tilde{\ssr}_x(t)= \sum_{k=1}^{n-1}  \phi^k_x 
	\hat{\tilde{\ssr}}_k(t) = \sum_{k=1}^{n-1}
	\left(\cos(\omega_kt) \hat{\tilde{\ssr}}_k(0) + \sin(\omega_kt) \hat{\tilde{\ssp}}_k(0) \right) 
	\phi^k_x, \\
	& \tilde{\ssp}_x(t)=\sum_{k=0}^{n-1} \sqrt{m_x} \varphi^k_x \hat{\tilde{\ssp}}_k(t)=
	\sum_{k=0}^{n-1} \left(\cos(\omega_kt) \hat{\tilde{\ssp}}_k(0) - \sin(\omega_kt) 
	\hat{\tilde{\ssr}}
	_k(0) \right)\sqrt{m_x} \varphi^k_x,
		\end{split}	 
	 \end{equation} 
	where $\hat{\tilde{\ssr}}_k(t)$, and $\hat{\tilde{\ssr}}_k(t)$ are
	implicitly defined. 
	
	Before proceeding, we need bounds similar to \eqref{initialboundt} in the quantum case.
	This can be done thanks to the first bound in \eqref{decayassum}:
	%We simply compute (we denote the set $\{(x,y) \in \bI_n^2 | |x-y| =l \}$ by $\{|x-y|=l \}$; we drop the subscript of $\expval{\cdot}_{\roo}$ ):
	\begin{multline} \label{Qpkbound}
		\expval{\hat{\tilde{\ssp}}_k^2}_{\roo}= \sum_{x,y=1}^n 
		\frac{\varphi^k_x \varphi^k_y}{\sqrt{m_xm_y}} \expval{\tilde{\ssp}_x
		\tilde{\ssp}_y}
		\leq \sum_{x=1}^n \frac{(\varphi^k_x)^2}{m_x} \expval{\tilde{\ssp}_x^2} +
		\sum_{l=1}^{n-1} \sum_{|x-y|=l} \left| \frac{\varphi^k_x \varphi^k_y}{
		\sqrt{m_xm_y}} \right| \left| \expval{\tilde{\ssp}_x \tilde{\ssp}_y} \right| \\ 			
		\le C_1 \sum_{x=1}^n (\varphi^k_x)^2+ \sum_{l=1}^{n-1} \frac{C_2}{l^2} 
		\sum_{|x-y|=l}  \left|\varphi^k_x \varphi^k_y \right| \leq C_1 + 
		\sum_{l=1}^{n-1} \frac{C_3}{l^2}  \leq C,	
	\end{multline}	
	where we have dropped the subscript $\roo$ on the r.h.s.
	To get this, we used the definition of $\hat{\tilde{\ssp}}_k$ \eqref{Qtildenitial}, we used \eqref{decayassum} in the second line, 
	we took advantage of the following Cauchy-Schwartz inequality,
	$$\sum_{|x-y|=l} |\varphi^k_x \varphi^k_y|=
	2\sum_{x=1}^{n-l} |\varphi^k_x\varphi^k_{x+l}| \leq 
	2 \left(\sum_{x=1}^{n-l} (\varphi^k_x)^2 \right)^{\frac12}
	 \left(\sum_{x=1}^{n-l} (\varphi^k_{x+l})^2 \right)^{\frac12} 
	 \leq 2\sum_{x=1}^n (\varphi^k_x)^2 \leq 2, $$
	and finally, in the last inequality of \eqref{Qpkbound}, we used the fact that 
	the featuring series is summable. Notice that the constant $C>0$ in \eqref{Qpkbound}
	is independent of $n$ and the realization of the masses. Similarly for $r$, 
	by an exact 
	identical computation, and using corresponding definition and bounds 
	\eqref{Qtildenitial},\eqref{decayassum} we have: 
	there exist $C'>0 $ such that for all $n$ and any
	 realization of the masses we have:
	 \begin{equation} \label{Qrkbound}
	 \expval{\hat{\tilde{\ssr}}_k^2}_{\roo} \leq C'.
\end{equation}	 	 
In the proceeding, we follow somewhat similar steps as in the classical setup, where each
step is modified due to the quantum nature:

\paragraph{Step1: Low modes.}
For a proper $\gamma$ with $0<\gamma<\frac12$, define $\tilde{\ssp}^o(nt)$, and 
$\tilde{\ssr}^o(nt)$ as low modes portion of $\ssp(nt)$, and $\ssr(nt)$ (we
will choose $\gamma$ opportunely later): for $x \in \bI_n$:
\begin{equation} \label{Qlowmodedef}
\tilde{\ssp}^o_x(nt):= \sum_{k \in \bZ \cap [0,n^{(1-\gamma)}]} \hat{\tilde{\ssp}}_k 
			(nt) \sqrt{m_x} \varphi_x^k, \quad 	  
			\tilde{\ssr}^o_x(nt):= \sum_{k \in \bZ \cap [1,n^{(1-\gamma)}]} 
			\hat{\tilde{\ssr}}_k 
			(nt)  \phi_x^k,
\end{equation}
where definition of $\hat{\tilde{\ssp}}_k(nt)$, $\hat{\tilde{\ssr}}_k(nt)$ 
should be recalled from 
\eqref{Qtildesol}, and by convention we set $\tilde{\ssr}_n(nt)=0$. Here we prove that: 
\begin{equation} \label{Qlowmodelemm}
\mathsf{L}_n:= \expval{\frac1n\left(\sum_{x=1}^n f(\frac{x}{n}) \tilde{\ssp}_x^o(nt)  \right)^2}_{\roo} \to 0, \quad \mathsf{L'}_n:= \expval{\frac1n\left(\sum_{x=1}^n f(\frac{x}{n}) \tilde{\ssr}_x^o(nt)  \right)^2}_{\roo} \to 0,
\end{equation}
almost surely w.r.t the distribution of the masses.

Before proceeding, thanks to the above mentioned assumptions, one can check that 
$(\roo)^{\frac12}a$ is a Hilbert-Schmidt operator for any operator $a$ that is a linear or quadratic function of $\mathsf p_x$ and $\mathsf r_x$. 
Therefore, it is straightforward to see that we can use the following form of Cauchy-Schwartz inequality for our desired operators: 
\begin{equation}\label{CSineq}
\left|\expval{ab^*}_{\roo} \right|^2 \leq \expval{aa^*}_{\roo} \expval{bb^*}_{\roo}.
\end{equation}
Later, we will simply use this inequality, 
and  we will not mention its legitimacy anymore which is easy to check in each case.

Following exactly the same steps as in \eqref{lowmodecomp} yields:
\begin{equation} \label{Qlowmodecomp}
0 \leq \mathsf{L}_n \leq \frac{C}{n} \sum_{k \in \tilde{I}(\gamma)}\expval{(
\hat{\tilde{\ssp}}_k(nt))^2}_{\roo},
\end{equation}
where $C>0$ is a constant independent of $n$ and the realization of the masses. 
Notice that the only difference between \eqref{Qlowmodecomp} and  \eqref{lowmodecomp} is the fact that 
in the former we should use the operator form of the Cauchy-Schwartz inequality \eqref{CSineq}.  
By using the definition of $\hat{\tilde{\ssp}}_k(nt)$ from \eqref{Qtildesol}, thanks to \eqref{Qpkbound}, \eqref{Qrkbound} and a Cauchy-Schwartz
 inequality, we have for any $k \in \bI_{n-1}^o$ 
 \begin{equation}\label{Qtrpkbound}
 \expval{(\hat{\tilde{\ssp}}_k(nt))^2}_{\roo} \leq C, \quad
  \expval{(\hat{\tilde{\ssr}}_k(nt))^2}_{\roo} \leq C,
\end{equation}  
where $C>0$ is independent of $n$, and the distribution of the masses. Combining 
\eqref{Qtrpkbound} with \eqref{Qlowmodedef} and recalling the fact that 
$\tilde{I}(\gamma)=[0,n^{1-\gamma}] \cap \bZ$ with $\gamma>0$ we deduce that 
$\mathsf{L}_n \to 0$ as $n \to \infty$. Again, a similar computation, using the
definition of $\hat{\tilde{\ssr}}^o$, the fact that $\phi^k$ is an orthonormal basis,
 and the second bound in \eqref{Qtrpkbound} yield  $\mathsf{L}_n'\to 0$ 
 deterministically. This finishes the proof of this step, i.e.\@ \eqref{Qlowmodelemm}.
 
\paragraph{Step2: High modes.}
As in \eqref{highmodedef}, define the ``high mode operators'' for $x  \in \bI_n$: 
\begin{equation} \label{Qhighmodedef}
\tilde{\ssp}_x^{\bullet}(nt):= \tilde{\ssp}_x(nt) - \tilde{\ssp}_x^o(nt)= \sum_{k \in 
I(\gamma)} \hat{\tilde{\ssp}}_k(nt) \sqrt{m_x}\varphi^k_x, \quad
\tilde{\ssr}_x^{\bullet}(nt):= \tilde{\ssr}_x(nt) - \tilde{\ssr}_x^o(nt)= \sum_{k \in 
I(\gamma)} \hat{\tilde{\ssr}}_k(nt) \sqrt{m_x}\phi^k_x,
\end{equation}
	where we choose $0<\gamma<\frac12$  opportunely later, and $I(\gamma)= ] n^{(1-
	\gamma),n}] \cap \bZ$. Moreover, by convention we have $\tilde{\ssr}_n(nt)=0$. 
	 Following the lines of the classical case, for a proper choice of 
	 $0<\gamma<\frac12$ we prove:
	 \begin{equation} \label{Qhighmodeleprp}
	 \mathsf{U}_n := \expval{\left(\frac1n \sum_{x=1}^n f(\frac{x}{n})
	 \tilde{\ssp}^{\bullet}_x(nt) \right)^2}_{\roo} \to 0, 
	 \quad  \mathsf{U'}_n := \expval{\left(\frac1n \sum_{x=1}^n f(\frac{x}{n})
	 \tilde{\ssr}^{\bullet}_x(nt) \right)^2}_{\roo} \to 0,
	 \end{equation}
	 almost surely, w.r.t distribution of the masses.
	 
	 \paragraph{Step 2.1: High modes $|x-y|>4n^{\theta'}$.}
	 Let $0<2\gamma<\theta<\theta'<1$. Fix $n$, and consider $x,y \in \bI_n$ such that 
	 $|x-y|>4n^{\theta'}$. Let us bound 
	 $|\expval{\tilde{\ssp}_x^{\bullet}(nt)\tilde{\ssp}_y^{\bullet}(nt)}_{\roo}|
	 /{m_xm_y}$ (recall $\tilde{\varphi}^k_x=\frac{\varphi^k_x}{\sqrt{m_x}}$):
	 \begin{equation} \label{Qeq:1}
	\begin{split}	 
	 &\frac{|\expval{\tilde{\ssp}_x^{\bullet}(nt)\tilde{\ssp}_y^{\bullet}(nt)}_{\roo}|}
	 {m_xm_y}= \left| \sum_{k,k' \in I(\gamma)} 
	 \expval{\hat{\tilde{\ssp}}_k(nt)
	 \hat{\tilde{\ssp}}_{k'}(nt)}_{\roo} \tilde{\varphi}^k_x 
	 \tilde{\varphi}^{k'}_y \right|\leq \\  
	 &\underbrace{\left|\sum_{k,k' \in I(\gamma)} \sin(\omega_k nt) \sin(\omega_{k'}nt) 
	 \tilde{\varphi}^k_x 
	 \tilde{\varphi}^{k'}_y 
	 \expval{\hat{\tilde{\ssr}}_k
	 \hat{\tilde{\ssr}}_{k'}}_{\roo}\right|}_{=: \mathsf{C}_{xy}^{rr}} + 
	 \underbrace{\left|\sum_{k,k' \in I(\gamma)} \cos(\omega_k nt) \sin(\omega_{k'}nt) 
	 \tilde{\varphi}^k_x 
	 \tilde{\varphi}^{k'}_y 
	 \expval{\hat{\tilde{\ssp}}_k
	 \hat{\tilde{\ssr}}_{k'}}_{\roo}\right|}_{=: \mathsf{C}_{xy}^{pr}}+ \\
	 &\underbrace{\left|\sum_{k,k' \in I(\gamma)} \sin(\omega_k nt) \cos(\omega_{k'}nt) 
	 \tilde{\varphi}^k_x 
	 \tilde{\varphi}^{k'}_y 
	 \expval{\hat{\tilde{\ssr}}_k
	 \hat{\tilde{\ssp}}_{k'}}_{\roo}\right|}_{=: \mathsf{C}_{xy}^{rp}} +
	 \underbrace{\left|\sum_{k,k' \in I(\gamma)} \cos(\omega_k nt) \cos(\omega_{k'}nt) 
	 \tilde{\varphi}^k_x 
	 \tilde{\varphi}^{k'}_y 
	 \expval{\hat{\tilde{\ssp}}_k
	 \hat{\tilde{\ssp}}_{k'}}_{\roo}\right|}_{=: \mathsf{C}_{xy}^{pp}},
	 \end{split}
	\end{equation}	  
	where we used the definition of $\tilde{\ssp}^{\bullet}_x(nt)$ \eqref{Qhighmodedef},
	as well as the definition of $\hat{\tilde{\ssp}}_k(nt)$ \eqref{Qtildesol}. Notice the
	definition of $\mathsf{C}^{pp}_{xy}$, $\mathsf{C}^{pr}_{xy}$, $\mathsf{C}^{rp}_{xy}$,
	and $\mathsf{C}^{rr}_{xy}$ in the above expression. Let us bound 
	$\mathsf{C}^{pp}_{xy}$. The other terms can be treated similarly.
	Define 
	$$\tilde{J}(x,y):= \{(z,z') \in \bI_n \times \bI_n \: \big| \: |x-z|\leq n^{\theta},
	|y-z'|\leq n^{\theta} \}.$$ 
	Plugging the definition of $\hat{\tilde{\ssp}}_k$ in $\mathsf{C}_{xy}^{pp}$, 
	rearranging the terms, using the linearity of trace,and bounding $|\cos(\cdot)|$ by 
	one we 
	have:
	\begin{equation} \label{Qcov1}
	\begin{split}
	&0 \leq \mathsf{C}^{pp}_{xy}=
	 \left|\sum_{z,z'=1}^n  \left(\sum_{k \in I(\gamma)} \cos(\omega_knt) 
	 \tilde{\varphi}^k_x \tilde{\varphi}^k_z \right)  
	 \left(\sum_{k' \in I(\gamma)} \cos(\omega_{k'}nt) 
	 \tilde{\varphi}^{k'}_y \tilde{\varphi}^{k'}_{z'} \right)
	 \expval{\tilde{\ssp}_z \tilde{\ssp}_{z'}}_{\roo}  \right|  \leq \\
	 &\underbrace{	
	 \sum_{(z,z')\in \tilde{J}(x,y)} \left(\sum_{k \in I(\gamma)}
	 	|\tilde{\varphi}^{k}_x \tilde{\varphi}^k_z| \right) 
	 	\left(\sum_{k' \in I(\gamma)}
	 	|\tilde{\varphi}^{k'}_y \tilde{\varphi}^{k'}_{z'}| \right) 
	 	\left|\expval{\tilde{\ssp}_z\tilde{\ssp}_{z'}}_{\roo} \right|}_{\mathsf{C}_{xy}^{pp,1}} + \\
	 	& \underbrace{	
	 \sum_{(z,z')\notin \tilde{J}(x,y)} \left(\sum_{k \in I(\gamma)}
	 	|\tilde{\varphi}^k_x \tilde{\varphi}^k_z| \right) 
	 	\left(\sum_{k' \in I(\gamma)}
	 	|\tilde{\varphi}^{k'}_y \tilde{\varphi}^{k'}_{z'}| \right) 
	 	\left|\expval{\tilde{\ssp}_z\tilde{\ssp}_{z'}}_{\roo} \right|}_{\mathsf{C}_{xy}^{pp,2}} 
	\end{split}
	\end{equation}
	Let us first bound $\mathsf{C}_{xy}^{pp,1}$. Each term involving sum over $k$ and 
	$k'$ is bounded by $1$ as we observed  in \eqref{eq:5}. Hence we have: 
	\begin{equation}  \label{Qcov2}
	0 \leq \mathsf{C}_{xy}^{pp,1} \leq \sum_{(z,z') \in \tilde{J}(x,y)}
	\left| \expval{\tilde{\ssp}_z \tilde{\ssp}_{z'}}_{\roo} \right| \leq 
	\sum_{(z,z')\in \tilde{J}(x,y)} \frac{C}{|z-z'|^2}\leq  
	\sum_{(z,z')\in \tilde{J}(x,y)} \frac{C}{n^{2\theta'}} =
	|\tilde{J}(x,y)| \frac{C}{n^{\theta'}} \leq  \frac{8C}{n^{2(\theta'-\theta)}},
	\end{equation}
	where in the third inequality we take advantage of the first bound in \eqref{decayassum}; moreover, in the fourth inequality we use 
	the fact that  
	 $|x-y|> 4n^{\theta'}$, and thanks to the definition of $\tilde{J}(x,y)$, for each 
	term in the above sum $|z-z'|>n^{\theta'}$. In the last inequality the definition 
	of  $\tilde{J}(x,y)$ is used.
	
	Now  we treat $\mathsf{C}_{xy}^{pp,2}$. Recall he definition of 
	$J(x)= \{ z \in \bI_n \: | \: |x-z| \leq n^{\theta} \}$. Thus
	for $n$ sufficiently large we get 
	\begin{equation}\label{Qeq:4}
	\begin{split}
	 &0 \leq \bE(\mathsf{C}_{xy}^{pp,2})\leq C 
	 \bE\left(\sum_{z \notin {J}(x)}\sum_{z'=1}^n
	\left(\sum_{k \in I(\gamma)}
	 	|\tilde{\varphi}^k_x \tilde{\varphi}^k_z| \right) 
	 	\left(\sum_{k' \in I(\gamma)}
	 	|\tilde{\varphi}^{k'}_y \tilde{\varphi}^{k'}_{z'}| \right)\right) + \\
	 	&C \bE\left( \sum_{z \in {J}(x)}\sum_{z'\notin J(y)}
	\left(\sum_{k \in I(\gamma)}
	 	|\tilde{\varphi}^k_x \tilde{\varphi}^k_z| \right) 
	 	\left(\sum_{k' \in I(\gamma)}
	 	|\tilde{\varphi}^{k'}_y \tilde{\varphi}^{k'}_{z'}| \right) \right)
	 	\leq C \bE \left(\sum_{(z,z')\in J^c(x) \times \bI_n}
	 	\left(\sum_{k \in I(\gamma)}
	 	|\tilde{\varphi}^k_x \tilde{\varphi}^k_z| \right)  \right) + \\
	 	& C \bE \left( \sum_{(z,z') \in J(x)\times J^c(y)} 
	 	\left(\sum_{k' \in I(\gamma)}
	 	|\tilde{\varphi}^{k'}_y \tilde{\varphi}^{k'}_{z'}| \right) \right)
	 	\leq 2n^2 C' \exp(-c (n^{\theta-n^{2 \gamma}})) \leq 
	 	 C' \exp(-\frac{c}2 n^{\theta}), 
	\end{split}
	\end{equation}
where first we bounded $|\expval{\tilde{\ssp}_z\tilde{\ssp}_{z'}}_{\roo}| \leq C$,
thanks to assumption \eqref{decayassum}, where $C>0$ is independent of $n$ and the 
realization of the masses; in the second line we take advantage of \eqref{eq:5}. In the  
third line  we used the properties of $J(x)$, $J(y)$ and the bound 
\eqref{localizationestimate} from Lemma 
\ref{localizationlemma}. Finally, in the last inequality we take advantage of the
choice of $0<2\gamma<\theta<1$ and the fact that $n$ is sufficiently large. Notice
that $C'>0$ and $c>0$ are independent of $n$.

Thanks to \eqref{Qeq:4}, for $n$ sufficiently large we have:
\begin{equation} \label{Qeq:5}
\bE \left(\frac{1}{n^2} \sum_{|x-y|>4 n^{\theta'}} \mathsf{C}_{xy}^{pp,2} \right) \leq 
C' \exp(-\frac{cn^{\theta}}{2}) \implies \frac{1}{n^2} \sum_{|x-y|>4 n^{\theta'}} \mathsf{C}_{xy}^{pp,2} \to 0,
\end{equation}  
almost surely w.r.t the distribution of the masses, thanks to Borel-Cantelli lemma (recall
that $C>0$, $c'>0$ are independent of $n$). Combining \eqref{Qcov2} and \eqref{Qeq:5},
using the fact that $\theta'>\theta$, and recalling the expression of  $\mathsf{C}_{xy}^{pp}$
 \eqref{Qcov1} we deduce:
 \begin{equation} \label{Qeq:6}
 \frac{1}{n^2} \sum_{|x-y|>n^{4\theta'}} \mathsf{C}_{xy}^{pp} \to 0,
\end{equation}  
almost surely w.r.t the distribution of the masses. We can prove similar identities for
$\sum \mathsf{C}_{xy}^{pr}$,  $\sum \mathsf{C}_{xy}^{pr}$, and $\sum \mathsf{C}_{xy}^{rr}$. To this end first we
should 
decompose them as in \eqref{Qcov1}. Then $\sum \mathsf{C}_{xy}^{pr,1}$ (and other similar terms)
can be treated exactly similar to \eqref{Qcov2}. However, the other terms such as 
$\sum \mathsf{C}_{xy}^{pr,2}$ need more attention: we should use Lemma \ref{boundomegalemma}
to treat this terms as we did in \eqref{eq:7}. Since these lengthy computation follow 
the exact same lines of the previous cases, and they do not contain any novel point,
 we omit them for the sake of brevity. Therefore, thanks to the \eqref{Qeq:1}, by
  bounding $f(\frac{x}{n})$ and $m_x$ by uniform constants, we have: 
 \begin{equation} \label{Qhighbound1}
 	\left|\frac{1}{n^2} \sum_{|x-y|>4n^{\theta'}} f(\frac{x}{n}) f(\frac{y}{n})
 	\expval{\tilde{\ssp}_x^{\bullet}(nt) \tilde{\ssp}_y^{\bullet}(nt)}_{\roo} \right|
 	\leq  \frac{C}{n^2} \sum_{|x-y|>4n^{\theta'}} 
 	\left(\mathsf{C}_{xy}^{pp}+\mathsf{C}_{xy}^{pr} + \mathsf{C}_{xy}^{rp}
 	+\mathsf{C}_{xy}^{rr}   \right) \to 0,
 \end{equation}
 almost surely w.r.t the distribution of the masses.
 Following same steps one can show that
 \begin{equation} \label{Qhighbound22}
 \left|\frac{1}{n^2} \sum_{|x-y|>4n^{\theta'}} f(\frac{x}{n}) f(\frac{y}{n})
 	\expval{\tilde{\ssr}_x^{\bullet}(nt) \tilde{\ssr}_y^{\bullet}(nt)}_{\roo} \right|\to 0
\end{equation}
almost  surely w.r.t.the distribution of the masses. Since this proof is identical, we do 
not bring it here.  This finishes the proof of this step.

\paragraph{Step 2.2: High modes $|x-y|\leq 4n^{\theta'}$.}
Let us bound $\frac{1}{m_x}\expval{(\tilde{\ssp}_x^{\bullet}(nt))^2}_{\roo}$: recall 
the expression \eqref{Qeq:1}, let $x=y$, and bound $(m_x)^{-1}$ by a constant; then we have:
\begin{equation} \label{Qhighbound2}
\frac{\expval{(\tilde{\ssp}_x^{\bullet}(nt))^2}_{\roo}}{m_x} \leq m_x(\mathsf{C}_{xx}^{pp}+
\mathsf{C}_{xx}^{pr}+\mathsf{C}_{xx}^{rp}+\mathsf{C}_{xx}^{rr}), 
\end{equation}
  where definition of $\mathsf{C}_{xx}^{\sharp_1 \sharp_2}$ for 
   $\sharp_1,\sharp_2\in \{r,p\}$ should be recalled from \eqref{Qeq:1}. Let us bound 
   $m_x\mathsf{C}_{xx}^{pr}$ by a constant uniform in $n$ and realization of the masses:
   Similar to \eqref{Qcov1}, by using the definition of $\hat{\tilde{\ssp}}_k$, and 
   $\hat{\tilde{\ssr}}_k$ we get (recall the convention $\tilde{\ssr}_n= \phi^k_n=0$)
   \begin{equation} \label{Qeq:7}
   \begin{split}
		&m_x\mathsf{C}_{xx}^{pr}= \left|\sum_{z,z'=1}^n 
		\underbrace{\left( \sum_{k \in I(\gamma)} \cos(\omega_k nt) \varphi^k_x \varphi^k_z\right)}_{=:\mathsf{a}_z}
		\underbrace{\left( \sum_{k'\in I(\gamma)} \sin(\omega_{k'} nt) \varphi^{k'}_x \phi^{k'}_{z'}
		\right)}_{=: \mathsf{b}_{z'}} \expval{\frac{\tilde{\ssp}_z}{\sqrt{m_z}}\tilde{\ssr}_{z'}}_{\roo} \right|    		
   \end{split}
   \end{equation}
 Notice the definition of $\ssa_z$, $\ssb_z$ in the above expression. Then we have:
 \begin{equation} \label{Qeq:8}
 \begin{split}
	&m_x \mathsf{C}_{xx}^{pr} \leq \sum_{z=1}^n |\ssa_z| |\ssb_z|
	\left| \expval{\frac{\tilde{\ssp}
	_z}{\sqrt{m_z}}\tilde{\ssr}_{z'}}_{\roo} \right|+
	\sum_{l=1}^{n-1} \sum_{|z-z'|=l}  \left|\expval{\frac{\tilde{\ssp}_z}{\sqrt{m_z}}\tilde{\ssr}_{z'}}
	_{\roo} \right| |\ssa_z||\ssb_{z'}| \leq  \\ &C
	\sum_{z=1}^n |\ssa_z||
	\ssb_{z}| + \sum_{l=1}^{n-1} \frac{C}{l^2} \sum_{z=1}^{n-l} \left(
	|\ssa_z||\ssb_{z+l}| +|\ssa_{z+l}||\ssb_z| \right)\leq 
	C \left(\sum_{z=1}^n|\ssa_z|^2\right)^{\frac12}
	 \left(\sum_{z=1}^n|\ssb_z|^2\right)^{\frac12} + \\
	 &   \sum_{l=1}^{n-1} \frac{C}{l} \left(\left(\sum_{z=1}^{n-l} |\ssa_z|^2 \right)^{\frac12}
	   \left(\sum_{z=1}^{n-l} |\ssb_{z+l}|^2 \right)^{\frac12} 
	   + \left(\sum_{z=1}^{n-l} |\ssa_{z+l}|^2 \right)^{\frac12}
	   \left(\sum_{z=1}^{n-l} |\ssb_{z}|^2 \right)^{\frac12}  \right) 
	   \leq \\& C \left(\sum_{z=1}^n|\ssa_z|^2\right)^{\frac12}
	 \left(\sum_{z=1}^n|\ssb_z|^2\right)^{\frac12} +
	 \sum_{l=1}^{n-1}  \frac{2C}{l^2} \left(\sum_{z=1}^n|\ssa_z|^2\right)^{\frac12}
	 \left(\sum_{z=1}^n|\ssb_z|^2\right)^{\frac12},
 \end{split}
 \end{equation}
 where in the second line we used the assumption \eqref{decayassum}, and then we take
 advantage of two Cauchy-Schwartz inequality. Recalling the definition of $\ssb_z$ 
 \eqref{Qeq:7} we have (recall the convention $\phi^k_n=0$): 
 \begin{equation} \label{Qeq:9}
 \begin{split}
	&\sum_{z=1}^n  \ssb_z^2= \sum_{z=1}^{n-1} \sum_{k,k'\in I(\gamma)}
	\sin(\omega_knt) \sin(\omega_{k'}nt) \varphi^k_x\varphi^{k'}_x \phi^k_z \phi^{k'}_z
	= \sum_{k,k'\in I(\gamma)}
	\sin(\omega_knt) \sin(\omega_{k'}nt) \varphi^k_x\varphi^{k'}_x\sum_{z=1}^{n-1}
	 \phi^k_z \phi^{k'}_z = \\
	 & \sum_{k\in I(\gamma)}\sin^2(\omega_k nt) (\varphi^k_x)^2 \leq \sum_{k=0}^{n-1}
	 (\varphi^k_x)^2 =1,
 \end{split}
 \end{equation}
	 where we used the fact that $\phi^k$ and $\varphi^k$ are orthonormal basis
	 ($\sum_{z=1}^{n-1}
	 \phi^k_z \phi^{k'}_z= \delta_{k,k'}$, $\sum_{k=0}^{n-1}
	 (\varphi^k_x)^2 =1$). Similarly, we have $\sum_{z=1}^n \ssa_z^2 \leq 1$. Plugging 
	 the latter  and \eqref{Qeq:9} inside \eqref{Qeq:8}, we get: 
	 \begin{equation} \label{Qeeq111}
	 m_x \mathsf{C}_{xx}^{pr} \leq C + \sum_{l=1}^{n-1} \frac{2C}{l^2} \leq c',
	 \end{equation}
	 where as usual $c'>0$ is a constant independent of $n$ and the realization of 
	 the masses. Similarly, we can bound other terms in \eqref{Qhighbound2} by 
	 constants independent of $n$ and the realization of the masses. Notice that this 
	 constants can be chosen to be independent of $x$ as well. Moreover, the exact 
	 same argument works for $\expval{(\tilde{\ssr}_x(nt)^{\bullet})}_{\roo}$. Hence, 
	 there exists a constant $C>0$ such that for $n$, any $x \in \bI_n$, and 
	 any realization of the masses we have:
	 \begin{equation} \label{Qhighbound3}
	 \frac{\expval{(\tilde{\ssp}_x^{\bullet}(nt))^2}_{\roo}}{m_x} \leq C, 
	 \quad \expval{(\tilde{\ssr}_x^{\bullet}(nt))^2}_{\roo} \leq C.
\end{equation}	  
	Before proceeding, notice that in the above argument, we can replace 
	$\tilde{\ssp}^{\bullet}_x(nt)$ with $\tilde{\ssp}_x(nt)$ and all the steps work.
	Therefore, exists $C>0$ such that for $n$, any $x \in \bI_n$, and 
	 any realization of the masses we have:
	 \begin{equation} \label{Qhighbound4}
	 \frac{\expval{(\tilde{\ssp}_x(nt))^2}_{\roo}}{m_x} \leq C, 
	 \quad \expval{(\tilde{\ssr}_x(nt))^2}_{\roo} \leq C.
\end{equation}
Having \eqref{Qhighbound3}, bounding $f(x/n)/(m_x)^{\frac12}$, and a Cauchy-Schwartz  \eqref{CSineq} yields:
\begin{equation} \label{Qhighbound6}
\begin{split}
\left| \frac{1}{n^2} \sum_{|x-y|\leq 4 n^{\theta'}}
f(\frac{x}{n}) f(\frac{y}{n})  \expval{\tilde{\ssp}_x^{\bullet}(nt)
\tilde{\ssp}_y^{\bullet}(nt)}_{\roo} \right| \leq \frac{C}{n^2} \sum_{|x-y| \leq 4n^{\theta'}} \left|\frac{\expval{\tilde{\ssp}_x^{\bullet}(nt)\tilde{\ssp}_y^{\bullet}(nt)}_{\roo}}{\sqrt{m_x m_y}} 
\right|  \leq \\
\frac{C'}{n^2} \sum_{|x-y|\leq 4n^{\theta'}}
 \left(\frac{\expval{(\tilde{\ssp}_x^{\bullet}(nt))^2}_{\roo}}{m_x} \right)^{\frac12} 
 \left(\frac{\expval{(\tilde{\ssp}_y^{\bullet}(nt))^2}_{\roo}}{m_y} \right)^{\frac12} 
 \leq \frac{C'}{n^2} \sum_{|x-y|\leq 4n^{\theta'}} C \leq \frac{C''}{n^2} \times cn^{1+\theta'}\to 0,
 \end{split}
 \end{equation}
as $n \to \infty$, thanks to the choice of $0<\theta'<1$. 
Similar identity holds for $\ssr$:
\begin{equation}\label{Qhighbound7}
\left| \frac{1}{n^2} \sum_{|x-y|\leq 4 n^{\theta'}}
f(\frac{x}{n}) f(\frac{y}{n})  \expval{\tilde{\ssr}_x^{\bullet}(nt)
\tilde{\ssr}_y^{\bullet}(nt)}_{\roo} \right|\to 0.
\end{equation} 
Combining  \eqref{Qhighbound1}, and \eqref{Qhighbound6} we can deduce 
$\mathsf{U}_n \to 0$ almost surely w.r.tthe distribution of the masses, similarly 
\eqref{Qhighbound22}, and \eqref{Qhighbound7} gives us $\mathsf{U}_n' \to 0$ almost surely.
This conclude the proof of \eqref{Qhighmodeleprp} and this step.

\paragraph{Step 3: Summing up.}
Recall  $\mathsf{P}_n$ \eqref{eqQ:proof3}.
By using the fact that $\tilde{\ssp}_x(nt)=\tilde{\ssp}_x^{o}(nt)+
\tilde{\ssp}_x^{\bullet}(nt)$, and a Cauchy-Schwartz 
($|\expval{AB}_{\roo}| \leq \expval{AA^*}_{\roo}^{\frac12} \expval{BB^*}^{\frac12}_{\roo}$) inequality we observe that $\mathsf{P}_n \leq \mathsf{L}_n + \mathsf{U}_n +
(\mathsf{L}_n)^{\frac12}(\mathsf{U}_n)^{\frac12}$. Combining this with 
\eqref{Qlowmodelemm}, and  \eqref{Qhighmodeleprp} we deduce that  
$\mathsf{P}_n \to 0$ almost surely w.r.tthe distribution of the masses. Similarly, 
we can deduce that $\mathsf{R}_n \to 0 $ almost  surely w.r.tthe distribution of the 
masses. This finishes the proof of \eqref{eqQ:proof3} and consequently \eqref{highermomentclassicalrQ}, \eqref{highermomentclassicalpQ}. 
\end{proof}

\subsection{Evolution of $e$}
Recall average variables $\bar{\ssp}_x(nt)$, $\bar{\ssr}_x(nt)$ in \eqref{avgfluc}.
As we mentioned, they satisfy the same system of ODEs as in \cite{A20} up to a vanishing constant. Therefore, by using the regularity of initial profiles 
$\bar{r}$, $\bar{p}$, we can adapt the argument of \cite{A20}, and deduce the following (cf. \cite{A20}, Sect.4, (4.29)) for any $y \in (0,1)$:
\begin{equation} \label{Qmechanical1}
\bar{\ssr}_{[ny]}(nt) \to \fr(y,t), \quad \frac{\bar{\ssp}_{[ny]}}{m_{[ny]}}\to
\frac{\fp(y,t)}{\bar{m}},
\end{equation}
almost surely w.r.t the distribution of the masses. Therefore, by using the averaging 
lemma of \cite{A20} (cf. (4.13) in \cite{A20}) we can deduce the convergence of the average of
mecanical energy in the following sense (cf. Lemma 6.1 in \cite{A20}): for  any $f \in C^0([0,1])$
\begin{equation} \label{Qmechnical2}
\frac{1}{n}  \sum_{x=1}^n f(\frac{x}{n})
\left( \frac{\expval{\ssp_x(nt)}_{\roo}^2}{2m_x}
+ \frac{\expval{\ssr_x(nt)}_{\roo}^2}{2}\right) \to \int_0^1 f(y)
\left(\frac{\fp(y,t)}{2 \bar{m}}+\frac{\fr^2(y,t)}{2} \right).
\end{equation} 
On the other hand, the evolution of fluctuation variables \eqref{Qtildesol} is 
identical
to the one in \cite{A20}. Although, the ensemble average here is different from 
\cite{A20}, one can observe that having the assumption \eqref{avgassume}, and the 
bounds \eqref{Qpkbound},\eqref{Qrkbound}, and \eqref{Qhighbound1} are sufficient for
using the result of \cite{A20}. More precisely, recall the fluctuation operators 
\eqref{avgfluc} $\tilde{\ssp}_x(nt)$, $\tilde{\ssr}_x(nt)$; having the aforementioned bounds
we can apply Lemma 6.4. of \cite{A20} and deduce for any $f \in C^0([0,1])$:
\begin{equation} \label{Qthermal1}
\begin{split}
&T_n(t)-T_n(0):=\\ &\frac{1}{n} \sum_{x=1}^n f(\frac{x}{n}) 
\left(\frac{\expval{\tilde{\ssp}_x^2(nt)}_{\roo}}{2m_x} + 
\frac{\expval{\tilde{\ssr}_x^2(nt)}_{\roo}}{2} \right) - 
\frac{1}{n} \sum_{x=1}^n f(\frac{x}{n}) 
\left(\frac{\expval{\tilde{\ssp}_x^2(0)}_{\roo}}{2m_x} + 
\frac{\expval{\tilde{\ssr}_x^2(0)}_{\roo}}{2} \right) \to 0, 
\end{split}
\end{equation}
almost surely w.r.t.\@ the distribution of the masses. 
Combining the latter with \eqref{avgassume} yields $T_n(t) \to \int_0^1 f(y)\bar{b}(y)dy$. 
Combining this with \eqref{Qmechnical2}, 
adding and subtracting the average and comparing the obtained expression with explicit solution to the macroscopic evolution equation \eqref{eq: Euler equations},
i.e.\@ $2\fe(y,t)=\fp^2(y,t)/\bar{m}+\fr^2(y,t)+2\bar{b}(y)$, we have the quantum counterpart of \eqref{eq: limit on average} for the energy: 
%for the general state $\roo$ satisfying assumptions 1-2 and first bound of assumption 3: for any $f \in C^0([0,1])$:
\begin{equation} \label{hydroeQ}
\frac{1}{n} \sum_{x=1}^n f(\frac{x}{n}) \expval{\sse_x(nt)}_{\roo} =
 \frac{1}{n} \sum_{x=1}^n f(\frac{x}{n}) 
\left(\frac{\expval{\ssp_x^2(nt)}_{\roo}}{2m_x} + 
\frac{\expval{\ssr_x^2(nt)}_{\roo}}{2} \right) \to \int_0^1 f(y) \fe(y,t)dy,
\end{equation}
almost surely w.r.t the distribution of the masses.
Having \eqref{hydroeQ}, we can complete the proof of Theorem \ref{mainthmQ}.

Before proceeding, notice that  thanks to 
	\eqref{CCR} and \eqref{Qevolution} we have for all $k,k' \in \bI_n^o$,
	and all $x,y \in \bI_n$ : 
	\begin{equation} \label{CCRt}
	\begin{split}	
	[\hat{\ssp}_k(nt),\hat{\ssp}_{k'}(nt)]=[\hat{\ssr}_k(nt),\hat{\ssr}_{k'}(nt)]=0, 
	\: \implies  [\tilde{\ssp}_x(nt),\tilde{\ssp}_y(nt)]= [\tilde{\ssr}_x(nt),
	\tilde{\ssr}_y(nt)]= 0, \\
	 [\tilde{\ssp}^o_x(nt),\tilde{\ssp}^o_y(nt)]=
	[\tilde{\ssr}^o_x(nt),\tilde{\ssr}^o_y(nt)]=[\tilde{\ssp}^{\bullet}_x(nt),
	\tilde{\ssp}
	^{\bullet}_y(nt)]=[\tilde{\ssr}^{\bullet}_x(nt),
	\tilde{\ssr}
	^{\bullet}_y(nt)]=0, \\
	[\tilde{\ssp}^{\bullet}_x(nt),\tilde{\ssp}^o_y(nt)]=[\tilde{\ssr}^{\bullet}_x(nt),\tilde{\ssr}^o_y(nt)]=0.
	\end{split}	
	\end{equation}	 
	In the following, we will use above identities without mentioning them.

\begin{proof}[Proof of Theorem \ref{mainthmQ}: \eqref{highermomentclassicaleQ}]
Let us define the following operator:
\begin{equation} \label{Qtildeedef}
\tilde{\sse}_x(nt):= \sse_x(nt)-\expval{\sse_x(nt)}_{\roo}.
\end{equation}
At $t=0$, notice the difference among $\tilde{\sse}_x$ and $\breve{\sse}_x$ 
\eqref{tildedef}. 
As we did before, adding and subtracting the average term, and using 
\eqref{hydroeQ}, proof of \eqref{highermomentclassicaleQ} boils down to proving that
\begin{equation} \label{Qeproof}
\mathsf{E}_n:= \expval{\left(\frac{1}n \sum_{x=1}^n f(\frac{x}{n}) \tilde{\sse}_x(nt)
\right)^2}_{\roo} \to 0,
\end{equation}
almost surely w.r.t the distribution of  the masses. By using the definition of
$\tilde{\sse}_x(nt)$ and a Cauchy-Schwartz inequality \eqref{CSineq}, it is 
sufficient to prove: 

\begin{equation} \label{Qerpproof1}
\begin{split}
\mathtt{P}_n := \expval{\left(\frac1n \sum_{x=1}^n f(\frac{x}{n})
\left(\frac{\ssp_x^2(nt)}{m_x} -\expval{\frac{\ssp_x^2(nt)}{m_x}}_{\roo} \right)\right)^2}_{\roo}\to 0, \\
\mathtt{R}_n := \expval{\left(\frac1n \sum_{x=1}^n f(\frac{x}{n})
\left({\ssr_x^2(nt)} -\expval{{\ssr_x^2(nt)}}_{\roo} \right)\right)^2}_{\roo} \to 0,
\end{split}
\end{equation} 
almost surely w.r.t  the distribution of  the masses. First, recall average variables
and fluctuation operators \eqref{avgfluc},  observe that 
computation of \eqref{eq:13} is valid in the quantum case as well: 
\begin{equation} \label{Qeq:13}
	\begin{split}	
	&\expval{\left(\ssp_x^2(nt)- \expval{\ssp_x^2(nt)}_{\roo}\right)\left(
	\ssp_y^2(nt)-\expval{\ssp_y^2(nt)}_{\roo}\right)}_{\roo}=
	\\ &\expval{\tilde{\ssp}^2_x(nt)
	\tilde{\ssp}^2_y(nt)}_{\roo}- \expval{\tilde{\ssp}^2_x(nt)}_{\roo}
	\expval{\tilde{\ssp}^2_y(nt)}_{\roo} +4\bar{\ssp}_x(nt)\bar{\ssp}_y(nt) 
	\expval{\tilde{\ssp}_x(nt)\tilde{\ssp}_y(nt)}_{\roo}+ \\
	&2 \bar{\ssp}_x(nt) \expval{\tilde{\ssp}_x(nt) \tilde{\ssp}_y^2(nt)}_{\roo} + 
	2\bar{\ssp}_y(nt) \expval{\tilde{\ssp}^2_x(nt)
	\tilde{\ssp}_y(nt)}_{\roo}.
	\end{split}
	\end{equation}
	In contrast  to the classical case, we cannot simply do pairings in $\expval{\cdot}_{\roo}$. 
	Therefore, controlling each term in \eqref{Qeq:13} is more involved. 
	We control each term separately: decompose them into low and high 
	modes, and take advantage of the localization for high modes with $|x-y|$ large enough, and control the remainder. 
	
	We only prove \eqref{Qerpproof1} for $\mathsf{P}_n$, thanks to symmetry in our 
	assumption, the proof for $\mathsf{R}_n$ is identical, up to small modifications which we will omit.
	By using a Cauchy-Schwartz inequality and recalling the definition of 
	$\bar{\ssp}_x(nt)$, $\tilde{\ssp}_x(nt)$ \eqref{avgfluc}, for obtaining 
	\eqref{Qerpproof1} it  is sufficient to prove: 
	\begin{equation}\label{Qprproof2}
	\begin{split}	
	\mathsf{P}_n':= \expval{\left(\frac{1}n \sum_{x=1}^n f(\frac{x}{n})
	\left(\frac{\tilde{\ssp}_x^2(nt)}{m_x}-
	\frac{\expval{\tilde{\ssp}_x^2(nt)}_{\roo}}{m_x}\right)\right)^2}_{\roo} \to 0, \\
	\mathsf{P}_n'' :=  \expval{\left(\frac1n \sum_{x=1}^n
	f(\frac{x}{n}) \frac{\bar{\ssp}_x(nt)\tilde{\ssp}_x(nt)}{m_x}\right)^2}_{\roo}
	\to 0, 	
	\end{split}	
	\end{equation}		
	as $n\to \infty$ almost surely w.r.t distribution of the masses.  The proof of second 
	identity is straightforward: thanks to \eqref{Qmechanical1}, and regularity of  
	the macroscopic equation, $\bar{\ssp}_x/m_x$ is uniformly (in $n$ and distribution
	of the masses) bounded (cf. \cite{A20} (4.11)). On the other hand, proof of 
	\eqref{eqQ:proof3} only rests on the fact that $f$ is bounded and does not use its 
	regularity. Therefore, we absorb $\bar{\ssp}_x/m_x$ inside $f$ and follow similar  
	steps and thanks to \eqref{eqQ:proof3} deduce that: 
	\begin{equation} \label{Qprproof3}
	\mathsf{P}''_n \to 0, 
	\end{equation}	  
	almost surely w.r.t the distribution of the masses. 
	
	Fix $0<2\gamma<\theta<\theta'<1$ (we choose them opportunely later), and recall
	the definition of $\tilde{\ssp}_x^o(nt)$, $\tilde{\ssp}^{\bullet}_x(nt)$ 
	\eqref{Qlowmodedef}, \eqref{Qhighmodedef} (recall 
	$I(\gamma)=]n^{1-\gamma},n] \cap \bZ$, $\tilde{I}(\gamma)=[1,n^{1-\gamma}]
	\: \cap \bZ$). Thanks to a Cauchy-Schwartz inequality proof of the first limit in 
	\eqref{Qprproof2} reduces to the following :
	\begin{equation} \label{Qepproof3}
	\begin{split}
		\Pi_n^o := \expval{\left(\frac1n \sum_{x=1}^n 
		\frac{f(\frac{x}{n})}{m_x} \left((\tilde{\ssp}_x^o)^2(nt) - 
		\expval{(\tilde{\ssp}_x^o)^2(nt)}_{\roo}\right) \right)^2}_{\roo} \to 0,\\
		\Pi_n^{\bullet} := \expval{\left(\frac1n \sum_{x=1}^n 
		\frac{f(\frac{x}{n})}{m_x} \left((\tilde{\ssp}_x^{\bullet})^2(nt) - 
		\expval{(\tilde{\ssp}_x^{\bullet})^2(nt)}_{\roo}\right)
		\right)^2}_{\roo}	\to 0,\\
		\Pi^{o \bullet}:= \expval{\left(\frac1n \sum_{x=1}^n 
		\frac{f(\frac{x}{n})}{m_x} \left(\tilde{\ssp}_x^o(nt) \tilde{\ssp}_x^{\bullet}(nt) - 
		\expval{\tilde{\ssp}_x^o(nt)\tilde{\ssp}_x^{\bullet}(nt)}_{\roo}\right) \right)^2}
		_{\roo} \to 0, 
	\end{split}
	\end{equation}	 
	 as $n \to \infty$ almost surely w.r.t the distribution of the masses. Let us begin
	 with the first limit. To this end, we need following bounds:
	Thanks to assumption 3 \eqref{decayassum}, and by a computation similar to  
	 \eqref{Qpkbound} we have: 
	 \begin{equation} \label{Qk4bound}
	 \begin{split}
		&\expval{\hat{\tilde{\ssp}}_k^4}_{\roo} \leq c\sum_{z_1,z_2,z_3,z_4} 
		|\varphi^k_{z_1}\varphi^k_{z_2}\varphi^k_{z_3}\varphi^k_{z_4}|
		\left|\expval{\tilde{\ssp}_{z_1}\tilde{\ssp}_{z_2}\tilde{\ssp}_{z_3}
		\tilde{\ssp}_{z_4}}_{\roo}\right| \leq 	 \\
		&c \sum_{\sigma}
		\sum_{z_{\sigma_1},z_{\sigma_2}}\frac{C}{|z_{\sigma_1}-z_{\sigma_2}|^2}
		|\varphi^k_{z_{\sigma_1}} \varphi^k_{z_{\sigma_2}}|
		\sum_{z_{\sigma_3},z_{\sigma_4}}\frac{C}{|z_{\sigma_3}-z_{\sigma_4}|^2}
		|\varphi^k_{z_{\sigma_3}} \varphi^k_{z_{\sigma_4}}| \leq C,
	\end{split}
	 \end{equation}
	where in the first line we used \eqref{decayassum}, in the second line we sum 
	over three following permutations $\sigma: \{1,2,3,4\} \to \{1,2,3,4\}$:
	$(12)(34)$, $(13)(24)$, and $(14)(23)$. The last inequality is obtained similar 
	to \eqref{Qpkbound}. Notice that $C$ is independent of $n$, and realization of 
	the masses. Thanks to symmetry of our assumption, we can deduce similar 
	bounds for $\expval{(\hat{\tilde{\ssp}}_k)^j(\hat{\tilde{\ssr}}_k)^{4-j}}_{\roo}$, 
	for $j \in \{0,1,2,3,4\}$. Therefore, we can deduce similar bounds at time $nt$:
	there exists $C>0$ such that for any $n$, $k \in \bI_n^o$, and any realization 
	of the masses:
	\begin{equation} \label{Qk4boundt}
		\expval{\hat{\tilde{\ssp}}_k^4(nt)}_{\roo} \leq C, \quad,
		\expval{\hat{\tilde{\ssr}}_k^4(nt)}_{\roo} \leq C.
	\end{equation}	 
	Recall $\tilde{I}(\gamma)=[0,n^{1-\gamma}] \cap \bZ$ and observe for any constant 
	$C>0$: 
	\begin{equation} \label{Qelowmodes}
		\begin{split}
			& \frac{C}{n^2} \sum_{x,y=1}^n \left|\expval{(\tilde{\ssp}^o_x)^2(nt)
	(\tilde{\ssp}^o_y)^2(nt)}_{\roo}- \expval{(\tilde{\ssp}^o_x)^2(nt)}_{\roo}
	\expval{(\tilde{\ssp}^o_y)^2(nt)}_{\roo}\right| \leq \\& \frac{C'}{n^2}\sum_{x,y=1}^n\frac{1}{m_xm_y}
	 \left(\expval{(\tilde{\ssp}^o_x)^2(nt)
	(\tilde{\ssp}^o_y)^2(nt)}_{\roo}+ \expval{(\tilde{\ssp}^o_x)^2(nt)}_{\roo}
	\expval{(\tilde{\ssp}^o_y)^2(nt)}_{\roo} \right) \leq \\
	&  \frac{C''}{n^2}\sum_{x,y=1}^n \sum_{k_1,k_2,k_3,k_4 \in \tilde{I}(\gamma)}
	\varphi^{k_1}_x \varphi^{k_2}_x	\varphi^{k_3}_y \varphi^{k_4}_y	
	\bigg(\expval{\hat{\tilde{\ssp}}_{k_1}(nt)\hat{\tilde{\ssp}}_{k_2}(nt)\hat{\tilde{\ssp}}_{k_3}(nt)\hat{\tilde{\ssp}}_{k_4}(nt)}_{\roo} +\\
	&\expval{\hat{\tilde{\ssp}}_{k_1}(nt)\hat{\tilde{\ssp}}_{k_2}(nt)}_{\roo}\expval{\hat{\tilde{\ssp}}_{k_3}(nt)\hat{\tilde{\ssp}}_{k_4}(nt)}_{\roo} \bigg)=\\
	&
	\frac{C''}{n^2}\sum_{k_1,k_3 \in \tilde{I}(\gamma)} \bigg(
	\expval{(\hat{\tilde{\ssp}}_{k_1}(nt))^2(\hat{\tilde{\ssp}}_{k_3}(nt))^2}_{\roo} 	
	+\expval{(\hat{\tilde{\ssp}}_{k_1}(nt))^2}
	\expval{(\hat{\tilde{\ssp}}_{k_3}(nt))^2}_{\roo}
		\bigg)	 \leq \frac{\tilde{C} n^{2-2\gamma}}{n^2} \to 0,	
		\end{split}
	\end{equation}	 
	as $n\to \infty$, where we used the fact that  $\sum_{x} \varphi^k_x
	 \varphi^{k'}_x=\delta_{k,k'}$, a Cauchy-Schwartz inequality and bounds 
	 \eqref{Qk4boundt}, \eqref{Qtrpkbound}, as well as the fact that $\gamma>0$.
	 Expanding the square in $\Pi^{o}_n$, bounding $|f(\frac{x}{n})/m_x|$, and 
	 taking advantage of \eqref{Qelowmodes} we deduce that $\Pi_n^o \to 0 $. \\
	 In order to prove two other limits, we need the following bound: there exists
	 $C>0$ such that for any $n$, $x \in \bI_n$, and any realization of the masses:
	 \begin{equation} \label{Q4xbound}
	\begin{split}	 
	 \expval{\tilde{\ssp}_x^4(nt)}_{\roo} \leq C, 
	 \quad \expval{(\tilde{\ssp}_x^o)^4(nt)}_{\roo} \leq C, \quad
	 \expval{(\tilde{\ssp}_x^{\bullet})^4(nt)}_{\roo} \leq C, \\
	 \expval{\tilde{\ssr}_x^4(nt)}_{\roo} \leq C, 
	 \quad \expval{(\tilde{\ssr}_x^o)^4(nt)}_{\roo} \leq C, \quad
	 \expval{(\tilde{\ssr}_x^{\bullet})^4(nt)}_{\roo} \leq C,
	\end{split}	 
	 \end{equation}
	we only show the above bound for $\expval{\tilde{\ssp}_x^4(nt)}_{\roo}$,
	other bound can be obtained similarly. 
	Notice that we have:
	\begin{equation}\label{Qeeq:1}
	\frac{\expval{\tilde{\ssp}_x^4(nt)}_{\roo}}{m_x}=
	\sum_{k_1,k_2,k_3,k_4} \varphi^{k_1}_x\varphi^{k_2}_x \varphi^{k_3}_x \varphi^{k_4}_x
	\expval{\hat{\tilde{\ssp}}_{k_1}(nt) \hat{\tilde{\ssp}}_{k_2}(nt)
	\hat{\tilde{\ssp}}_{k_3}(nt) \hat{\tilde{\ssp}}_{k_4}(nt)}_{\roo}.
	\end{equation}
	Considering the definition of  $\hat{\tilde{\ssp}}_{k_1}(nt)$ \eqref{Qtildesol},
	the last average in the above expression contains $16$ terms. However, thanks to 
	symmetry of our assumption it is sufficient to  bound one of them and the rest are
	treated similarly. Let us take the term involving only $\ssp$ terms. Let us denote 
	\begin{equation} \label{Qea}
	\ssa_z^x := \sum_{k} \varphi^k_x \varphi^k_z \cos(\omega_k nt).  	
	\end{equation}	
	Then by using the definition of $\hat{\tilde{\ssp}}_k$ we have:
	\begin{equation} \label{Qeeq:2}
		\begin{split}
	&\sum_{k_1,k_2,k_3,k_4} \varphi^{k_1}_x\varphi^{k_2}_x \varphi^{k_3}_x \varphi^{k_4}_x
	\cos(\omega_{k_1}nt) \cos(\omega_{k_2}nt) \cos(\omega_{k_3}nt) 
	\cos(\omega_{k_4}nt) 	
	\expval{\hat{\tilde{\ssp}}_{k_1}\hat{\tilde{\ssp}}_{k_2} 
	\hat{\tilde{\ssp}}_{k_3} \hat{\tilde{\ssp}}_{k_4}}_{\roo}= \\
 &	\sum_{z_1,z_2,z_3,z_4=1}^n \expval{\frac{\tilde{\ssp}_{z_1}}{
	\sqrt{m_{z_1}}}\frac{\tilde{\ssp}_{z_2}}{
	\sqrt{m_{z_2}}}\frac{\tilde{\ssp}_{z_3}}{
	\sqrt{m_{z_3}}}\frac{\tilde{\ssp}_{z_4}}{
	\sqrt{m_{z_4}}}}_{\roo} \ssa_{z_1}^x \ssa_{z_2}^x \ssa_{z_3}^x \ssa_{z_4}^x
	\leq \\ &\sum_{\sigma} \sum_{z_{\sigma_1},z_{\sigma_2}} \frac{c}{|z_{\sigma_1}-
	z_{\sigma_2}|^2} |\ssa_{z_{\sigma_1}}||\ssa_{z_{\sigma_2}}| 
	\sum_{z_{\sigma_3},z_{\sigma_4}} \frac{c}{|z_{\sigma_3}-
	z_{\sigma_4}|^2} |\ssa_{z_{\sigma_3}}||\ssa_{z_{\sigma_4}}|\leq C,
	\end{split}
	\end{equation}
	where $\sum_{\sigma}$ is a summation over the three above-mentioned permutations;
	in the second line we used the assumption \eqref{decayassum}. Comparing each sum 
	($\sum_{z_{\sigma_i},z_{\sigma_j}}$) with the expression appearing in 
	\eqref{Qeq:7}, and using the bound \eqref{Qeeq111}  
	it is straightforward to observe that the above expression is bounded by a constant.
	 Notice that as in \eqref{Qeq:7} our argument can 
	be adapted for $\tilde{\ssp}^{\bullet}$, and $\tilde{\ssp}^o$. Hence,we can deduce 
	\eqref{Q4xbound}.
	
	 Recall $0<2 \gamma <\theta <\theta' <1$. 
	 Bounding $f(\frac{x}{n})/m_x$, taking advantage of the Cauchy-
	Schwartz inequality \eqref{CSineq}, and using the bounds \eqref{Qhighbound3},
	and \eqref{Q4xbound} we have: 
	\begin{equation} \label{Qepproof4}
			\left|\frac{1}{n^2} \sum_{|x-y|\leq 16 n^{\theta'}} 
		\frac{f(\frac{x}{n})f(\frac{y}{n})}{m_xm_y} 
		\left((\tilde{\ssp}_x^{\bullet})^2(nt) - 
		\expval{(\tilde{\ssp}_x^{\bullet})^2(nt)}_{\roo}\right)	
		\left((\tilde{\ssp}_y^{\bullet})^2(nt) - 
		\expval{(\tilde{\ssp}_y^{\bullet})^2(nt)}_{\roo}\right)\right|
		\leq C\frac{n^{\theta'}}{n} \to 0,
	\end{equation}
	as $n \to \infty$. 
	Now take $|x-y|> 4n^{\theta'}$, and observe:
	\begin{equation} \label{Qeproof5}
	\begin{split}
		&\left|\expval{(\tilde{\ssp}_x^{\bullet})^2(nt)
		(\tilde{\ssp}_y^{\bullet})^2(nt)}_{\roo} -
		\expval{(\tilde{\ssp}_x^{\bullet})^2(nt)}_{\roo} \expval{(\tilde{\ssp}_y^{\bullet})^2(nt)}_{\roo} \right| = 
		\bigg| \sum_{k_1,k_2,k_3,k_4 \in I(\gamma)}  \varphi^{k_1}_x\varphi^{k_2}_x \varphi^{k_3}_y \varphi^{k_4}_y \times \\
	&\expval{\hat{\tilde{\ssp}}_{k_1}(nt) \hat{\tilde{\ssp}}_{k_2}(nt)
	\hat{\tilde{\ssp}}_{k_3}(nt) \hat{\tilde{\ssp}}_{k_4}(nt)}_{\roo} -
	\expval{\hat{\tilde{\ssp}}_{k_1}(nt) \hat{\tilde{\ssp}}_{k_2}(nt)}_{\roo}
	\expval{
	\hat{\tilde{\ssp}}_{k_3}(nt) \hat{\tilde{\ssp}}_{k_4}(nt)}_{\roo}\bigg|	.
	\end{split}
	\end{equation}
	Using the definition of  $\hat{\tilde{\ssp}}_{k}(nt)$, we obtain 16 different terms. 
	However, tanks to the symmetry of \eqref{decayassum}, and Lemma \ref{boundomega}, all
	these terms can be treated similarly. Therefore, we take the terms involving only 
	$\tilde{\ssp}$. Let us denote 
	\begin{equation} \label{Qaee}
				\tilde{\ssa}_z^x:= \sum_{k\in I(\gamma)} \cos(\omega_knt) \varphi^k_x 
				\varphi^k_z. 	
	\end{equation}
	Then by using the definition of $\hat{\tilde{\ssp}}_k$ we have:
	\begin{equation}\label{Qeproof6}
	\begin{split}
		& \tilde{\mathsf{C}}_{xy}:=\bigg| \sum_{k_1,k_2,k_3,k_4 \in I(\gamma)}  \varphi^{k_1}_x\varphi^{k_2}_x \varphi^{k_3}_y \varphi^{k_4}_y	 \cos(\omega_{k_1} nt)
		 \cos(\omega_{k_1} nt) \cos(\omega_{k_2} nt)  \cos(\omega_{k_3} nt)  \cos(\omega_{k_4} nt) \\& \expval{\hat{\tilde{\ssp}}_{k_1} \hat{\tilde{\ssp}}_{k_2}
	\hat{\tilde{\ssp}}_{k_3} \hat{\tilde{\ssp}}_{k_4}}_{\roo} -
	\expval{\hat{\tilde{\ssp}}_{k_1} \hat{\tilde{\ssp}}_{k_2}}_{\roo}
	\expval{
	\hat{\tilde{\ssp}}_{k_3} \hat{\tilde{\ssp}}_{k_4}}_{\roo} \bigg| = \\
	&\bigg|\sum_{z_1,z_2,z_3,z_4=1}^n \left(\expval{\frac{\tilde{\ssp}_{z_1}}{
	\sqrt{m_{z_1}}}\frac{\tilde{\ssp}_{z_2}}{
	\sqrt{m_{z_2}}}\frac{\tilde{\ssp}_{z_3}}{
	\sqrt{m_{z_3}}}\frac{\tilde{\ssp}_{z_4}}{
	\sqrt{m_{z_4}}}}_{\roo}-\expval{\frac{\tilde{\ssp}_{z_1}}{
	\sqrt{m_{z_1}}}\frac{\tilde{\ssp}_{z_2}}{
	\sqrt{m_{z_2}}}}_{\roo} \expval{\frac{\tilde{\ssp}_{z_3}}{
	\sqrt{m_{z_3}}}\frac{\tilde{\ssp}_{z_4}}{
	\sqrt{m_{z_4}}}}_{\roo}\right) \times \\
	& \tilde{\ssa}_{z_1}^x \tilde{\ssa}_{z_2}^x
	\tilde{\ssa}_{z_3}^x \tilde{\ssa}_{z_4}^x \bigg| 
	\leq \sum_{z_1,z_3} \frac{c}{|z_1-z_3|^2}
	|\tilde{\ssa}_{z_1}^x| |\tilde{\ssa}_{z_3}^y|\sum_{z_2,z_4}  \frac{c}{|z_2-z_4|^2}
	|\tilde{\ssa}_{z_2}^x| |\tilde{\ssa}_{z_4}^y| +
	\\ &\sum_{z_1,z_4} \frac{c}{|z_1-z_4|^2}
	|\tilde{\ssa}_{z_1}^x| |\tilde{\ssa}_{z_4}^y|\sum_{z_2,z_3}  \frac{c}{|z_2-z_3|^2}
	|\tilde{\ssa}_{z_2}^x| |\tilde{\ssa}_{z_3}^y| 
	\end{split}
	\end{equation}	 
	where we used the assumption \eqref{decayassum}. We can decompose each sum
	in the above expression as we did in \eqref{Qcov1}. In fact, comparing above sums,
	with  the expression in \eqref{Qcov1}, thanks to \eqref{Qeq:4}, 
	\eqref{Qeq:5} it is straightforward to observe that
	\begin{equation} \label{Qeproof7}
		\sum_{|x-y|>4n^{\theta'}} \tilde{\mathsf{C}}_{xy} \to 0,
	\end{equation}
	almost surely w.r.t the distribution of the masses, where one should recall 
	\eqref{Qeq:4}, and \eqref{Qcov2}. Combining the above symmetry argument, 
	\eqref{Qeproof5} 
	\eqref{Qeproof7}, and \eqref{Qeproof6}, we have:

	\begin{equation} \label{Qepproof44}
			\left|\frac{1}{n^2} \sum_{|x-y|>4 n^{\theta'}} 
		\frac{f(\frac{x}{n})f(\frac{y}{n})}{m_xm_y} 
		\left((\tilde{\ssp}_x^{\bullet})^2(nt) - 
		\expval{(\tilde{\ssp}_x^{\bullet})^2(nt)}_{\roo}\right)	
		\left((\tilde{\ssp}_y^{\bullet})^2(nt) - 
		\expval{(\tilde{\ssp}_y^{\bullet})^2(nt)}_{\roo}\right)\right|
		\to 0,
	\end{equation}
		almost surely w.r.t the distribution of the masses. Finally,  having
		\eqref{Qepproof44}, \eqref{Qepproof4} we conclude $\Pi^{\bullet}_n \to 0$ 
		almost surely w.r.t the distribution of the masses.\\
		 Now we prove the last 
		limit of \eqref{Qepproof3} (we abbreviate: $a(nt)b(nt)=(ab)(nt)$):
		\begin{equation}
		\begin{split} \label{Qeproof8}
			&\Pi_n^{o\bullet} \leq \frac{C}{n^2}\sum_{x,y=1}^{n} 
			\left|\expval{(\tilde{\ssp}_x^o\tilde{\ssp}_y^o
			\tilde{\ssp}_x^{\bullet} \tilde{\ssp}_y^{\bullet})(nt)}_{\roo}\right|+
			\frac{C}{n^2} \sum_{x,y=1}^n 
			\left|\expval{(\tilde{\ssp}_x^o\tilde{\ssp}_x^{\bullet})(nt)}_{\roo} \right|\left|
			\expval{
			(\tilde{\ssp}_y^o\tilde{\ssp}_y^{\bullet})(nt)}_{\roo} \right|\\	 \leq
			&\frac{C}{n^2}	\sum_{x,y=1}^n
			\expval{(\tilde{\ssp}_x^o\tilde{\ssp}_y^o)^2(nt)}_{\roo}^{\frac12} 
			\expval{(\tilde{\ssp}_x^{\bullet} \tilde{\ssp}_y^{\bullet})^2(nt)}_{\roo}^{\frac12} + \\
			& \frac{C}{n^2} \sum_{x,y=1}^n \expval{(\tilde{\ssp}_x^o)^2(nt)}_{\roo}^{\frac12}
		\expval{(\tilde{\ssp}_y^o)^2(nt)}_{\roo}^{\frac12} \expval{(\tilde{\ssp}_x^{\bullet})^2(nt)}_{\roo}^{\frac12}\expval{(\tilde{\ssp}_y^{\bullet})^2(nt)}_{\roo}^{\frac12}	\leq\\
				&\frac{c'}{n^2} \sum_{x,y=1}^n
				\expval{(\tilde{\ssp}_x^o)^2(nt) (\tilde{\ssp}_y^o)^2(nt)}_{\roo}^{\frac12}
		+\frac{c'}{n^2} \sum_{x,y=1}^n	\expval{(\tilde{\ssp}_x^o)^2(nt)}_{\roo}^{\frac12}
		\expval{(\tilde{\ssp}_y^o)^2(nt)}_{\roo}^{\frac12} 
		\leq  \\
		& c' \left(\frac{C}{n^2} \sum_{x,y=1}^n
				\frac{1}{m_xm_y}\expval{(\tilde{\ssp}_x^o)^2(nt) (\tilde{\ssp}_y^o)^2(nt)}_{\roo} \right)^{\frac12}  + 
				\\ &c'\left(\frac{C}{n^2} 
				\sum_{x,y=1}^n \frac{1}{m_xm_y}	\expval{(\tilde{\ssp}_x^o)^2(nt)}_{\roo}
		\expval{(\tilde{\ssp}_y^o)^2(nt)}_{\roo} \right)^{\frac12} \leq c'' \frac{n^{1-\gamma}}{n} \to 0,
		\end{split}
		\end{equation}
	as $n \to \infty$, where in the above computation we take advantage of the fact that
	above operators commute \eqref{CCRt}; in the first line we bounded $f/m_x$, in
	the second inequality we performed Cauchy-Schwartz \eqref{CSineq}, in the third 
	inequality we performed another Cauchy-Schwartz and used the bounds on "second
	and fourth moments" \eqref{Q4xbound}, \eqref{Qhighbound3}. In the fourth 
	inequality we take advantage of the Jensen inequality thanks to the fact that
	$f(z)=z^{\frac12}$ is concave. Finally, the last inequality is deduced from 
	\eqref{Qelowmodes}.  
	
	This finishes proof of \eqref{Qepproof3}, which yields \eqref{Qprproof2}. This means 
	$\mathsf{P}_n \to 0$ almost surely. As we mentioned the proof of the other limit
	in \eqref{Qerpproof1} is  similar. Hence we conclude \eqref{Qeproof}, and 
	this finishes the proof of \eqref{highermomentclassicaleQ} which complete the proof
	of Theorem \ref{mainthmQ}.
\end{proof}

%\begin{rem} \label{rempair}
%Although in our proof we used the "pairing" structure of  our assumption 
%\eqref{decayassum}, we can prove the same result for states without this pairing structure.
%Instead we can replace our assumption with the following stronger "clustering" decay: for proper $A ,B\subset \bI_n$, and $x \in A$, $y\in B$, 
%$\expval{O^{\sharp_1}_x O^{\sharp_2}_y}_{\roo} \leq C/|d(A,B)|^a$. In this case, we need
%to assume some bounds on $\expval{\hat{\tilde{\ssp}}_k^4}_{\roo}$, and  $\expval{\hat{\tilde{\ssp}}_x^4}_{\roo}$.  
%\end{rem}

\subsection{Locally Gibbs state}
In this section, we  me prove Theorem \eqref{thmGibbs}. Let us emphasize that this theorem 
can be deduced  from tools developed in Sect. 5  of \cite{A20}. Here we only mention the main critical
steps and refer reader to \cite{A20} for more details. Before proceeding, 
let us recall  basic notations, definitions, and identities related to $\Roo$ mostly
from \cite{A20}:

Recall the definition of $\Roo$ \eqref{locallyGibbs}. Thanks to  (3.41), (3.42) of
\cite{A20} up to an error that vanishes in the limit,% (cf. \cite{A20} sect.3, we do not bother ourselves with this vanishing constant here; since it can be treated easily), we have:
\begin{equation}
\tilde{\ssp}_x:= \ssp_x - \expval{\ssp_x}_{\Roo}= \ssp_x -\bar{p}(\frac{x}{n})
(\frac{m_x}{\bar{m}}), \quad \tilde{\ssr}_x:=\ssr_x-\expval{\ssr_x}_{\Roo}=
\ssr_x- \bar{r}(\frac{x}{n}).
\end{equation} 
Therefore, $\Roo$ can be written as $\Roo= \frac{1}{Z_n}\exp(-H_n^{\beta})$, 
with 
\begin{equation}
H_{\beta}^n=\frac12\sum_{x=1}^n
\left( \frac{\beta(\frac{x}{n})}{m_x}
 \tilde{\ssp}_x^2 + \beta(\frac{x}{n})\tilde{\ssr}_x^2\right). 
\end{equation}
Let us diagonalize $H_{\beta}^n$. We begin by a couple of definitions:
Let $M_\beta=M\tilde{\beta}^{-1}$, with $\tilde{\beta}:= \diag(\beta(\frac{1}{n}),\dots,\beta(\frac{n}{n}))$ and $\beta^o:=\diag(\beta(\frac{1}{n}),\dots,\beta(\frac{n-1}{n}))$.
Define $A_{p}^{\beta}$ and  $A_r^{\beta}$ as: 
	\begin{equation} \label{thermalmatrices}
		A_p^{\beta}=M_{\beta}^{-\frac12}(-\nabla_-\beta^o \nabla_+)M_{\beta}^{-\frac12}, \qquad A_r^{\beta}= (\beta^o)^{\frac12}(-\nabla_+ M_{\beta}^{-1} \nabla_-)(\beta^o)^{\frac12}.
	\end{equation}
 $A_p^{\beta}$ is symmetric positive semidefinite with almost sure non-degenerate spectrum.
 Denote $\{ \psi^k \}_{k=0}^{n-1}$ to be its orthonormal set of eigenvectors, forming a basis for $\mathbb{R}^n$. Denote the corresponding set of increasing eigenvalues by $\{0=\gamma_0^2<\gamma_1^2<\dots<\gamma_{n-1}^2\}$. Notice $\langle \psi^k, \psi ^{k'}\rangle=\sum_{x=1}^n \psi^k_x \psi^{k'}_x=\delta_{k,k'}$, and $\sum_{k=0}^{n-1}\psi^k_x\psi^{k'}_y=\delta_{x,y}$. $A_r^{\beta}$ is symmetric positive definite. Define for $k \in \mathbb{I}_{n-1}$, 
$$ \tilde{\psi}^k:= \frac{1}{\gamma_k} (\beta^o)^{\frac12} \nabla_+M_{\beta}^{-\frac12} \psi^k,$$  
observe that $\{ \tilde{\psi}_k\}_{k=1}^{n-1}$  is the set of eigenvectors of $A_r^{\beta}$ with similar eigenvalues $\gamma_1^2<\dots<\gamma_{n-1}^2$. 
We define another set  of coordinates: for $k \in\bI_n$ (convention: $\rmr_n=0$):
\begin{equation} \label{thermalcoordinate}
\rmp_k:=\langle \psi^k,M_{\beta}^{-\frac12}\tilde{\ssp} \rangle_n,
\quad \rmr_k:=\langle \tilde{\psi}^k,\tilde{\ssr} \rangle_{n-1}, 
\end{equation}
Correspondingly, we define the following set of  bosonic operators for $k \in \bI_{n-1}$:
\begin{equation} \label{thermalbosonic}
\rmb_k := \frac{1}{2\gamma_k}\left(\rmr_k+i\rmp_k \right), \quad
\rmb_k^* := \frac{1}{2\gamma_k}\left(\rmr_k-i\rmp_k \right).
\end{equation}
 From \eqref{thermalcoordinate}, and \eqref{thermalbosonic}, by using properties of 
 $\psi^k$, $\tilde{\psi}^k$, we have upto a vanishing error: 
 \begin{equation} \label{Hthermal}
 H_{\beta}^n= \frac12\sum_{k=1}^{n-1}(\rmp_k^2+ \rmr_k^2)= \sum_{k=1}^{n-1}
 \gamma_k (\rmb_k^*\rmb_k+\frac12),
 \end{equation}
where we have following commutation relations  for 
$k,k' \in \bI_n$ thanks to \eqref{CCR} and above 
definitions:
\begin{equation} \label{CCRthermal}
\begin{split}
[\rmp_k,\rmp_{k'}]=[\rmr_k,\rmr_{k'}]=0,\quad [\rmr_k,\rmp_{k'}]=i\gamma_k\delta_{k,k'}, \\
[\rmb_k,\rmb_{k'}]=[\rmb_k^*,\rmb_{k'}^*]=0, \quad [\rmb_k^*,\rmb_{k'}]=\delta_{k,k'}.
\end{split}
\end{equation}
The detailed computation of above identities can be found in sec.~3 of \cite{A20}, in particular (3.30)-(3.36).

Having \eqref{Hthermal} and \eqref{CCRthermal} the discrete spectrum of $H_{\beta}^n$
is understood. Moreover, thanks to spectral theorem we can compute average of certain 
observables w.r.t $\Roo$. In particular we have (cf. \cite{A20} Sect.3 
(3.38), (3.48) and Appendix C): for $k,k' \in \bI_{n-1}$, and $x,y \in \bI_n$:
\begin{equation} \label{thermalavgs}
\begin{split}
&\expval{\rmb_k}_{\Roo}= \expval{\rmb_k^*}_{\Roo} =0, \quad \expval{\rmb_k^* \rmb_k}_{\Roo} 
=\frac{\delta_{k,k'}}{e^{\gamma_k}-1}, \implies \\
&\expval{\rmp_k}_{\Roo}=\expval{\rmr_k}_{\Roo}=0, \quad 
\expval{\rmp_k \rmp_{k'}}_{\Roo}=\expval{\rmr_k \rmr_{k'}}_{\Roo}=\delta_{k,k'}\frac{\gamma_k}{2}
\coth(\frac{\gamma_k}{2}).
\end{split}
\end{equation}
Finally, thanks to ladder structure of bosonic operators, by using spectral theorem we can 
observe that $\rmb_k$ has pairing structure in $\Roo$. Denote $\rmb_k^1:= \rmb_k$, 
$\rmb_k^2:= \rmb_k^*$. Then we have for any $k_1,k_2,k_3,k_4$, and any $\sharp_1, \sharp_2
, \sharp_3, \sharp_4 \in\{ 1,2\}$:
\begin{equation}\label{Qpairing}
\begin{split}
&\expval{\rmb_{k_1}^{\sharp_1}\rmb_{k_2}^{\sharp_2}\rmb_{k_3}^{\sharp_3}}_{\roo} =0,\\
&\expval{\rmb_{k_1}^{\sharp_1}\rmb_{k_2}^{\sharp_2}\rmb_{k_3}^{\sharp_3}
\rmb_{k_4}^{\sharp_4}}_{\roo}= \expval{\rmb_{k_1}^{\sharp_1}\rmb_{k_3}^{\sharp_3}}_{\Roo}\expval{\rmb_{k_2}^{\sharp_2}\rmb_{k_4}^{\sharp_4}}_{\Roo}+ \expval{\rmb_{k_1}^{\sharp_1}\rmb_{k_2}^{\sharp_2}}_{\Roo}\expval{\rmb_{k_3}^{\sharp_3}\rmb_{k_4}^{\sharp_4}}_{\Roo}
+ \\ &\expval{\rmb_{k_1}^{\sharp_1}\rmb_{k_4}^{\sharp_4}}_{\Roo}\expval{\rmb_{k_2}^{\sharp_2}\rmb_{k_3}^{\sharp_3}}_{\Roo}.
\end{split}
\end{equation}
Now we state the proof of  Theorem \ref{thmGibbs}.

\begin{proof}[Proof of theorem \ref{thmGibbs}] 
First notice that thanks to the pairing structure \eqref{Qpairing}, the pairing structure 
remains true for $\tilde{\ssp}_x(nt)$, $\tilde{\ssr}_x(nt)$: similar to \eqref{pairing}
with proper modification (since they are linear combination of bosonic operators, it is a 
matter of lengthy but straightforward computation to
see this pairing structure persists). Therefore, $\Roo$ satisfies the third bound in 
\eqref{decayassum}, since odd moments are zero. Moreover, thanks to the pairing structure,
it is sufficient to prove the first bound in \eqref{decayassum}. Finally,
we prove the first bound for $\tilde{\ssp}_x$, $\tilde{\ssp}_y$. Adapting the proof to the case of 
$\tilde{\ssr}_x$, $\tilde{\ssr}_y$ and cross terms is straightforward. 

Inverse of \eqref{thermalcoordinate} reads ($\beta_x:= \beta(\frac{x}{n})$):
\begin{equation} \label{thermalinverse}
\tilde{\ssp}_x= \sqrt{\frac{m_x}{\beta_x}}\sum_{k=1}^{n-1} \psi^k_x \rmp_k. 
\end{equation}  
Hence,by using \eqref{thermalavgs} we have: 
\begin{equation} \label{Qcovxy}
\expval{\tilde{\ssp}_x \tilde{\ssp}_y}_{\Roo}= \sqrt{\frac{m_xm_y}{\beta_x\beta_y}}
\sum_{k,k'=1}^{n-1} \psi^k_x\psi^{k'}_y \expval{\rmp_k\rmp_{k'}}_{\Roo}= 
\sqrt{\frac{m_xm_y}{\beta_x\beta_y}} \sum_{k=1}^{n-1} \frac{\gamma_k}{2} 
\coth(\frac{\gamma_k}{2}) \psi^k_x\psi^k_y.
\end{equation}
Recall the matrix $A_p^{\beta}$, and the fact  that $\psi^k$ is their eigenvectors.
Denoting 
\begin{equation} \label{xcothx}		
			\mathfrak{f}(z)=
			\begin{cases}
				z^{\frac12}\coth(z^{\frac12}), \quad z \neq  0, \\
				1, \quad z=0.
			\end{cases}
	\end{equation}
	by above expression \eqref{Qcovxy} it is clear that we have (up to a
	vanishing error  as $n \to \infty$):
\begin{equation}  \label{Qcovxy2}
\expval{\tilde{\ssp}_x \tilde{\ssp}_y}_{\Roo}= \langle x, \mathfrak{f}(A_{p}^{\beta}) y
\rangle_n,
\end{equation}
where $\ket{x}= (0,0,\dots,1,0,\dots,0)\in \bR^n$ with $1$ located at $x$th position.

The poles of the function $z^{\frac12}\coth(z^{\frac12})$ is the following set: $\{z \in 
\mathbb{C} | z= -k^2\pi^2, k \in \mathbb{Z} \}$, and this function is analytic on the rest 
of the complex plane. The point zero is a removable pole, and  by redefining the function 
at zero, we can remove this pole: $\coth(z)$ has the following Taylor series expression for 
$0<|z|<\pi$: $\coth(z)=z^{-1}+\sum_{n=1}^{\infty}a_n z^{2n-1}$, where $a_n=\frac{2^{2n}
B_{2n}}{(2n)!}$, and $B_{2n}$ 
are Bernoulli numbers. Hence, $z\coth(z)=1+\sum_{n=1}^{\infty} a_nz^{2n}$, and   
$z^{\frac12}\coth(z^{\frac12})$ is given by the following Taylor series: $1+\sum_{n=1}
^{\infty}a_nz^n$ for $0<|z|<\pi^2$.  Consequently, the pole of $\mathfrak{f}$ is given by 
the set $\{z \in \mathbb{C} | z= -k^2\pi^2, k \in \mathbb{N}, k>0 \}.$ 
		Notice that by definition of $A_p^{\beta}$, thanks to properties of masses and $
		\beta$ there is a constant $c_0>0$, uniform in $n$\footnote{This constant can be 
		taken equal to $4\frac{\beta_{max}^2}{m_{min}}$.}, such that for any configuration 
		of the masses, we have $||A_p^{\beta}||_2 <c_0$ ($||\cdot||_2$ denotes the usual 
		operator norm of the matrix). Define $\alpha:=\frac12(c_0+1)$, let $\mathcal{R}:=
		\alpha+\pi^2$, by the above argument $\mathfrak{f}(z)$ is analytic in the open disk 
		$|z-\alpha|<\mathcal{R}$, and $\mathcal{R}$ is the radius of convergence for the 
		Taylor expansion of $\mathfrak{f}$,  $\mathfrak{f}(z)=\sum_{k=0}^{\infty}a_k(z-
		\alpha)^k$. Moreover, by the choice of $\alpha$ and $c_0$, the eigenvalues of 
		$A_p^{\beta}$ and $A_r^{\beta}$ lies in the disk $|z-\alpha|< \mathcal{R}$. So we 
		have
		 the following Taylor expansions for $\mathfrak{f}(A_p^{\beta})$ (cf.
		 \cite{A20} sect. 5 for more details):
		 	\begin{equation}\label{Taylorexpansionmatrix}
			\mathfrak{f}(A_p^{\beta})=\sum_{k=0}^{\infty} a_k (A_p^{\beta}-\alpha I_n)^k. 
	 			\end{equation}	
				
	By a random walk representation argument(cf. (5.11) of \cite{A20}),  since $A_p^{\beta}$ is tridiagonal 
	one can observe that for any $k< |x-y|-1$, we have:
	\begin{equation} \label{Qlowtaylor}
	\langle x, (A_p^{\beta})^k y \rangle_n=0.
	\end{equation}
	Plugging the Taylor series \eqref{Taylorexpansionmatrix} inside the expression
	\eqref{Qcovxy2} and using \eqref{Qlowtaylor} we have:
	\begin{equation} \label{Qboundtaylor}
	|\expval{\tilde{\ssp}_x \tilde{\ssp}_y}_{\Roo}| \leq ||\sum_{k>|x-y|-2} a_k (A_p^{\beta}-\alpha I_n)^k||
	_2  \leq \sum_{k>|x-y|-2} |a_k|||(A_p^{\beta}-\alpha I_n)||_2^k \leq C
	\mathfrak{q}^{|x-y|},
	\end{equation}
for some $0<\mathfrak{q}<1$, where we used the fact that $||(A_p^{\beta}-\alpha I_n)||_2$ 
is inside the radius of convergence of the Taylor series,  and we take advantage of properties of Taylor series (cf. (5.25) of \cite{A20}).
The latter estimate gives us the first bound in \eqref{decayassum}. In fact, it provides
much stronger exponential decay. This finishes proof of Theorem \ref{thmGibbs}. 
\end{proof}
\medskip
\section*{Acknowledgement}
This work was partially supported by ANR-15-CE40-0020-01 LSD of the French National Research Agency. A.H. would like to thank prof. Gigliola Staffilani and M.I.T math department for their hospitality during this project. 

\bibliographystyle{plain}
\bibliography{bibilo2}
 
\end{document}